\documentclass[aps,prx,onecolumn,superscriptaddress,nofootinbib,longbibliography]{revtex4-2}
\usepackage[a4paper, left=1in, right=1in, top=0.95in, bottom=0.95in]{geometry}
\usepackage{fix-cm} % Allows arbitrary font sizes

\usepackage{cmap} % Load before fontenc 
\usepackage[utf8]{inputenc}
\usepackage[english]{babel}
\usepackage[T1]{fontenc}
\usepackage{amsmath}
\usepackage{amsfonts}
\usepackage{mathrsfs}
\usepackage{bm}
\usepackage{xcolor}
\definecolor{blueviolet}{rgb}{0.2, 0.2, 0.6}
\definecolor{webgreen}{rgb}{0,.5,0}
\definecolor{webbrown}{rgb}{.6,0,0}

\allowdisplaybreaks
\usepackage{tikz}
\usepackage{braket}
\usepackage{natbib}
\usepackage{amsthm}
\usepackage{dsfont}
\usepackage{listings}
\usepackage{refcount}

\usepackage[ruled,vlined,linesnumbered]{algorithm2e}

\usepackage[T1]{fontenc}
\usepackage{lmodern}
\usepackage{fix-cm}

\usepackage[pdftex,
	bookmarks=false,
	colorlinks=true, %allcolors=blueviolet,
	urlcolor=webbrown, 
	linkcolor=blueviolet, 
	citecolor=webgreen,
	pdfstartpage=1,
	pdfstartview={FitH},  % FitBH
	bookmarksopen=false
	]{hyperref}

\numberwithin{equation}{section}

\newtheorem{theorem}{Theorem}

\newtheorem{definition}{Definition}

\newtheorem{lemma}{Lemma}

\newtheorem{fact}{Fact}
\newtheorem{observation}{Observation}

\theoremstyle{definition}

\providecommand{\myvec}[1]{\ensuremath{\boldsymbol{#1}}}

\providecommand{\RR}{\ensuremath{\myvec{R}}}

\providecommand{\ss}{\ensuremath{\myvec{s}}}

\providecommand{\RR}{\ensuremath{\myvec{R}}}

\providecommand{\ss}{\ensuremath{\myvec{s}}}

\def\01{\{0,1\}}

\newcommand{\C}{\ensuremath{\mathcal{C}}}
\newcommand{\CA}{\mathcal{A}}

\newcommand{\CE}{\mathcal{E}}

\newcommand{\CF}{\mathcal{F}}

\newcommand{\CJ}{\mathcal{J}}
\newcommand{\CK}{\mathcal{K}}

\newcommand{\CN}{\mathcal{N}}
\newcommand{\BN}{\mathbb{N}}
\newcommand{\CO}{\mathcal{O}}

\newcommand{\CR}{\mathcal{R}}

\newcommand{\CS}{\mathcal{S}}
\newcommand{\CT}{\mathcal{T}}
\newcommand{\CU}{\mathcal{U}}

\newcommand{\vertiii}[1]{{\left\vert\kern-0.25ex\left\vert\kern-0.25ex\left\vert #1 \right\vert\kern-0.25ex\right\vert\kern-0.25ex\right\vert}}

\newcommand{\rom}[1]{\mathtt{\uppercase\expandafter{\romannumeral #1\relax}}}

\DeclareMathOperator*{\argmax}{arg\,max}

\DeclareMathOperator{\spec}{Spec}

\DeclareMathOperator{\Tr}{tr}

\usepackage{mdframed}
\usepackage{xcolor}

\definecolor{puzzlebackground}{RGB}{247, 250, 253}    % Light ocean blue background
\definecolor{puzzleborder}{RGB}{70, 142, 175}   % Deep ocean blue border
\definecolor{openQbackground}{RGB}{255, 248, 244}  % Warm white background
\definecolor{openQborder}{RGB}{188, 140, 112} % Warm brown border

\newcounter{puzzlecount}

\newcounter{openQcount}

\newcommand{\ketbra}[2]{\lvert #1 \rangle \! \langle #2 \rvert}

\newcommand{\norm}[1]{\left\lVert#1\right\rVert}

\usepackage{refcount}
\newcommand\pspace{\mathsf{PSPACE}}

\newcommand\negl[1]{\mathsf{negl}(#1)}
\newcommand\hilbert{\mathcal{H}}%hilbert space
\newcommand\pr{\operatorname{Pr}}%probablity
\newcommand\alg{\mathcal{A}}%algorithm
\newcommand\poly[1]{\mathsf{poly}(#1)}
\newcommand\tm{\mathsf{M}}%turing machine
\newcommand\utm{\mathsf{UM}}
\newcommand\D{\mathsf{D}}
\newcommand\acc{\textsf{Accept}}
\newcommand\rej{\textsf{Reject}}
\newcommand\qpe{\operatorname{QPE}}
\newcommand\qft{\operatorname{QFT}}
\newcommand\sfA{\mathsf{A}}
\newcommand\sfB{\mathsf{B}}
\newcommand\sfS{\mathsf{S}}

\begin{document}
\fontsize{10}{12}\selectfont

\title{Random unitaries that conserve energy}

\author{Liang Mao}
\affiliation{California Institute of Technology, Pasadena, California 91125, USA}
\affiliation{Institute for Advanced Study, Tsinghua University, Beijing, 100084, China}

\author{Laura Cui}
\affiliation{California Institute of Technology, Pasadena, California 91125, USA}

\author{Thomas Schuster}
\affiliation{California Institute of Technology, Pasadena, California 91125, USA}
\affiliation{Google Quantum AI, Venice, California 90291, USA}

% \author{Fernando Brand\~ao}
% \affiliation{AWS Center for Quantum Computing, Pasadena, California 91125, USA}
% \affiliation{California Institute of Technology, Pasadena, California 91125, USA}

\author{Hsin-Yuan Huang}
\affiliation{California Institute of Technology, Pasadena, California 91125, USA}
\affiliation{Google Quantum AI, Venice, California 90291, USA}

\date{\today}

\begin{abstract}
Random unitaries sampled from the Haar measure serve as fundamental models for generic quantum many-body dynamics. Under standard cryptographic assumptions, recent works have constructed polynomial-size quantum circuits that are computationally indistinguishable from Haar-random unitaries, establishing the concept of pseudorandom unitaries (PRUs). While PRUs have found broad implications in many-body physics, they fail to capture the energy conservation that governs physical systems. In this work, we investigate the computational complexity of generating PRUs that conserve energy under a fixed and known Hamiltonian $H$. 
We provide an efficient construction of energy-conserving PRUs when $H$ is local and commuting with random coefficients. 
%We demonstrate that, under standard cryptographic assumptions, there exists an efficient construction of energy-conserving PRUs when $H$ is a local commuting Hamiltonian. 
Conversely, we prove that for certain translationally invariant one-dimensional $H$, there exists an efficient quantum algorithm that can distinguish truly random energy-conserving unitaries from any polynomial-size quantum circuit. This establishes that energy-conserving PRUs cannot exist for these Hamiltonians. Furthermore, we prove that determining whether energy-conserving PRUs exist for a given family of one-dimensional local Hamiltonians is an undecidable problem. Our results reveal an unexpected computational barrier that fundamentally separates the generation of generic random unitaries from those obeying the basic physical constraint of energy conservation.
\end{abstract}

\maketitle

\section{Introduction}

Haar-random unitaries, which are unitaries drawn uniformly from the unitary group, provide a powerful theoretical model for generic quantum dynamics in complex quantum many-body systems. These ensembles capture universal signatures of quantum chaos and thermalization~\cite{deutsch1991quantum,fisher2023random,srednicki1994chaos,rigol2008thermalization,nahum2017entgrowth,sekino2008fast,cotler2022fluctuations,hayden2007black}, and have found wide-ranging applications across quantum science, from quantum device benchmarking and tomography~\cite{emerson2005scalable,knill2008randomized,elben2023randomized,guta2020fast,huang2020predicting}, to quantum machine learning~\cite{mcclean2018barren,huang2023learning,larocca2025barren}, to black hole physics and holography~\cite{sekino2008fast,hayden2007black,brown2023quantum,nezami2023quantum}. Despite their central role, Haar-random unitaries are computationally intractable: specifying or implementing an $n$-qubit Haar-random unitary requires exponential resources~\cite{knill1995approximation}, rendering them physically unrealistic.

Pseudorandom unitaries (PRUs) address this issue by offering a practical alternative~\cite{ji2018pseudorandom,metger2024simple,chen2024efficient,ma2024construct}. PRUs are ensembles of random unitaries that can be efficiently generated, yet are indistinguishable from Haar-random unitaries $U$ in any polynomial-time quantum experiment with oracle access to $U$. %They therefore provide an efficient method to simulate chaotic Haar-random dynamics.
Their existence hence provides crucial evidence for the use of Haar-random unitaries as models of chaotic polynomial-time quantum circuits in the real world.
Recent constructions achieve this indistinguishability even at logarithmic circuit depths~\cite{schuster2024random,schuster2025strong,cui2025unitary,foxman2025random}, leading to a broad range of physical implications: the hardness of recognizing quantum phases of matter~\cite{feng2025hardness,schuster2025hardness}, the surprising efficiency of scrambling information~\cite{gu2024simulating,schuster2025strong}, and the existence of large families of indistinguishable states with completely different entanglement structures~\cite{schuster2024random,ma2024construct,akers2024holographic}. Despite these successes, PRUs suffer from a critical limitation: they fail to respect energy conservation, a fundamental and universal constraint governing  physical systems.

%In any system governed by a Hamiltonian $H$, the allowed evolution operator $U$ must satisfy the commutation relation $[U, H] = 0$ to ensure that $U$ preserves the energy spectrum. 
A unitary evolution $U$ conserves energy under a given Hamiltonian $H$ if it satisfies the commutation relation, $[U, H] = 0$.
%To respect energy conservation given by a Hamiltonian $H$, the allowed evolution operator $U$ must satisfy the commutation relation $[U, H] = 0$.
This energy constraint distinguishes physical dynamics from unconstrained Haar-random unitaries, and dramatically alters dynamical  phenomena.
For examples, energy conservation is responsible for quantum thermalization to finite-temperature Gibbs states instead of maximally-mixed states~\cite{deutsch1991quantum,srednicki1994chaos,rigol2008thermalization,polkovnikov2011colloquium,nandkishore2015many,gogolin2016equilibration},
the emergence of hydrodynamics at finite energy densities and at late times~\cite{forster2018hydrodynamic,kovtun2012lectures,sirker2009diffusion,von2018operator,ljubotina2019kardar,zu2021emergent,wei2022quantumgas},
non-ergodic behavior of some exotic models~\cite{pal2010many,choi2016exploring,turner2018weak}, and a high computational complexity for determining quantum equilibration~\cite{moore1990unpredictability,shiraishi2021undecidability,devulapalli2025complexity}.
% This constraint distinguishes physical dynamics from unconstrained Haar-random unitaries.
% The energy constraint dramatically alters the dynamics, leading to distinctive physical phenomena such as quantum thermalization to finite-temperature Gibbs states~\cite{deutsch1991quantum,srednicki1994chaos,rigol2008thermalization,polkovnikov2011colloquium,nandkishore2015many,gogolin2016equilibration},
% non-ergodic behavior of some exotic models~\cite{pal2010many,choi2016exploring,turner2018weak},
% and the emergence of hydrodynamics at finite energy densities and late times~\cite{forster2018hydrodynamic,kovtun2012lectures,sirker2009diffusion,von2018operator,ljubotina2019kardar,zu2021emergent,wei2022quantumgas}. It also results in the high  computational complexity for determining quantum equilibration~\cite{moore1990unpredictability,shiraishi2021undecidability,devulapalli2025complexity}. 
Any realistic model of physical dynamics must therefore incorporate energy conservation. This motivates the following question:
\begin{center}
    \emph{Can we construct pseudorandom unitaries that respect energy conservation?}
\end{center}
Our work provides a comprehensive answer to this question by establishing deep connections between the existence of energy-conserving PRUs and computational complexity theory.
Our results reveal that the existence of energy-conserving PRUs depends critically on the specific structure of the Hamiltonian $H$, and that even determining their existence is, in general,  undecidable.
%Our work provides a comprehensive answer to this question by revealing that the existence of energy-conserving PRUs depends critically on the specific structure of the Hamiltonian $H$, and that determining their existence is, in general, undecidable.

To establish our results, we first formulate this question in a precise manner.
We observe that for any local Hamiltonian $H \neq I$, any energy-conserving unitary can be efficiently distinguished from a Haar-random unitary (over the entire unitary group) by measuring the energy of $U\ket{\psi}$ for any non-zero energy state $\ket{\psi}$. 
%We begin by observing that for any local Hamiltonian $H \neq I$, any ensemble of energy-conserving unitaries can be efficiently distinguished from Haar-random unitaries by checking the energy of $U\ket{\psi}$ for a low energy state $\ket{\psi}$. 
This motivates a refined definition: \emph{energy-conserving PRUs} should be computationally indistinguishable not from Haar-random unitaries, but from \emph{energy-conserving Haar-random unitaries}, i.e. unitaries drawn according to the Haar measure of the group containing all unitaries that commute with the given Hamiltonian $H$, $\{U \, : \, [U, H] = 0\}$.

Our main results establish a striking dichotomy. 
For a simple class of random commuting Hamiltonians, we prove that energy-conserving PRUs exist by providing an efficient construction.
%For commuting random Hamiltonians of the form $H = \sum_i \mathcal{J}_i h_i$ with Gaussian random coefficients $\mathcal{J}_i$, we provide an efficient construction of energy-conserving PRUs.
On the other hand, more surprisingly, we construct explicit families of one-dimensional, local, translationally invariant Hamiltonians for which energy-conserving PRUs provably cannot exist. Furthermore, we establish this non-existence  by proving an even stronger statement: there is an efficient quantum algorithm that can distinguish energy-conserving Haar-random unitaries for these Hamiltonians from any polynomial-size quantum circuit. Both of these results hold under standard cryptographic and complexity-theoretic conjectures, such as the existence of quantum-secure one-way functions. 
Given these contrasting results, a natural question is how one can determine whether energy-conserving PRUs exist for a given Hamiltonian family. 
%Given that there exist Hamiltonians with and without energy-conserving PRUs, a natural question is whether we can determine if energy-conserving PRUs exist for a given Hamiltonian family. 
Unfortunately, using standard complexity-theoretic tools, we prove that this problem is in general undecidable: no algorithm can solve it, even given exponential time and space resources. Our results reveal fundamental computational barriers emerging from physical energy constraints, highlighting the tension between common models of ergodicity and the physical requirement of energy conservation. The results are summarized in Fig~\ref{fig:main}.

\begin{figure}
    \centering
    \includegraphics[width=0.99\linewidth]{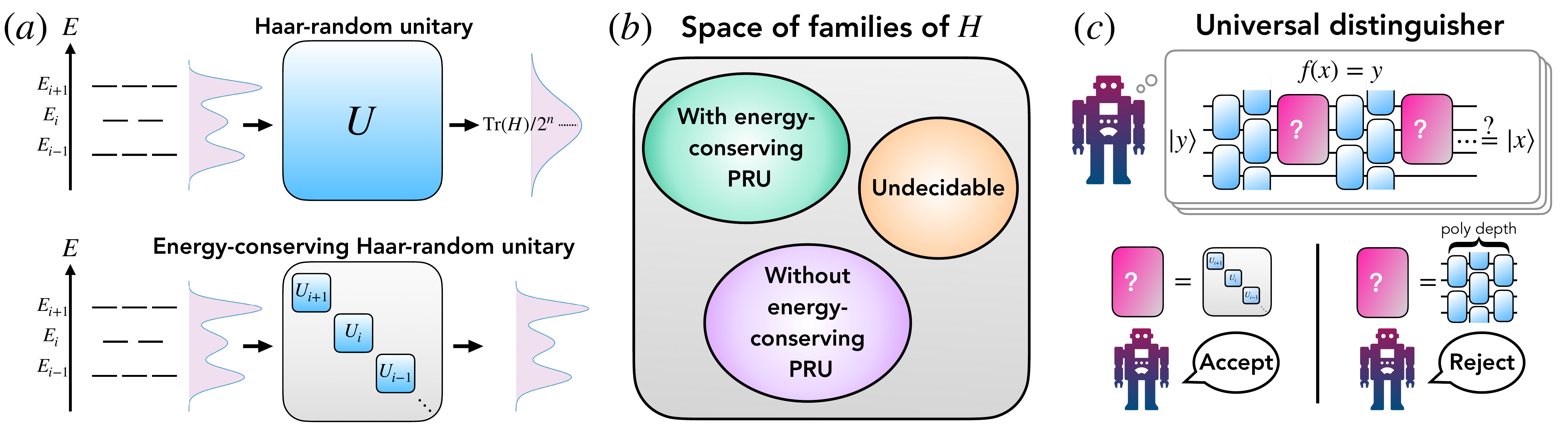}
    \caption{(a) Haar random unitary and energy-conserving Haar random unitary. Haar random unitaries are those that scramble the full Hilbert space. They transit each wavefunction to the infinite-temperature state. In contrast, energy-conserving Haar random unitaries are those that only scramble degenerated subspaces. They respect the energy-occupation of each wavefunction. (b) Summary of our results. We construct local Hamiltonian families with and without energy-conserving PRU. Further, we demonstrate there exists a certian set of Hamiltonian families such that determining if they have energy-conserving PRU is an undecidable problem. (c) We construct a universal distinguishing algorithm to prove Theorem~\ref{main:thm:hard-ham}. The algorithm accepts the energy-conserving Haar random unitaries of the hard Hamiltonian, whereas rejects any polynomial-size quantum circuit.}
    \label{fig:main}
\end{figure}

\section{Pseudorandom unitaries and energy conservation}

We recall the standard definition for pseudorandom unitaries~\cite{ji2018pseudorandom,ma2024construct}:

\begin{definition}[Pseudorandom unitaries; informal]
\label{def:pru}
A sequence $\{\mathcal{U}_n\}_{n\in\mathbb{N}}$ of ensembles of $n$-qubit random unitaries is a pseudorandom unitary (PRU) if:
\begin{enumerate}
    \item {Efficiency:} Every unitary in $\mathcal{U}_n$ can be implemented by a quantum circuit of size $\mathsf{poly}(n)$.
    \item {Pseudorandomness:} For any polynomial-time quantum algorithm $\CA$ that receives oracle access to a unitary $U$ sampled either from $\mathcal{U}_n$ or from the Haar measure, the distinguishing advantage
    \begin{equation}
    \left| \Pr_{U \sim \mathcal{U}_n}[\CA^U(1^n) = 1] - \Pr_{U \sim \operatorname{Haar}}[\CA^U(1^n) = 1] \right| \leq \mathsf{negl}(n),
    \end{equation}
    where $\mathsf{negl}(n)$ is a function smaller than any $\frac{1}{\mathsf{poly}(n)}$.
\end{enumerate}
\end{definition}

\noindent The existence of PRUs has been proven under the existence of quantum-secure one-way functions~\cite{ma2024construct}. Quantum-secure one-way functions are families of efficiently computable functions $f_n$ that map $x\in\{0,1\}^n$ to $y\in\{0,1\}^m$~\cite{peikert2016decade,bernstein2017post}. The key feature is that given $y = f_n(x)$, no efficient quantum algorithm can find a preimage $x$, hence their name \emph{one-way}. These functions can be constructed explicitly using Learning With Errors (LWE)~\cite{regev2009lattices,albrecht2015concrete}. For a comprehensive overview, we refer readers to Appendix~\ref{ap:preliminary}.

To incorporate energy conservation into random unitaries, we first observe that:

\begin{observation}[Energy-based distinguisher] \label{main:obs:energy-distinguisher}
Any random unitary ensemble $\CU_H$ that conserves energy for a local Hamiltonian $H$ can be efficiently distinguished from a Haar-random unitary.
\end{observation}

\noindent To establish the above observation, consider the following simple distinguishing algorithm: prepare a product state $\ket{\psi}$ with an energy $\bra{\psi} H \ket{\psi}$ sufficiently different from $\mathrm{Tr}(H)/2^n$ (the infinite temperature value), apply the unknown unitary $U$, and measure the energy $\bra{\psi} U^\dagger H U \ket{\psi}$. For Haar-random unitaries, this expectation value concentrates around $\mathrm{Tr}(H)/2^n$, while for energy-conserving unitaries it equals $\bra{\psi} H \ket{\psi}$. Existing results~\cite{huang2022learning, anshu2021improved, brandao2013product} guarantee that product states with significant energy deviation from the infinite temperature value can be efficiently prepared for any local Hamiltonian, hence the distinguishing algorithm is efficient. See Theorem~\ref{thm:distinguish-energy-conserving-form-RU} in Appendix~\ref{ap:distinguish-energy-conserving-form-RU} for a detailed description and proof.

This observation motivates our proposed definition of energy-conserving PRUs:

\begin{definition}[Energy-Conserving PRUs; informal]
\label{def:ec_pru}
Let $n$ denote system size, let $ \{H_n\}_{n\in\mathbb{N}}$ be a sequence of $n$-qubit local Hamiltonians, and let $\{ \mathcal{C}^H_n \}_{n\in \mathbb{N}}$ denote the ensemble of energy-conserving Haar-random unitaries, i.e., Haar measure over the group $\{U_n : [U_n, H_n] = 0\}$ for each $n$. A sequence $\{\mathcal{U}_n\}_{n\in\mathbb{N}}$ of ensembles of $n$-qubit random unitaries is an energy-conserving PRU with respect to $\{H_n\}$ if:
\begin{enumerate}
    \item {Efficiency:} Every unitary in $\mathcal{U}_n$ can be implemented by a quantum circuit of size $\mathsf{poly}(n)$.
    \item {Pseudorandomness:} For any polynomial-time quantum algorithm $\CA$ that receives oracle access to a unitary $U$ sampled either from $\mathcal{U}_n$ or from $\mathcal{C}_n^H$, the distinguishing advantage
    \begin{equation}
    \left| \Pr_{U \sim \mathcal{U}_n}[\CA^U(1^n) = 1] - \Pr_{U \sim \mathcal{C}^H_n}[\CA^U(1^n) = 1] \right| \leq \mathsf{negl}(n),
    \end{equation}
    where $\mathsf{negl}(n)$ is a function smaller than any $\frac{1}{\poly{n}}$.
\end{enumerate}
\end{definition}

\noindent With this definition, we can formally ask the question of whether energy-conserving PRUs exist for a family of local Hamiltonians. Note that because pseudorandom objects are, by definition, asymptotic with respect to a parameter $n$ (we require that no polynomial-time algorithm exists, which is itself an asymptotic statement), we can only discuss the existence or absence of energy-conserving PRUs for a \emph{sequence} of local Hamiltonians $H_n$ for varying system sizes $n$.

\section{Main results}

We first provide a simple and efficient construction of energy-conserving PRUs, for random commuting Hamiltonians of the form $H=\sum_i \CJ_i h_i$, where $\CJ_i$ are Gaussian random coefficients, and $\{h_i\}$ form a complete set of commuting local observables. Here, a complete set means all common eigenstates of $\{h_i\}$ can be determined by specifying the local eigenstates of each $h_i$. This additional constraint is introduced to prevent insufficient covering, i.e., to ensure that no qubits exist that do not support any $h_i$.

\begin{theorem}[Constructing energy-conserving PRUs for random commuting Hamiltonians]
\label{main:thm:construct-pru-commuting}
Let $H=\sum_i\CJ_ih_i$ be an $n$-qubit commuting Hamiltonian with Gaussian random coefficients $\CJ_i$, where $\{h_i\}$ forms a complete set of commuting observables. There exists an efficient ensemble of unitaries that forms an energy-conserving PRU of $H$ with probability at least $1-\mathrm{negl}(n)$. 
\end{theorem}

\noindent The concrete constructions are described in Section~\ref{main:easy-ham}, and a detailed statement and proof of the theorem is given in Theorem~\ref{thm:construct-pru-commuting} in Appendix~\ref{ap:pru-commuting-ham}.

Having shown that energy-conserving PRUs exist for random commuting Hamiltonians, we now construct a different family of Hamiltonians for which energy-conserving PRUs provably do \emph{not} exist.
The Hamiltonians we consider are one-dimensional, local, and translationally invariant.
We establish the non-existence of their energy-conserving PRUs by providing an efficient quantum algorithm to distinguish an energy-conserving Haar-random unitary (under this family of Hamiltonians) from any polynomial-size quantum circuit.
This provision is even stronger than strictly proving non-existence.

%Next, we construct a family of one-dimensional, local, and translationally invariant Hamiltonians whose energy-conserving PRUs do not exist. We establish this non-existence by providing an efficient algorithm to distinguish the energy-conserving Haar-random unitaries from any polynomial-size quantum circuit, a statement stronger than merely proving non-existence.

\begin{theorem}[Hard Hamiltonians with no energy-conserving PRUs]    
\label{main:thm:hard-ham}
There exists a uniform family of one-dimensional, local, and translationally invariant Hamiltonians $\mathcal{H}$ whose matrix elements belong to $\{0,1,10,1/2,1/4\}$, such that there is a universal algorithm to distinguish the energy-conserving Haar-random unitaries of $\mathcal{H}$ from any polynomial-size quantum circuit.
\end{theorem}

\noindent We choose the stated coefficients to ensure that each Hamiltonian $H$ in the family $\mathcal{H}$ is directly representable on digital computers. Our construction builds upon the Feynman-Kitaev Hamiltonian for quantum simulation of Turing machines~\cite{feynman1986quantum,kitaev2002classical}. We make several significant improvements over this standard construction to enable the energy-conserving Haar-random unitary of $H$ to solve $\mathsf{PSPACE}$-complete problems, a set of problems believed to be hard for quantum computers. In particular, $\mathsf{PSPACE}$-complete problems can be used to invert quantum-secure one-way functions, which we utilize to construct the universal distinguishing algorithm in Theorem~\ref{main:thm:hard-ham}. Our constructions of the Hamiltonian and distinguishing algorithm are overviewed in Section~\ref{main:hard-ham}. A detailed statement and proof of the theorem are provided in Theorem~\ref{thm:universal-distinguisher} in Appendix~\ref{ap:verifier-random-unitary}.

Our results so far have shown that some simple Hamiltonians allow energy-conserving PRUs, while other Hamiltonians do not.
This naturally raises the question: How can one determine whether energy-conserving PRUs exist for a given Hamiltonian $H$?
Our final result proves that this problem is undecidable.
We prove this by combining the two theorems above and embedding the canonical undecidable problem, the halting problem, into the degeneracy of Hamiltonians.
% Finally, by combining the two theorems above and embedding the halting problem, the canonical undecidable problem, into the degeneracy of Hamiltonians, we demonstrate that the existence of energy-conserving PRUs is undecidable.

\begin{theorem}[Undecidability of existence of energy-conserving PRUs]
\label{main:thm:undecidable-deciding-pru}
Determining whether a given uniform family of local Hamiltonians has energy-conserving PRUs is an undecidable problem.
\end{theorem}

\noindent Our construction is described in Section~\ref{main:undecidability}. A detailed statement and proof of this theorem is given in Theorem~\ref{thm:undecidable-deciding-pru} in Appendix~\ref{ap:undecidable}.

\section{Discussions}

Energy-conserving PRUs produce dynamics indistinguishable from truly random dynamics that satisfy energy conservation, making them natural tools for investigating and mimicking chaotic and thermalizing quantum dynamics in local Hamiltonians.
Our results lead to several implications and open questions from this perspective.
%We will discuss the physical implications of our results from this perspective.

%
To obtain physical intuition about energy-conserving PRUs, let us first clarify what energy-conserving Haar-random unitaries represent. Energy-conserving Haar-random unitaries are most directly related to Hamiltonian dynamics when the Hamiltonian $H$ has a generic non-degenerate energy spectrum.
%where the former are equivalent to extremely-long-time dynamics generated by $H$~\cite{mark2024maximum, mok2025optimal}.
%
%For Hamiltonians $H$ with generic non-degenerated eigenspectrum, energy-conserving Haar-random unitaries are equivalent to extremely-long-time dynamics generated by the Hamiltonian $H$~\cite{mark2024maximum, mok2025optimal}.
%
Consider an $n$-qubit Hamiltonian $H = \sum_{k=1}^{2^n} E_k \ketbra{k}{k}$ with a generic non-degenerate spectrum.
The subgroup of unitaries that conserve energy under $H$ are of the simple form $U = \sum_k e^{i \theta_k} \ketbra{k}{k}$, for an arbitrary phase $\theta_k \in [0, 2\pi)$ for each $k$ from $1$ to $2^n$.
An energy-conserving Haar-random unitary is hence:
\begin{equation}
U = \sum_k e^{i \theta_k} \ketbra{k}{k},
\end{equation}
for uniformly random phases $\theta_k \in [0, 2\pi)$.
When the energy spectrum is generic, the set of energy-conserving random unitaries is the same as the set of unitaries $e^{-iHt}$ for time $t \in (-\infty, \infty)$~\cite{mark2024maximum, mok2025optimal}. An energy-conserving Haar-random unitary corresponds to $e^{-iHt}$ for an extremely large random time $t$.
We note that the same insight is used to prove Theorem~\ref{main:thm:construct-pru-commuting} and~\ref{main:thm:hard-ham}, where the constructed Hamiltonians have nondegnerated spectra (with high probability).

%Due to the $2^n$ independent random phases $\theta_1, \ldots, \theta_{2^n}$, the space of energy-conserving unitaries is doubly exponentially large in system size $n$, requiring $\mathcal{O}(2^n)$ bits to specify a single unitary.
The $2^n$ independent random phases $\theta_1, \ldots, \theta_{2^n}$ require $O(2^n)$ bits to even specify a single unitary. Hence, the total degree of randomness of this set is $2^{O(2^n)}$.
This extreme randomness requirement cannot be efficiently achieved through finite-time Hamiltonian evolution by the following counting argument.
If we consider unitaries generated by $H$ through evolution over $t \in [-T, T]$, the number of distinct unitaries is at most $O(T)$.
Therefore, achieving exponentially many bits of randomness requires $T$ to be doubly exponential in $n$.
%
%Furthermore, if the spectrum is arbitrarily close to being degenerate, the required time can be much larger than doubly exponential.
%
Even for fast-forwardable Hamiltonians~\cite{atia2017fast, gu2021fast} such as commuting systems, creating an energy-conserving Haar-random unitary requires simulating $e^{-iHt}$ for doubly exponentially large $t$, demanding exponentially large quantum circuits.
When $H$ has energy degeneracies, energy-conserving Haar-random unitaries are  even stronger than extremely-long-time Hamiltonian dynamics: they require the unitary to scramble every degenerate subspace. 
%An energy-conserving Haar-random unitary is even harder to specify in those cases.
%
Energy-conserving PRUs attempt to circumvent these intrinsic barriers by providing computationally efficient approximations that are indistinguishable from these exponentially complex objects.

Our central finding is that this circumvention is not always possible.
For our hard Hamiltonians (Theorem~\ref{main:thm:hard-ham}), the exponential complexity inherent in energy-conserving Haar-random unitaries creates detectable computational signatures.
The key insight is that certain features of this exponential complexity can be revealed in polynomial time, allowing an observer to efficiently distinguish genuine energy-conserving Haar unitaries from any polynomial-size quantum circuit.
An observer need only interact with the unitary polynomially many times to obtain evidence that the underlying dynamics encode superpolynomially complex circuits, thereby distinguishing genuine energy-conserving Haar unitaries from any attempted polynomial-size circuit approximation.
Under standard subexponential-hardness assumptions for quantum-secure one-way functions, our distinguishing algorithms can separate energy-conserving Haar unitaries from any subexponential-size quantum circuit.

These findings represent worst-case statements about the computational complexity of constructing energy-conserving pseudorandom dynamics.
The dichotomy between efficient energy-conserving PRUs for commuting Hamiltonians, and their proven impossibility for our hard systems, represent extreme ends of the computational spectrum.
The central open question is whether generic random Hamiltonians or naturally-occurring Hamiltonians are likely to admit energy-conserving PRUs.
This average-case complexity question remains unresolved and is not precluded by our undecidability result, since probabilistic statements about random Hamiltonian families could potentially be established even when no algorithm can decide individual cases.
Resolving this question would determine whether the computational barriers we have identified constitute fundamental obstacles to mimicking energy-conserving random quantum dynamics, or merely represent pathological edge cases within the broad space of physically realizable Hamiltonians.

\section{Proof overview}

In this section, we provide detailed descriptions of the theoretical constructions and proofs underlying each of our main results, Theorems~\ref{main:thm:construct-pru-commuting},~\ref{main:thm:hard-ham}, and~\ref{main:thm:undecidable-deciding-pru}.

\subsection{Hamiltonians with energy-conserving PRUs}\label{main:easy-ham}

In this section, we describe how to construct energy-conserving PRUs for random commuting Hamiltonians. We focus on Hamiltonians of the form $H=\sum_i\CJ_i h_i$, where $\{h_i\}$ forms a complete set of local terms that commute with each other, and $\{\CJ_i\}$ are i.i.d. Gaussian random variables (more precisely, Gaussian random variables with proper digitization). Technical details can be found in Appendix~\ref{ap:pru-commuting-ham}.

To construct energy-conserving PRUs, we first clarify the structure of energy-conserving Haar-random unitaries. As already mentioned, any unitary that commutes with $H$ must be block-diagonal with respect to its eigenspaces. The energy-conserving Haar-random unitaries then scramble each subspace independently.
\begin{observation}
[Structure of energy-conserving Haar-random unitaries, Fact~\ref{non-degenerated-pru} in Appendix~\ref{ap:preliminary}]\label{main:obs:non-degenerated-pru}
    Let $\mathcal{H}^{k}$, $1\leq k\leq m$ be the degenerate subspaces of $H$, then any $U\in\mathcal{C}_H$ takes the form $U=\bigoplus_{k=1}^mU_k$, where $U_k$ is a random unitary inside $\mathcal{H}^{k}$. 
\end{observation}
\noindent Specifically, when $H$ has no energy degeneracy, $U=\sum_{k=1}^{2^n}e^{i\theta_k}\ketbra{k}{k}$, where $\ket{k}$ is the $k$-th eigenstate, and $\theta_k$ is a random phase factor. This ensemble is called the random phase ensemble in the literature~\cite{mark2024maximum, mok2025optimal}. The simplified structure of the random phase ensemble enables an efficient construction. This is where the random Gaussian coefficients $\{\CJ_i\}$ prove essential: they prevent $H$ from having energy degeneracies since it lacks level repulsion.

To add random phases to each energy eigenstate, we use the quantum phase estimation (QPE) algorithm~\cite{kitaev1995quantum} as a subroutine. 
% QPE was originally proposed by Kitaev~\cite{kitaev1995quantum} as an algorithm to measure the eigenvalues of unitary matrices. 
The algorithm requires ancilla qubits to set the estimation precision~\cite{cleve1998quantum,kitaev2002classical,dalzell2023quantum}. Assume we are working with an $n$-qubit system with $m$ ancillas. Denote $U$ to be the $n$-qubit unitary, and $\ket{\psi}$ to be an eigenstate of $U$ with eigenvalue $e^{i2\pi\psi}$, where $\psi\in[0,1)$. QPE then transforms $\ket{\psi}\ket{0^m}$ to $\ket{\psi}\ket{\widetilde{\psi}}$, where $\ket{\widetilde{\psi}}$ represents the digitization of $\psi$ stored in the ancilla qubits. To achieve $2^{-m}$ precision, QPE applies $m$ controlled operations $U_c^{2^j}$ for $j=1,2,\ldots,m$, where $U_c$ is the controlled $U$ by a single qubit. We refer readers to Appendix~\ref{ap:preliminary} for a more detailed introduction.

Using linearity, QPE transforms any superposition state $\big(\sum_{\psi}c_{\psi}\ket{\psi}\big)\ket{0^m}$ into $\sum_{\psi}c_{\psi}\ket{\psi}\ket{\widetilde{\psi}}$. Then we add phase factors to each eigenstate by applying an oracle $\CO_f:\ket{x}\rightarrow (-1)^{f(x)}\ket{x}$ for $x\in\{0,1\}^m$ to the ancilla register, where $f:\{0,1\}^m\rightarrow \{0,1\}$ is a Boolean function. By choosing $f$ from pseudorandom Boolean functions (which can be constructed from one-way functions), $\CO_f$ can be efficiently constructed, and  generates pseudorandom phases that are computationally indistinguishable from truly random phases. After applying the inverse of QPE to uncompute the ancillas, the resulting state $\big(\sum_{\psi}(-1)^{f(\widetilde{\psi})}c_{\psi}\ket{\psi}\big)\ket{0^m}$ is computationally indistinguishable from $\big(\sum_{\psi}e^{i\theta_\psi}c_{\psi}\ket{\psi}\big)\ket{0^m}$ for truly random $\{\theta_\psi\}$. Therefore, when choosing $U$ to be $e^{iH}$, the sequential application of QPE, pseudorandom phases, and QPE-inverse acts indistinguishably from a random phase unitary, providing a construction for energy-conserving PRUs.

To properly discriminate different energy eigenvalues by QPE, we need to choose $m$ to be polynomial in $n$ to achieve an inverse exponential precision. This requires applying $m$ different controlled operations $e^{iHt}$ with $t=O(2^m)$. Nevertheless, since $H=\sum_i\CJ_ih_i$ is a commuting Hamiltonian, one can factorize $e^{iHt}=\prod_je^{i\CJ_jh_j t}$ and simulate each local evolution to achieve efficient simulation, a property known as fast-forwardability~\cite{atia2017fast, gu2021fast}. Together, the overall running time for the construction is polynomial in $n$.

\begin{figure}
    \centering
    \includegraphics[width=1.0\linewidth]{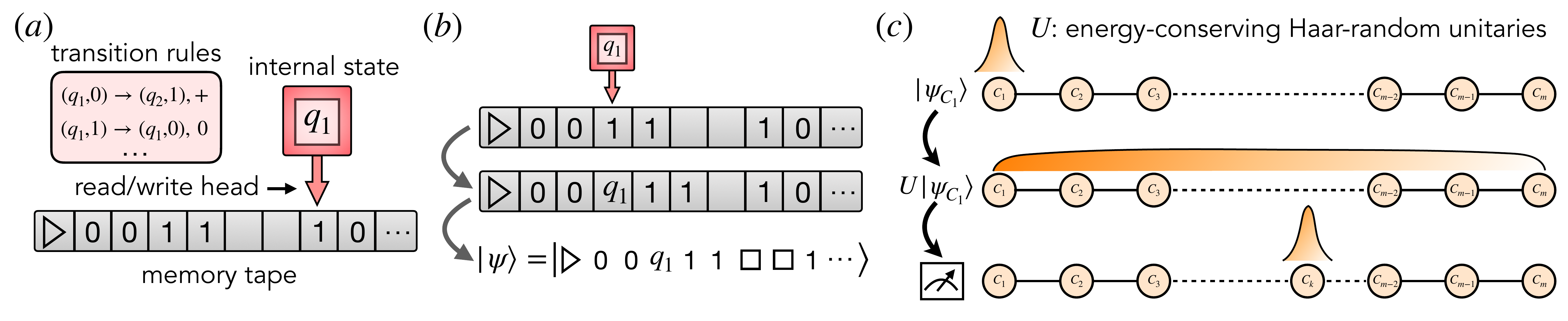}
    \caption{(a) Illustration of Turing machines. A TM consists of an infinite tape storing symbols, a read/write head that moves left or right who has a  internal state, and a set of transition rules to determine the actions and movements of the head. (b) To represent any configuration of TMs (with finite memory size) with a product state, we first put the head state into the tape to make the whole system strictly one dimensional. Then the configuration can be represented by a quantum product state with local Hilbert spaces containing all the symbols and internal states. (c) The idea of using energy-conserving Haar-random unitaries to do fast computation: any computation process can be viewed as moving along a one-dimensional path with vertices labeled by TM's configurations. After initialize the wavefunction to the left terminal corresponding to inputs, the Haar-random unitaries scramble the wavefunction to disperse along the whole path. Then a subsequent measurement collapses the wavefunction to any vertex with almost equal probability, making an exponentially large step forward at once. This idealized picture faces practical obstacles, see Section~\ref{main:hard-ham}. }
    \label{fig:illustration-tm}
\end{figure}

\subsection{Hamiltonians without energy-conserving PRUs}\label{main:hard-ham}

Now we turn to the construction of our hard Hamiltonian $H$ whose energy-conserving PRU do not exist. We prove the non-existence by presenting a universal algorithm to distinguish energy-conserving Haar-random unitaries of $H$ from any polynomial-size quantum circuit (or subexponential-size quantum circuit under subexponential hardness assumption for quantum-secure one-way functions).

The essential feature of $H$ is that its energy-conserving Haar-random unitary can be used to solve $\mathsf{PSPACE}$-complete problems: problems that can be solved with polynomial memory cost but unbounded time complexity. These problems are believed to be extremely hard even for quantum computers. With a $\mathsf{PSPACE}$-complete solver, one can efficiently invert quantum-secure one-way functions, which cannot be achieved by any polynomial-size quantum circuit. We construct the universal distinguishing algorithm by verifying whether one has successfully inverted the one-way function.

\subsubsection{Turing machine and Hamiltonian construction}

We first describe the construction of our hard Hamiltonian satisfying the mentioned properties. Our construction proceeds by embedding the dynamics of a Turing machine (TM) into a local Hamiltonian, similar to the Feynman-Kitaev construction of computational Hamiltonians~\cite{feynman1986quantum,kitaev2002classical}.

A TM, as illustrated in Fig.~\ref{fig:illustration-tm}, can be thought of as a one-dimensional dynamical system: it consists of an infinite tape storing symbols from a finite set $\Gamma$, a read/write head that moves left or right, and a finite set of internal states $Q$ of the head that together evolve under simple local update rules $\Delta$. When performing computation, we first load the inputs consecutively into the tape. Then the machine updates the tape and internal state according to the rules until it halts, if at all. The output of computation, $\acc$ or $\rej$, can then be read out. The number of steps the TM takes before halting is referred to as the time complexity of the problem. Despite this elementary structure, TMs are computationally universal in the sense that any algorithm can be expressed as such a sequence of updates. Thus, TMs serve as the foundation for modern theoretical computer science. For a more detailed overview, see Appendix~\ref{ap:preliminary} or standard textbooks~\cite{sipser1996introduction,papadimitriou2003computational,arora2009computational}.

For our purpose, we restrict the tape length to be finite and polynomial in input size, and require the TM to be reversible, i.e., any configuration of the machine has at most one predecessor. Note that the problems solvable by this type of TM form the set of $\mathsf{PSPACE}$ problems, named for the finite memory space of the TMs. We further require the tape to have periodic boundary conditions. With these restrictions, we can map any configuration of the TM with tape length $L$ to a one-dimensional quantum product state with length $L+1$, as illustrated in Figure~\ref{fig:illustration-tm}. In this mapping, each local Hilbert space contains all the symbols and internal states, thus forming identical qudit Hilbert spaces. Furthermore, since the update rules are local and homogeneous along the tape, we can express them as quantum isometries $V_\delta=\sum_{i=1}^L V_{\delta,i}$ for $\delta\in \Delta$, where $V_{\delta,i}$ represents the realization of transition rule $\delta$ at site $i$. Details of the mapping can be found in Appendix~\ref{ap:hardness-ham-construction}. Therefore, we define $H=\sum_{\delta\in \Delta}\big(V_{\delta}+V_{\delta}^\dagger\big)$ as a one-dimensional, local, and translationally invariant Hamiltonian.

By definition, $V_{\mathrm{forward}}=\sum_{\delta\in\Delta}V_{\delta}$ encodes all the transition rules. For example, let $\ket{\psi_\mathcal{C}}$ be the state that corresponds to some configuration $\mathcal{C}$ of the TM. Then $\ket{\psi_{\mathcal{C}'}}=V_{\mathrm{forward}}\ket{\psi_{\mathcal{C}}}$ represents the successor configuration after one step of update. Moreover, since the TM is designed to be reversible, $V_{\mathrm{forward}}^\dagger\ket{\psi_{\mathcal{C}'}}=\ket{\psi_{\mathcal{C}}}$. Therefore, we can view a computation process upon some input configuration as a unidirectional hopping along a one-dimensional path formed by successive configurations, as illustrated in Figure~\ref{fig:illustration-tm}. Each computational process corresponds to an invariant subspace, with effective Hamiltonian being the hopping Hamiltonian along the one-dimensional path. The initial configuration corresponds to a localized wavefunction at the initial terminal of the path. The hardness of $\mathsf{PSPACE}$-complete problems lies in the exponential length of the corresponding paths, thus having exponential time complexity.

This picture suggests that, if we have access to the energy-conserving Haar-random unitaries of $H$, difficult problems may be solved efficiently using the following method. To begin with, one prepares the product state $\ket{\psi_{\mathcal{C}_1}}$ that corresponds to the input configuration $\mathcal{C}_1$ for a given input $x$. This state serves as the source of the path, whose sink encodes the solution. Then one can sample a unitary $U$ from the energy-conserving Haar-random ensemble and apply it to $\ket{\psi_{\mathcal{C}_1}}$. The scrambling nature of $U$ will produce a wavefunction $U\ket{\psi_{\mathcal{C}_1}}$ dispersing along the entire path. A follow-up measurement in the computational basis will collapse $U\ket{\psi_{\mathcal{C}_1}}$ to any state $\ket{\psi_{\mathcal{C}_t}}$ along the chain with almost equal probabilities. In this way, with high probability, one can make an exponentially large step forward using one query of $U$. The core ideas are illustrated in Fig~\ref{fig:illustration-tm}.

This method, while plausible, faces two obstacles. First, even if we can make an exponentially large step forward along the path at once, the probability of precisely collapsing to the sink is still exponentially small. This can be overcome by adding idling steps to the TM after the solution is reached. In Appendix~\ref{ap:hardness-ham-construction}, we introduce duplications of the TM to double the length while keeping the solution readable along the second half of the path. Thus, the probability of reading out the solution after one query of $U$ is amplified to $O(1)$.

Second, different subspaces of chains can have energy degeneracies when, e.g., two paths have the same lengths. In this case, the energy-conserving Haar-random unitaries entangle multiple subspaces, causing a false readout. The way we overcome this is to add perturbations to the Hamiltonian to break degeneracies. More concretely, we demonstrate that perturbations with coefficients drawn from $\{0,1,10,1/2,1/4\}$ suffice to energetically separate all the $\acc$ and $\rej$ paths from each other and the rest of the Hilbert space. The proof can be found in Appendix~\ref{ap:hardness-ham-construction}, with additional details given in Appendices~\ref{hopping-problem} and~\ref{ap:proof}.

Using these results, we can construct a one-dimensional, local, and translationally invariant Hamiltonian whose energy-conserving Haar-random unitaries can be used to solve $\mathsf{PSPACE}$-complete problems.

\begin{lemma}
[Energy-conserving random unitary can be used to solve $\pspace$ problems, Theorem~\ref{thm:ru-solve-pspace} in Appendix~\ref{ap:hardness-ham-construction}]
\label{main:thm:ru-solve-pspace}
    For any $\pspace$-complete problem, there exists a  one-dimensional, local, and translational-invariant Hamiltonian $H$, whose matrix elements belong to $\{0,1,10,1/2,1/4\}$, such that a polynomial quantum algorithm with query access to the energy-conserving Haar-random unitaries of $H$  exists to solve the problem  with high probabilities.
\end{lemma}

\subsubsection{Distinguishing a $\pspace$-solver from efficient quantum circuits}

A solver for $\mathsf{PSPACE}$-complete problems can be used to invert quantum-secure one-way functions, a property not possessed by any polynomial-size quantum circuit. To see this, we use the \emph{True Quantified Boolean Formula} problem (TQBF), a canonical $\mathsf{PSPACE}$-complete problem~\cite{stockmeyer1973word,chandra1981alternation}, as an illustration. TQBF consists of all fully quantified Boolean formulas that evaluate to true over the Boolean domain. Formally, the language is defined as:
\begin{equation*}
\mathrm{TQBF} := \left\{ \varphi = Q_1 x_1 \cdots Q_n x_n \; \psi(x_1, \ldots, x_n) \;\middle|\;
\begin{array}{l}
Q_i \in \{ \forall, \exists \},\ \psi \text{ is a Boolean formula}\\
\text{and } \varphi \text{ evaluates to true.}
\end{array}
\right\}.
\end{equation*}
We assume all quantifiers precede the propositional formula, and that $\psi$ is encoded either in conjunctive normal form or as a Boolean circuit. The input size of a formula is the number of bits required to encode the quantifiers and $\psi$. For example, the formula
$\phi = \forall x_1 \exists x_2 \forall x_3 \;\left[(x_1 \lor \neg x_2) \land (x_2 \lor x_3)\right]$ is in $\mathrm{TQBF}$.

Given a circuit $C$ computing the one-way function $f$ and target $y$, we construct the TQBF formula:
$\Phi_{C,y} := \exists x_1, \ldots, x_n \; \psi(x_1, \ldots, x_n, y)$
where $\psi(x_1, \ldots, x_n, y)$ is a Boolean formula that evaluates to true if and only if $C(x_1, \ldots, x_n) = y$. This formula can be constructed in polynomial time by simulating the circuit $C$. First, we query the TQBF solver on $\Phi_{C,y}$. If it returns $\rej$, then no preimage exists. If it returns $\acc$, we extract a witness using binary search: for each bit position $i = 1, \ldots, n$, we construct the formula: $\Phi_i^0 = \exists x_{i+1}, \ldots, x_n \; \psi(a_1, \ldots, a_{i-1}, 0, x_{i+1}, \ldots, x_n)$ where $a_1, \ldots, a_{i-1}$ are the bits determined in previous iterations. We query the solver on $\Phi_i^0$. If it returns $\acc$, we set $a_i = 0$; otherwise, we set $a_i = 1$. After $n$ such queries, we obtain $(a_1, \ldots, a_n)$, which is guaranteed to be a valid preimage since the original formula was satisfiable and the oracle is correct. The total number of oracle queries is $n + 1$. In this way, we invert the one-way function efficiently.

In contrast, any polynomial-size quantum circuit cannot invert one-way functions. So, we can construct an efficient algorithm to distinguish a $\mathsf{PSPACE}$-complete solver from any polynomial-size quantum circuit by checking if it correctly inverts the one-way function. Details are in Appendices~\ref{ap:tqbf} and~\ref{ap:verify-tqbf-via-oneway}.

\begin{lemma}[Distinguishing the TQBF solver, Theorem~\ref{thm:tqbf-verification} in Appendix~\ref{ap:verify-tqbf-via-oneway}]
\label{main:thm:tqbf-verification}
Assume quantum-secure one-way functions exist. There exists a polynomial-time classical algorithm $\mathcal{V}$ (the verifier) that, given black-box access to a purported $\mathsf{TQBF}$ solver $\mathcal{O}$, such that if $\CO$ solves $\mathsf{TQBF}$ correctly, outputs $\acc$ with high probability, if $\CO$ implements any polynomial size quantum circuit, outputs $\rej$ with high probability.
\end{lemma}

\noindent Combining Lemmas~\ref{main:thm:ru-solve-pspace} and~\ref{main:thm:tqbf-verification}, we can prove Theorem~\ref{main:thm:hard-ham}. Proof details are shown in Appendix~\ref{ap:verifier-random-unitary}.

\subsection{Undecidability of existence energy-conserving PRU}\label{main:undecidability}

\begin{figure}
    \centering
    \includegraphics[width=0.8\linewidth]{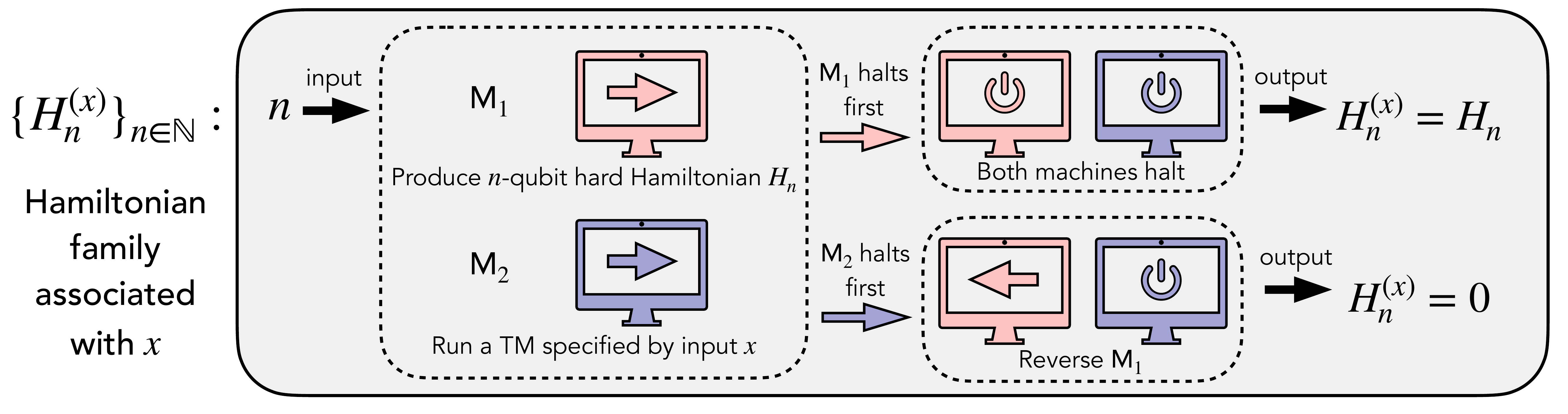}
    \caption{
    Turing machine that generates a Hamiltonian family $H_n^{(x)}$ with or without energy-conserving PRU, depending on whether universal TM halts upon input $x$. After receiving an put $n$, two machines run in parallel: $\tm_1$ generates the hard Hamiltonian $H$ from Section~\ref{main:hard-ham}, and $\tm_2$ is a UTM operating upon input $x$. If $\tm_1$ halts, then the whole system halts with output $H$. If $\tm_2$ halts, $\tm_1$ halts and does the reverse computing. Thus the system outputs $0$.}
    \label{fig:undecidable}
\end{figure}

Given that both easy and hard instances exist, it is natural to seek an algorithm that can decide whether a given family of Hamiltonians has energy-conserving PRUs. However, this problem is inherently undecidable, meaning that no algorithm, even with exponential computational resources, can solve it.

To meaningfully define a computational problem, we are not asking only for the solution for a given input length, where one can always solve it by brute force. Rather, we specify an infinite series of problems $\{f_n\}_{n\in\mathbb{N}^+}$ with all possible input lengths, and ask about the asymptotic cost to solve them when $n$ grows. When stating the aforementioned problem rigorously, we are asking whether \emph{a uniform family of} Hamiltonians $\{H_n\}_{n\in\mathbb{N}^+}$ has energy-conserving PRUs or not. What we take as input is in fact a set of rules that specifies the Hamiltonian for any system size, and more concretely, a TM that generates the Hamiltonian's description when taking $n$ as input.

Since we are taking TMs as inputs, the existence of energy-conserving PRUs can be linked to the solution of the halting problem, a canonical undecidable problem~\cite{turing1936computable,church1936unsolvable}. The halting problem asks whether the universal Turing machine (UTM), a TM that can simulate all other TMs, eventually halts or not upon given inputs.

Suppose we let two TMs run in parallel. The first TM generates the description of the hard Hamiltonian in Theorem~\ref{main:thm:hard-ham}, and the second one is a UTM that operates upon a fixed input $x$. If the first TM halts, then the whole system halts. If the second one halts, then the first one stops and performs exactly the inverse operations to uncompute the previous outputs. In this way, if the UTM eventually halts upon $x$, the whole machine outputs $0$ for sufficiently large $n$, whose energy-conserving PRU is the conventional PRU. Otherwise the machine outputs the hard Hamiltonian. Therefore, if we can decide whether a family of Hamiltonians has energy-conserving PRUs or not, we must be able to decide whether the UTM halts upon any given input $x$. The whole construction is illustrated in Fig~\ref{fig:undecidable}.

We note that when restricted to translationally invariant Hamiltonians, our proof does not apply. There a Hamiltonian is completely specified by its finite set of local terms. It is unclear whether determining the existence of energy-conserving PRUs is still hard in those cases.

\vspace{1em}
\section*{Acknowledgments}

We are grateful to Adam Bouland, Jeongwan Haah, Tony Metger, and John Preskill for insightful discussions about energy-conserving pseudorandom unitaries.
L.M. is grateful to the helpful discussions with Hongkang Ni on the proof of Lemma~\ref{lem:solution-eigenvalue}.
L.M. acknowledges support from Tsinghua University.
T.S. acknowledges support from the Walter Burke Institute for Theoretical Physics at Caltech. T.S. and H.H. acknowledge support from the U.S. Department of Energy, Office of Science, National Quantum Information Science Research Centers, Quantum Systems Accelerator. The Institute for Quantum Information and Matter is an NSF Physics Frontiers Center.

\setcounter{lemma}{0}
\setcounter{definition}{0}
\setcounter{fact}{0}
\setcounter{theorem}{0}
\clearpage
\vspace{2.5em}
\appendix

\noindent 
\textbf{\LARGE{}Appendices}
\vspace{2em}

% \noindent We provide a road map to help the readers navigate through our appendices. Appendix \ref{app:literature-review} presents a brief literature review of existing results relevant to this work.
% Appendix \ref{app: unitary designs meta} establishes our main results regarding random unitary designs.
% Appendix~\ref{app: PRUs} establishes our results regarding pseudorandom unitaries.
% Appendix~\ref{app: any geometry} shows how to generalize our results from 1D circuits to any circuit geometry. % that can be represented by a connected bounded-degree graph.
% Finally, Appendix~\ref{app: applications} presents the technical details of the applications considered in the main text.

\tableofcontents

\section{Preliminaries}\label{ap:preliminary}

\subsection{Notations and matrix norms}
We use the standard notations $O$, $\Omega$ and $\Theta$ to denote asymptotic upper, lower, and tight bound, respectively. We use the  $\poly{n}$ ($\mathsf{exp}(n)$) to denote any functions $f(n)$ such that there exists $c>0$ satisfying $f(n)=O(n^c)$ ($O(e^{n^c})$). $\negl{n}$  denotes $f(n)$ smaller than any inverse polynomial functions of $n$.

 Throughout the appendices, we will mostly use capital letter like $U$ to denote a unitary matrix, or a quantum circuit, and calligraphic letter like $\CU$ to denote an ensemble of unitary matrices, or other types of sets. The only exception that we use $\CO$ to denote an oracle, which is also a circuit, since $O$ is used to denote the asymptotic upper bound.

 Norms of matrices and vectors are denoted as $\norm{A}$ and $\norm{v}$. In this paper, we will use the trace norm of matrices, defined for matrices of any size
 \begin{align}
     \norm{A}_1:=\Tr{\sqrt{A^\dagger A}}.
 \end{align}
 It satisfies submultiplicativity, and triangle inequality
 \begin{align}
     \norm{A_1A_2}_1\leq \norm{A_1}_1\cdot\norm{A_2}_1,\quad
     \norm{A_1+A_2}_1\leq\norm{A_1}_1+\norm{A_2}_1
 \end{align}
If the matrix is a vector, we will use the standard concept of Euclidean norm, which is equivalent to the trace norm if the vector is viewed as a $n\times 1$ matrix. 
 \begin{align}
     \norm{v}_2=\sqrt{v^\dagger v}= \sqrt{\sum_i |v_i|^2}.
 \end{align}
We use $\norm{\cdot}_2$ instead of $\norm{\cdot}_1$ to emphasize the vector form.

 We will use sans-serif letter as $\mathsf{U}$ to denote a quantum channel. 
 Different from multiplication of matrices, we use $\mathsf{A}\circ\mathsf{B}$ to denote subsequent actions of channels, i.e., $\mathsf{A}\circ\mathsf{B}[\rho]:=\mathsf{A}[\mathsf{B}[\rho]]$.
 Channels are subjected to the measure of diamond norm,
 \begin{align}
     \norm{\mathsf{A}}_\diamond:=\sup_{\rho}\norm{(\mathsf{A}\otimes\mathsf{I})[\rho]}_1,
 \end{align}
where $\mathsf{I}$ is the identity channel of an ancillary system of any size, and $\rho$ is a density matrix.
The diamond norm also satisfies submultiplicativity, and triangle inequality. Note that since $\norm{\rho}_1=1$ when $\rho$ is a density matrix, $\norm{\mathsf{A}}_\diamond=1$ when $\mathsf{A}$ is a complete positive trace-preserving mapping.

We will also use sans-serif letter as $\mathsf{S}$ to denote registers, a collection of qubits to define ``system" or ``ancillas". Different usages can be easily identified from the context. We use $|\mathsf{S}|=n$ to denote the number of qubits in a register. We sometimes use subscript like $\ket{0_\sfS}$ to specify the register that a state lives in.

Specifically, $\tm$ and $\D[\tm]$ are used  to denote Turing machines (TMs), which we will define more clearly later.

\subsection{Pseudorandom unitaries}
In this section we review the properties of pseudorandom unitary (PRU) ensembles. We are interested in ensembles which reproduce the characteristics of, and are indistinguishable from, \emph{uniformly} random unitary transformations sampled via the Haar measure.
\begin{definition}[Haar ensemble]
    Given a compact Lie subgroup $\mathbb{V}$ of the $n$-qubit unitary group $\mathbb{U}(2^n)$, the Haar ensemble $\mu(\mathbb{V})$ is the unique ensemble over $\mathcal{\mathbb{V}}$ with normalized probability measure $d$ that is both left- and right-invariant, i.e., for any subset $\mathbb{S}\subseteq \mathbb{V}$ and any $V\in \mathbb{V}$, $d(\mathbb{S})=d(V\cdot \mathbb{S})=d(\mathbb{S}\cdot V)$. 
\end{definition}
A PRU ensemble over $\mathcal{U}$ is a unitary ensemble that can be efficiently generated, but can not be efficiently distinguished from $\mu(\mathcal{U})$\cite{ji2018pseudorandom,ma2024construct}. 
\begin{definition}[Pseudorandom unitaries]
    Let $\{\mathcal{U}_n\}_{n\in\mathbb{N}^+}$ be a uniform family of unitary ensembles $\mathcal{U}_n=\{U_k\}_{k\in\mathcal{K}_n}$, where $U_k\in\mathbb{U}(2^n)$ and $\mathcal{K}_n$ denotes the key subspace. We call $\{\mathcal{U}_n\}$ a pseudorandom unitary ensemble if the followings are satisfied:
    \begin{itemize}
    \item There exists a $\poly{n}$-time quantum algorithm parametrized by $k\in\mathcal{K}_n$ to implement all $U_k\in\mathcal{U}_n$.
        \item For any $\poly{n}$-time quantum algorithm $\{\alg_n^{\mathcal{E}}()\}_{n\in\mathbb{N}^+}$ with query access to unitary ensemble $\mathcal{E}$, if $k$ is uniformly drawn from $\mathcal{K}_n$,
    \begin{align*}
        \Big|\pr[\alg_n^{\mu(\mathbb{U}(2^n))}()=1]-\pr[\alg_n^{\mathcal{U}_n}()=1]\Big|\leq\negl{n}
    \end{align*}
    holds for sufficiently large $n$.
    \end{itemize}
\end{definition}
In this paper, we focus on the cases where the subgroup is a commuting group with some fixed local Hamiltonian.
With energy-conservation, the Haar-random unitaries factorize to the sum of random unitaries inside each degenerated subspace.
\begin{fact}\label{non-degenerated-pru}
    Let $\hilbert_n^{i}$, $i=1\leq i\leq k$ be the degenerated subspaces of $H_n$, then any $U\in\mu(\mathcal{C}_n[H_n])$ takes the form of $U=\oplus_{i=1}^kU_i$, where $U_i\in \mu(\mathbb{U}(\hilbert_n^i))$. Specifically, when $H_n$ has no energy degeneracy, $U=\sum_{E}e^{i\theta_E}\ketbra{E}{E}$, where $E$ spans over all the energy eigenstates, and $\theta_E\overset{\text{i.i.d.}}{\sim}\operatorname{Unif(0,2\pi)}$. 
\end{fact}

This definition of PRU  generalizes to the setting with conserved energy or charges.
\begin{definition}[Energy-conserving pseudorandom unitary ensemble]
Let $\{H_n\}_{n\in\mathbb{N}^+}$ be a uniform family of $n$-qubits  Hamiltonians. Define the commuting subgroup $\mathcal{C}_n^H=\{U|U\in\mathbb{U}(2^n),UH_n=H_nU\}$. 
Let $\{\mathcal{U}_n\}_{n\in\mathbb{N}^+}$ be a uniform family of unitary ensembles $\mathcal{U}_n=\{U_k\}_{k\in\mathcal{K}_n}$, where $U_k\in\mathbb{U}(2^n)$ and $\mathcal{K}_n$ denotes the key subspace. We call $\mathcal{U}_n$ a energy-conserving pseudo-random unitary ensemble if the followings are satisfied:
    \begin{itemize}
    \item There exists a  $\poly{n}$-time quantum algorithm parametrized by $k\in\mathcal{K}_n$ to implement all $U_k\in\mathcal{U}_n$.
        \item For any $\poly{n}$-time quantum algorithm $\{\alg_n^{\mathcal{E}}()\}_{n\in\mathbb{N}^+}$ with query access to unitary ensemble $\mathcal{E}$, if $k$ is uniformly drawn from $\mathcal{K}_n$,
    \begin{align*}
        \Big|\pr[\alg_n^{\mu(\mathcal{C}_n^H)}()=1]-\pr[\alg_n^{\mathcal{U}_n}()=1]\Big|\leq\negl{n}
    \end{align*}
    holds for sufficiently large $n$.
    \end{itemize}
\end{definition}
In above, we state standard definitions of energy-conserving PRU for qubit systems.
They can be generalized to qudit systems for $d=O(1)$ in a straightforward way, where we skip the explicit statements.

\subsection{Cryptographic primitives}

Previous constructions of PRU\cite{schuster2024random,cui2025unitary,foxman2025random} assume the existence of quantum-secure one-way functions, which is a standard cryptography assumption and widely believed to hold.

\begin{definition}[Quantum-Secure One-Way Function]
\label{def:quantum-owf}
A function $f: \{0,1\}^* \to \{0,1\}^*$ is a \emph{quantum-secure one-way function} if:
\begin{enumerate}
\item $f$ is computable in polynomial time;
\item For any $\poly{n}$-time quantum algorithm $\{\alg_n()\}_{n\in\mathbb{N}^+}$, for sufficiently large $n$ we have
\begin{equation}
    \Pr_{x \in \{0,1\}^n}[f(\CA_n(f(x))) = f(x)] \leq \negl{n}
\end{equation}
holds for sufficiently large $n$.
\end{enumerate}
\end{definition}

Throughout this paper, we will make the same assumption in the proofs of both easy and hard Hamiltonians. In the construction of energy-conserving PRU for random commuting Hamiltonians, we will use the pseudo-random functions, which can be constructed out of one-way functions.
\begin{definition}
    [Pseudorandom function]\label{def:prf}
    Let $\{\CF_n\}$ denote a uniform family of functions $\CK_n=\{f_k\}_{k\in{\CK_n}}$, where $f_k:\{0,1\}^n\rightarrow \{0,1\}$ and $\CK_n$ denotes the key subspace. We say $\{\CF_n\}$ are pseudorandom functions if
    \begin{itemize}
        \item $f_k$ is computable in polynomial time;
        \item For any $\poly{n}$-time quantum algorithm $\{\alg_n^{\CO}()\}_{n\in\mathbb{N}^+}$ with query access to an oracle $\CO$, if $k$ is uniformly drawn from $\mathcal{K}_n$,
    \begin{align*}
        \Big|\pr[\alg_n^{\CO_\CF}()=1]-\pr[\alg_n^{\CO_\CR}()=1]\Big|\leq\negl{n}
    \end{align*}
    holds for sufficiently large $n$, where $\CO_\CF:\ket{x}\rightarrow (-1)^{f_k(x)}\ket{x}$ for $x\in\{0,1\}^n$, and $\CO_\CR:\ket{x}\rightarrow e^{i\theta}\ket{x}$ for $\theta\sim\operatorname{Unif}[0,2\pi)$.
    \end{itemize}
\end{definition}

\subsection{Turing machines and complexity classes}

We now review Turing machines (TM), a model for universal computation\cite{sipser1996introduction,papadimitriou2003computational,arora2009computational}. We will use this framework to define notions of computational complexity. We focus on decision problems where the answer is either $\acc$ or $\rej$.
\begin{definition}
    [Decision problem] A decision problem is specified by a language $L\subseteq \{0,1\}^*=\cup_{n=0}^\infty \{0,1\}^n$. The problem is that give any $x\in \{0,1\}^*$, output $\acc$ if $x\in L$, otherwise output $\rej$. $|x|$ is called the size of input.  
\end{definition}\noindent
Given an instance of a decision problem, we can consider the corresponding TM. Informally, a TM consists of three parts, as illustrated in Fig~\ref{fig:illustration-tm}: a tape with some symbols, a read-write head with some internal states set $Q$, and a set of transition rules $\Delta$. The tape provides the space for storage information and doing operations. Each cell of the tape contains one symbol from a $\textit{finite}$ set of alphabet $\Gamma$, like 0,1.
%Without specification, the tape is a semi-infinite one-dimensional object that has a left terminal. 
The head can read the symbol of the cell it currently locates, and is associated with an internal state $q_\mu\in Q$. Then according to the symbol and the current internal state $q_\mu$, the head can change the internal state, rewrite the symbol, and then move left or right or stay still. The rules that determine the behavior of the head is the transition function $\Delta:Q\times\Gamma\rightarrow Q\times \Gamma \times \{+,-,\mathsf{0}\}$.
Here $q,x$ denote the current internal state and symbol, and in one round of operation they will be changed to $q'$ and $x'$, respectively. After then the head will move right, left, or stay, depending on $s=+,-,\mathsf{0}$. In this way a TM change the contents of the tape step by step and arrive the final output.
%Specifically, there are internal states called \textit{initial} state $q_0$ and \textit{halt} state $q_h$. $q_0$ denotes the beginning of computation. $q_h$ denotes the end of the computation. That is, the transition function can be viewed as a mapping $\Gamma\times (Q\setminus \{q_h\})\rightarrow \Gamma\times (Q\setminus \{q_0\})\times\{+,-,\mathsf{0}\}$.  Importantly, for a given TM, the set of internal states $Q$ and transition functions $\delta$ are all \textit{finite}.  

%Before doing computation, all cells on the tape are initialized to contain the \textit{blanket} symbol (denoted as $b$). Once the computational task $L$, $x\in\{0,1\}^*$  is specified, $x$ is uploaded onto the tape from the left terminal.
%During the computational, the TM evolves step-by-step by the guidance of $\delta$ from a given input symbols on the tape, until the internal state changes to $q_h$. Then depending on the current symbol that the head reads is $0$ or $1$, the machine outputs $\yes$ or $\no$.
%During the computation process, we use the \textit{configuration} $\C$ to denote an instant (one step) of it, where the symbols of each cell, the position of head and the internal state are fixed. 

More formally,
\begin{definition}
    [Turing machines] A (deterministic) Turing Machine $\tm$ is defined by a triple $\langle Q,\Gamma,\Delta\rangle$, such that
    \begin{itemize}
        \item Q is a finite set of internal states that contains the initial state $q_0\in Q$ and two halting states $q_r,q_a\in Q$, corresponding to $\acc$ and $\rej$.
        \item $\Gamma$ is a finite set of symbols that contains the blank symbol $b$.
        \item $\Delta$ is a finite set of transition functions $\Delta: Q\times\Gamma\rightarrow Q\times \Gamma \times \{+,-,\mathsf{0}\}$ that satisfies for any $(q,x,q',x',s)\in \Delta$, $q\neq q_a,q_r$ and $q'\neq q_0$. 
    \end{itemize}
    
    The TM has a two-way infinite tape of cells indexed by $\mathbb{Z}$ and a single read/write
    tape head that moves along the tape. A configuration of $\tm$, $\C=(q,B,i)\in Q\times \Gamma^*\times \mathbb{Z}=:\mathscr{C}$, is a complete
    description of the contents of the tape $B$, the location of the tape head $i$ and the internal state $q$ of the head. At any time, only a finite number of the tape cells may
    contain non-blank symbols.

    For any configuration $\C\in \mathscr{C}$, the successor configuration $\C'$is defined by
    applying the transition function to the current state and the symbol scanned by the
    head, replacing them by those specified in the transition function and moving the
    head right ($+$), left ($-$), or stay still ($\mathsf{0}$). 

    The initial configuration satisfies the following conditions: the head is in cell 0, called the starting cell, and the machine is in state $q_0$. We say that an initial configuration has input $x\in (\Gamma\setminus \{b\})^*$ if $x$ is written on the tape in positions $0,1,2,\cdots ,|x|$ and all other tape cells are blank. The TM halts on input $x$ if it eventually enters the final state $q_a$ or $q_r$. The output of this computation task is then $\acc$ or $\rej$, depending on halting state. The number of steps a TM takes to halt on input $x$ is its running time on input $x$. 
    %\begin{itemize}
 %       \item $Q$ is a finite set of internal states.
 %       \item $\Gamma$ is a finite set of symbols on the tape. $\yes,\no\in \Gamma$ denote the final outputs of the machine.
 %       \item $b\in \Gamma $ is the blanket symbol, the only symbol that is allowed to appear infinitely often at any step of the computation.
 %       \item $\Sigma\subseteq \Gamma\setminus \{b,\yes,\no\}$ is the set of symbols that represents the input. We choose $\Sigma=\{0,1\}$.
  %      \item $q_0\in Q$ is the initial state. The internal state is initiated to $q_0$ before each computation task.
  %      \item $q_h\in Q$ is the halting state where the machine stops. 
  %      \item $\Delta$ is the finite set of transition functions $\delta=(q\rightarrow q',x\rightarrow x',s)$, i.e., maps $\Gamma\times (Q\setminus \{q_h\})\rightarrow \Gamma\times (Q\setminus \{q_0\})\times\{+,-,\mathsf{0}\}$. For any $(q\rightarrow q',x\rightarrow x',s)\in \Delta$, if $q'=q_h$, then $x'\in\{\yes,\no\}$ and $s=\mathsf{0}$. If $x'\in\{\yes,\no\}$, then $q'=q_h$ and $s=\mathsf{0}$.
  %  \end{itemize}
\end{definition} \label{turing-machine}

By choosing appropriate alphabet, internal states set and transition functions set, a TM runs an \textit{algorithm} which solves a corresponding decision problem. 

To simulate the TM with a quantum Hamiltonian we need to use TMs with a fixed-size memory, i.e., TMs with a fixed tape-length. The tape-length will correspond to the length of the one-dimensional Hilbert space that the Hamiltonian residues.
\begin{definition}
    [Turing machines with fixed-size memory]
   Given a  Turing machine $\tm=\langle Q,\Gamma,\Delta\rangle$, the corresponding Turing machine with fixed memory-size is denoted as $\tm(l)$, where $l\in\mathbb{N}^+$ is the tape length. For any $\tm(l)$, we choose a periodic boundary condition, i.e., the left-most cell of the tape is linked to the $l-$the cell. The configuration of TM with fixed-size memory is defined analogously.
\end{definition}

TMs have many variants. One useful variant is the \textit{Reversible Turing machines} (RTM), which is the TM that each configuration has at most one predecessor. 
\begin{definition}
    [Reversible Turing machine] 
    A Reversible Turing machine  is a TM $\braket{Q,\Gamma, \Delta}$ that any configuration has at most one predecessor and one successor.
\end{definition} \label{reversible-turing-machine}
\begin{fact} [Structure of RTM; cf.\cite{sipser1996introduction,papadimitriou2003computational,arora2009computational}]
    Let $\tm=\langle Q,\Gamma, \Delta\rangle$ be a TM defined by Definition. \ref{turing-machine}. $\tm$ is reversible iff $\Delta$ satisfies the follows:
    \begin{itemize}
    \item Unidirection: For any $\delta=(q,x,q',x',s) \in \Delta$, $s$ is uniquely determined by $q'$.
    \item One-to-one: Ignoring $s$ and viewing $\Delta$ as a map from $Q\times \Gamma$ to $Q\times \Gamma$, it is an injection.
    \end{itemize}
\end{fact}

 The construction of RTM insures every configuration has at most one predecessor and one successor. Then is it possible to construct another RTM to reverse all the computational processes? The answer is definitely yes, but with non-standard form of transition rules. Recall the transition rule $(q,x,q',x',s)$ is defined that the head movement is always the final action. As a result, to exactly reverse the computation of a RTM, we need to move the head first.
 \begin{definition}
     [Reverse  transition rules]\label{def:inverse-transition-rule}
     Given a RTM $\tm=\langle Q,\Gamma, \Delta\rangle$, one can always construct a set of transition rules to reverse its computation, $\Delta^{-1}=\{(q,s,x,q',x')~|~(q',x',q,x,-s)\in\Delta\}$. Here $(q,s,x,q',x')$ means that depending on the current internal state $q$, the head first moves along $s$ direction and then read the current symbol $x$. Subsequently, the head modifies of internal state and symbol to $q'$ and $x'$. We refer transition rules of this form the reverse form. Note that $(q,\mathsf{0},x,q',x')=(q,x,q',x',\mathsf{0})$, so we always regard the transition rules with $s=\mathsf{0}$ as the standard form.

     In the following, we regard a machine that satisfies the definition of TM except for  containing reverse form transition rules also as a TM.
 \end{definition}

 The class $\pspace$ contains the languages that can be solved by a TM with polynomial-size tape with the input. But the computing time is unlimited.
\begin{definition}
    [$\pspace$] $\pspace=\{L\}$ is the set of languages $L$ where there exists a Turing machine $\tm_L$ and a polynomial $p(\cdot)$, such that any $x\in\{0,1\}^*$, one can determine if $x\in L$ using $\tm_L(l)$ for $l=\mathfrak{O}(p(|x|))$. We call this  $\tm_L$ the Turing machine that solves $L$ in polynomial space. 
\end{definition}

\subsection{Quantum phase estimation}
\begin{figure}[!t]
    \centering
    \includegraphics[width=0.39\linewidth]{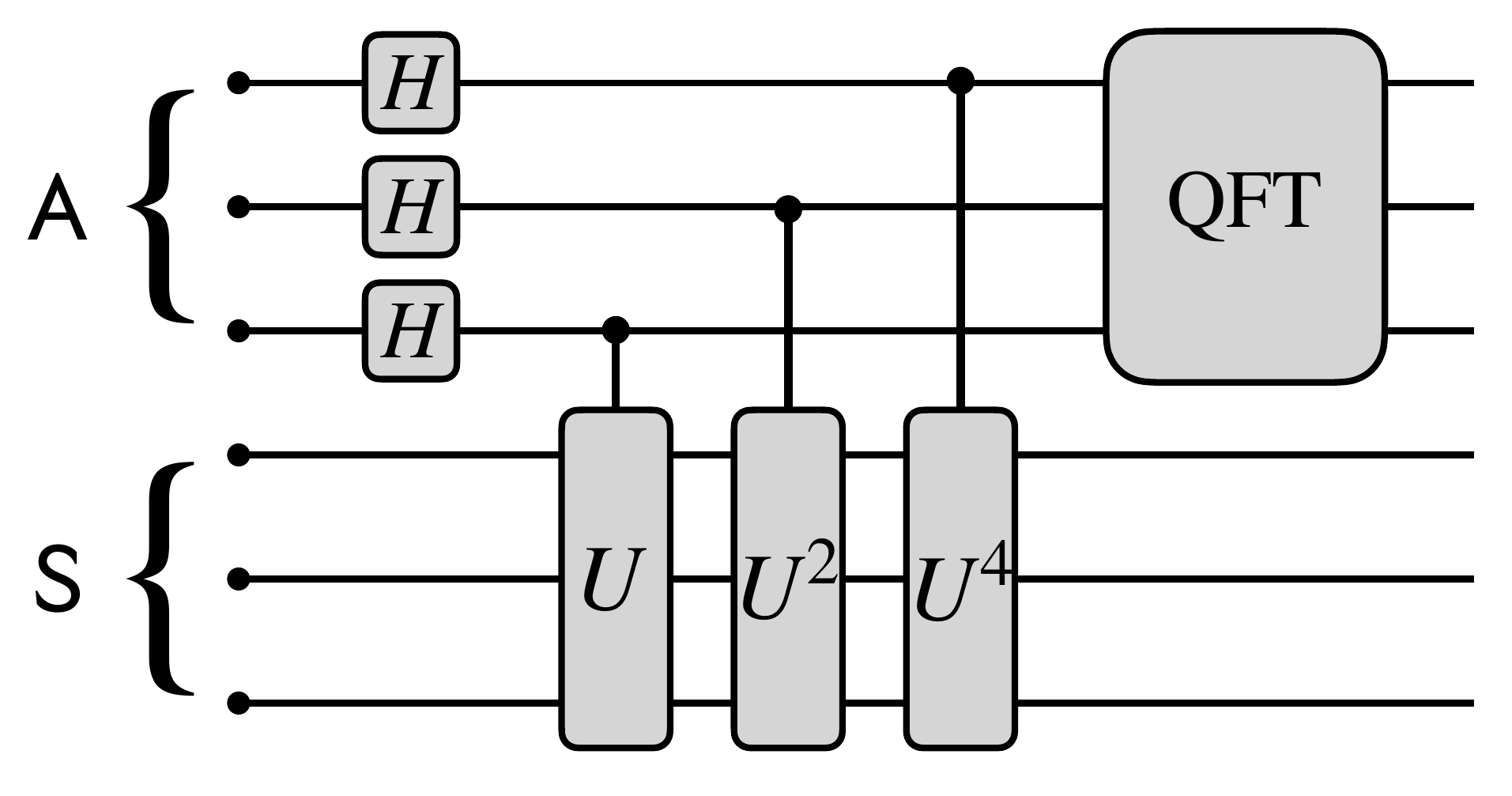}
    \caption{Circuit structure for quantum phase estimation.}
    \label{fig:qpe}
\end{figure}

In this section we review the properties of the quantum phase estimation (QPE) algorithm\cite{kitaev1995quantum,cleve1998quantum,kitaev2002classical,dalzell2023quantum}, which we use as a subroutine in the construction of energy-conserving PRUs. The circuit implementation of QPE is sketched in Fig~\ref{fig:qpe}. 

 Let $\ket{\psi}$ be a $n$-qubit quantum state on register $\sfS$ and $U$ be a unitary such that $U\ket{\psi}=e^{i2\pi \theta}\ket{\psi}$. The goal of QPE is to estimate the phase $\theta\in[0,1)$. For this, we need an ancilla register $\sfA$ to restore the digitalization of $\theta$, which we take to be $m=\poly{n}$ qubits.

Starting from $\ket{\psi}\otimes\ket{0^m}$, the first step is to apply $H^{\otimes m}$ of $\sfA$, where $H$ is Hadamard gate. 
\begin{align}
    \ket{\psi}\otimes\ket{0^m}\xrightarrow{H^{\otimes m}}\frac{1}{\sqrt{2^m}}
    \sum_{x\in\{0,1\}^m}\ket{\psi}\otimes\ket{x}.
\end{align}
Second, we apply controlled $U^{2^{i-1}}$ on $\sfS$ controlled by the $i$-th qubit in $\sfA$ for $1\leq i\leq m$. After this step,
\begin{align}
    \frac{1}{\sqrt{2^m}}
    \sum_{x\in\{0,1\}^m}\ket{\psi}\otimes\ket{x}
    \xrightarrow{\text{controlled $U$}}
    \frac{1}{\sqrt{2^m}}
    \sum_{x\in\{0,1\}^m}e^{i2\pi \hat{x}\theta}\ket{\psi}\otimes\ket{x},
\end{align}
where $\hat{x}$ is the decimal number of $x$. Finally, we utilize the inverse quantum Fourier transformation (QFT),
\begin{align}
    \qft^\dagger:\quad \ket{x\in\{0,1\}^m}\rightarrow \frac{1}{\sqrt{2^m}}\sum_{y\in\{0,1\}^m}e^{-i2\pi \hat{x}\hat{y}/2^m}\ket{y}.
\end{align}
Applying inverse QFT on the register $\sfA$, we have
\begin{align}
    \frac{1}{\sqrt{2^m}}
    \sum_{x\in\{0,1\}^m}e^{i2\pi \hat{x}\theta}\ket{\psi}\otimes\ket{x}
    &\xrightarrow{\qft^\dagger}
    \frac{1}{2^m}\sum_{x,y\in\{0,1\}^m}e^{i2\pi \hat{x}(\theta-\hat{y}/2^m)}\ket{\psi}\otimes\ket{y}\notag\\
    &=
    \sum_{y\in\{0,1\}^m}\alpha_y(\theta)\ket{\psi}\otimes\ket{y},
\end{align}
where
\begin{align}
    \alpha_y(\theta)=\frac{1}{2^m}\frac{1-
    e^{i2\pi 2^m(\theta-\hat{y}/2^m)}
    }{1-e^{i2\pi (\theta-\hat{y}/2^m)}}.
\end{align}

The amplitude distribution $\alpha_y(\theta)$ is sharply peaked around $\theta$. From standard error analysis, we have

\begin{fact}
    [Accuracy of quantum phase estimation, cf.\cite{kitaev1995quantum,cleve1998quantum,kitaev2002classical,dalzell2023quantum}]
    \label{fact:accuracy-qpe}
    For a $n$-qubit unitary $U\in\mathbb{U}(2^n)$ and a $n$-qubit state $\ket{\psi}$ such that $U\ket{\psi}=e^{2\pi i\phi}\ket{\psi}$, $m\in\BN^+$ and $\epsilon\in(0,1)$, there is a quantum circuit acting on $n+k\equiv n+m+1+\lceil 2\log_2\epsilon^{-1}\rceil$ qubits that satifies the following: When taking $\ket{\psi}\otimes\ket{0}_k$ as input, outputs a state $\ket{\psi}\otimes \ket{\phi_{m,\epsilon}}$, such that $\norm{\prod_{\phi,m,\epsilon}\ket{\phi_{m,\epsilon}}}_2\leq \epsilon$, where $\prod_{\phi,m,\epsilon}=\sum_{b\in\{0,1\}^k,|\hat{b}/2^k-\phi|>2^{-(m+1)}}\ketbra{b}{b}$, with $\hat{b}$ the decimal number of $b$. The circuit uses $2^k-1$ times of controlled $U$.
\end{fact}

The main cost of QPE is the repeated  implements of controlled $U$, with other parts implemented efficiently. As we have seen, to achieve a precision $2^{-m}$, we need to implement $U^{2^j}$ for $j$ from $1$ to $m$. This means that an inverse exponential precision generally requires exponential time complexity to achieve. For our application later, however, we will use a specific $U$ such that the cost could be significantly reduced.

\section{Distinguishing energy-conserving ensemble from Haar-random unitaries}
\label{ap:distinguish-energy-conserving-form-RU}
To motivate the energy-conserving PRU, we begin by proving any energy-conserving unitary ensemble does not form PRU (with no energy constraint). The intuition for this result is to simply check the energy. Let $\ket{\psi}$ be an arbitrary input state, one can estimate the energy expectation $\mathbb{E}_U\big[\bra{\psi}U^\dagger H U\ket{\psi}\big]$ for $U$ drawn from some random unitary ensemble. When $U$ is drawn from Haar-random unitaries, $\mathbb{E}_U\big[\bra{\psi}U^\dagger H U\ket{\psi}\big]\propto \Tr{H}$, which is the energy expectation of infinite temperature state. When $U$ is drawn from any ensemble $\{U\}$ satisfying $[U,H]=0,\forall U$, $\mathbb{E}_U\big[\bra{\psi}U^\dagger H U\ket{\psi}\big]=\bra{\psi}H\ket{\psi}$. Thus, a distinguishing algorithm for $\{U\}$ and Haar-random unitaries can be constructed by checking the average energy. 

More precisely, we consider the local sparse Hamiltonian: $H$ is given by the sum of $k$-local terms, where each qubit is acted by at most $d$ $k$-local terms. This is often referred to as $k$-local Hamiltonians with bounded degree $d$. Existing results~\cite{huang2022learning, anshu2021improved, brandao2013product} demonstrated that for such Hamiltonians, a product state with significant energy deviation from the Haar-random ones can be efficiently constructed.
\begin{lemma}
    [Product state with large deviation. Corollary 2 of~\cite{huang2022learning} ]\label{lem:product-state-large-deviation}
    Given an $n$-qubit $k$-local Hamiltonian $H=\sum_{P\in\{I,X,Y,Z\}^n:|P|\leq k}\alpha_P P$ with bounded degree $d$, $|\alpha_P|\leq 1$ for all $P$, and $k=O(1)$. There is a random algorithm that runs in time $O(nd)$ and produces either a random maximizing state $\ket{\psi}=\ket{\psi_1}\otimes\cdots\ket{\psi_n}$ satisfying
    \begin{align}
        \mathbb{E}_{\ket{\psi}}\big[\bra{\psi}H\ket{\psi}\big]
        \geq \mathbb{E}_{\ket{\phi}:\operatorname{Haar}}\big[\bra{\phi}H\ket{\phi}\big]+\frac{C}{\sqrt{d}}\sum_{P\neq I}|\alpha_P|,
    \end{align}
    or a random minimizing state $\ket{\psi}=\ket{\psi_1}\otimes\cdots\ket{\psi_n}$ satisfying
    \begin{align}
        \mathbb{E}_{\ket{\psi}}\big[\bra{\psi}H\ket{\psi}\big]
        \leq\mathbb{E}_{\ket{\phi}:\operatorname{Haar}}\big[\bra{\phi}H\ket{\phi}\big]-\frac{C}{\sqrt{d}}\sum_{P\neq I}|\alpha_P|
    \end{align}
    for some constant $C$.
\end{lemma}

This result enables explicit construction of an efficient distinguishing algorithm.
\begin{theorem}
    [Energy-conserving unitary ensemble does not form PRU]\label{thm:distinguish-energy-conserving-form-RU}
    Given an $n$-qubit $k$-local Hamiltonian with bounded degree $d$ satisfying $k,d=O(1)$.
    Any random unitary ensemble $\{U\}$  satisfying $[U,H]=0$ for all $U$ can be efficiently distinguished from Haar-random ensemble.
\end{theorem}
\begin{proof}
    To prove this, we demonstrate that there is a universal efficient algorithm to distinguish $\{U\}$ from the Haar-random ensemble.

    \paragraph{The algorithm.} The algorithm takes two steps. First, estimating $E_{\operatorname{Haar}}=\mathbb{E}_{\ket{\phi}:\operatorname{Haar}}\big[\bra{\phi}H\ket{\phi}\big]$ to a high precision by simulating pseudo-random states $\ket{\phi}$. Second, estimating 
    \begin{align}
        E_{U}=\mathbb{E}_{\ket{\psi},U} \big|\big[\bra{\psi}U^\dagger HU\ket{\psi}\big]-E_{\operatorname{Haar}}\big|
    \end{align}
    for $\ket{\psi}$ sampled by Lemma~\ref{lem:product-state-large-deviation}, and for $U$ sampled from the unitary ensemble to be distinguished. If $E_U$ is smaller than  $\frac{C}{2\sqrt{d}}\sum_{P\neq I}|\alpha_P|$, output $\acc$. Otherwise output $\rej$.

    \paragraph{Performance guarantee.} When $U$ is drawn from the Haar-random ensemble,
    \begin{align}
        E_U&=\mathbb{E}_{\ket{\psi},U} \big|\big[\bra{\psi}U^\dagger HU\ket{\psi}\big]-\mathbb{E}_{\ket{\phi}:\operatorname{Haar}}\big[\bra{\phi}H\ket{\phi}\big]\big|
        \notag\\
        &\leq \mathbb{E}_{\ket{\psi}}\Big[
        \sqrt{
        \mathbb{E}_{\ket{\phi}:\operatorname{Haar}}[\bra{\phi}H\ket{\phi}\bra{\phi}H\ket{\phi}]-
        \mathbb{E}_{\ket{\phi}:\operatorname{Haar}}\big[\bra{\phi}H\ket{\phi}\big]^2
        }
        \Big]\notag\\
        &=\mathbb{E}_{\ket{\psi}}\Bigg[
        \sqrt{\frac{1}{2^n(2^n+1)}\Big(\Tr{H^2}-\frac{\Tr{H}^2}{2^n}\Big)
        }
        \Bigg]\notag\\
        &=\sqrt{\frac{1}{2^n(2^n+1)}}\norm{
        H-\frac{\Tr{H}}{2^n}\operatorname{I}
        }_F\notag\\
        &\leq \sqrt{\frac{1}{2^n+1}}\norm{
        H-\frac{\Tr{H}}{2^n}\operatorname{I}
        }_\infty\notag\\
        &=\CO(n/2^{\frac{n}{2}}).
    \end{align}

     When $U$ is drawn from any energy-conserving unitary ensemble,
     \begin{align}
         E_U=\mathbb{E}_{\ket{\psi},U} \big|\big[\bra{\psi}U^\dagger HU\ket{\psi}\big]-\mathbb{E}_{\ket{\phi}:\operatorname{Haar}}\big[\bra{\phi}H\ket{\phi}\big]\big|
         \geq \frac{C}{\sqrt{d}}\sum_{P\neq I}|\alpha_P|=\Omega(1)
     \end{align}
     by Lemma~\ref{lem:product-state-large-deviation}.  Therefore, there is an exponential separation for $E_U$ between these two cases. Estimating $E_{\operatorname{Haar}}$ and $E_{U}$ up to an inverse polynomial precision suffices to distinguish these two cases.

     \paragraph{Runtime.} Since $H$ is $k$-local with bounded degree $d$, estimating $\bra{\psi}H\ket{\psi}$ up to $\epsilon$ error requires $O(\log(n)/\epsilon^2)$ samples of $\ket{\psi}$ using, e.g., classical shadow tomography~\cite{huang2020predicting}. While $\norm{H}_\infty=O(n)$, estimating $\mathbb{E}_{\ket{\psi}}\big[\bra{\psi}H\ket{\psi}\big]$ up to $\epsilon$ precision requires $O(n^2/\epsilon^2)$ samples of $\bra{\psi}H\ket{\psi}$ by Chernoff bound. Combining these two, one can estimate $E_{\operatorname{Haar}}$ and $E_U$ to an inverse polynomial precision by querying $U$ polynomial times.
\end{proof}

\section{Energy-conserving PRU for commuting random Hamiltonians}
\label{ap:pru-commuting-ham}

In this section, we construct the energy-conserving PRU for commuting Hamiltonians with random coefficients drawn from Gaussian distribution. Roughly speaking, our construction takes three steps: first, introduce ancillary qubits and apply QPE to rotate into energy eigenbasis. Second, apply quantum-secure pseudo-random function to ancilla qubits. This amounts to add a pseudo-random phase to each energy eigenstates. Finally, applying QPE$^\dagger$ to rotate back to original basis. Pseudo-random phase unitaries, which are energy-conserving PRU for non-degenerated Hamiltonians, are constructed this way. Crucially, commuting Hamiltonians enable the application of $e^{iHt}$ for $t=\mathsf{exp}(n)$, which makes it possible to reserve every energy eigenstate perfectly.

To avoid insufficient covering, we restrict ourselves to local terms that are drawn from a complete set
\begin{definition}
    [Complete set of commuting observables]\label{def:csco}
    A set of observables $\{A_1,\cdots,A_k\}$ is called a complete set of commuting observables, iff the followings hold
    \begin{itemize}
        \item $[A_i,A_j]=0$ for $i\neq j$.
        \item Let $\{\ket{\psi(l)}\}$ be a set of common basis such that $A_i\ket{\psi(l)}=\mu_i(l)\ket{\psi(l)}$ for all $i$, then the ordered tuple $(\mu_1(l),\cdots,\mu_k(l))$ uniquely determine $\ket{\psi(l)}$.
    \end{itemize}
    In words, in the complete set of commuting observables the quantum numbers of all $A_i$s together complete specify the quantum state.
\end{definition}

Ideally, we consider the family of commuting Hamiltonians $H_n=\sum_i\CJ_ih_i$, where $\{h_i\}$ forms a complete set and $\CJ_i$s are drawn from Gaussian distribution. However, to facilitate computation on digital computers, one needs to digitalize every continuous variable. Therefore, we use the digitalized Gaussian random variable as coefficients.
\begin{definition}
    [Digitalization of Gaussian random variable]\label{def:digitalization}
    Let $R,\delta>0$ such that $R/\delta$ is an odd integer. A $(R,\delta)$-digitalization of Gaussian random variable $\CJ$ is denoted as $\CJ'=Q_{R,\delta}(\CJ)$. $\CJ'$ is a random variable taking the values $-R,-R+\delta,\cdots,-\delta,0,\delta,\cdots,R-\delta,R$ such that
    \begin{align}
        \CJ'=\begin{cases}
            k\delta,\quad \CJ\in\Big(\Big(k-\frac{1}{2}\Big)\delta,\Big(k+\frac{1}{2}\Big)\Big]\\
            -R,\quad \CJ\in \Big(-\infty, -R+\frac{1}{2}\delta\Big]\\
            R,\quad \CJ\in\Big(R-\frac{1}{2}\delta,\infty\Big).
        \end{cases}
    \end{align}
    A $(R,\delta)$-digitalization takes $O(\log_2(R/\delta))$ number of bits to represent a Gaussian random variable.
\end{definition}

\begin{definition}
    [Random commuting local Hamiltonian ensemble]\label{def:commuting-ham}
    Fix a locality parameter $d\ge 1$ and a set of precision parameter $(R,\delta)$. Lett $\mathcal{T}$ be a finite set of Hermitian operators with dimensions at most $2^d$ and no spectral degeneracy.   For each system size $n$, pick up $M(n)=\Theta(n)$ number of template observables $h^{(\alpha_i)}\in\mathcal{T}$, $1\leq i\leq M(n)$, and assign local supports $S_i\subseteq [n]$ with $|S_i|=|\operatorname{supp}(h^{(\alpha_i)})|$ for each $h^{(\alpha_i)}$, such that $\{h_i\} :=\{ \bigl(h^{(\alpha_i)}\bigr)_{S_i}\otimes I_{[n]\setminus S_i}\}$ is a complete set of commuting observables.
    Draw i.i.d. coefficients $\CJ_i\sim Q_{R,\delta}(\mathcal N(0,1))$, independent of the $h_i$, the commuting random local Hamiltonian is defined as $H_n = \sum_{i=1}^{M(n)} \CJ_i h_i$.
\end{definition}
As a simple example, we can take $\CT=\{\sigma_z\}$ and $\{h_1,h_2,\cdots h_n\}=\{\sigma_z^1,\sigma_z^2,\cdots \sigma_z^n\}$.

Throughout the rest of this section, if $\lambda\in\{0,1\}^*$ is a bitstring, we will denote its corresponding decimal number (i.e, $\sum_{i=0}^{|\lambda|-1}2^{i}\lambda_{|\lambda|-i}$) be $\hat{\lambda}$.

 \subsection{Smallest energy gap}

In this subsection, we prove that random commuting Hamiltonians have at least exponentially small energy gap with high probability, enabling an accurate discrimination of eigenstates by QPE with polynomial number of ancillas.
 
To warm up, we first prove the claim first for random commuting Hamiltonians with coefficients drawn from the continuous Gaussian distribution.
\begin{fact}[Convolution of Gaussian distributions]
    \label{fact:convolution-gaussian}
    Let $X_i\sim \CN(\mu_i,\sigma_i)$, $1\leq i\leq n$ be Gaussian random variables. Then the random variable $\sum_{i=1}^n\alpha_iX_i$ obeys distribution $\CN(\sum_{i=1}^n\alpha_i\mu_i,\sum_{i=1}^n|\alpha_i|\sigma_i)$ for any $\alpha_1,\cdots\alpha_n\in \RR$.
\end{fact}

\begin{lemma}
    [Smallest energy gap of random commuting Hamiltonians]
    \label{lem:smallest-gap-random-commuting}
    Let $H_n = \sum_{i=1}^{M(n)} \CJ_i h_i$ be a random commuting Hamiltonian with $\CJ_i\overset{\text{i.i.d.}}{\sim}\CN(0,1)$. $\{\lambda_i\}_{1\leq i\leq 2^n}$ labels the eigenvalues of $H_n$. Then for sufficiently large $\beta>0$, 
    \begin{align}
        \pr\Big(\min_{(k,k')\in[2^{M(n)}]^2,k\neq k'}\{|\lambda_k-\lambda_{k'}|\}
        \leq  2^{-\beta n}
        \Big)=O(2^{-(\beta n-2M(n)}).
    \end{align}
\end{lemma}
\begin{proof}
    Let $\CE=\bigcup_{h\in\CT}\spec(h)$. 
    There exists a set of eigenstates $\ket{\psi(\mathbf{e})}$ completely specified by ordered tuple $\mathbf{e}=(e_1,\cdots, e_{M(n)})$ for $e_i\in\CE$, whose eigenvalues are $\lambda(\mathbf{e})=\sum_{i=1}^{M(n)}\CJ_ie_i$. As a result,
    \begin{align}
        \lambda_{k}-\lambda_{k'}=\Big|\sum_{i=1}^{M(n)}\CJ_i e_{k_i}-\sum_{i=1}^{M(n)}\CJ_i e_{k'_i}\Big|
        =\Big|\sum_{i=1}^{M(n)}\CJ_i(e_{k_i}-e_{k'_i})\Big|
        =:\big|E(\mathbf{e}_k,\mathbf{e}_{k'}) \big|
    \end{align}
    for $\mathbf{e}_k\neq\mathbf{e}_{k'}$.
    
     Due to Fact.\ref{fact:convolution-gaussian}, 
    \begin{align}
        E(\mathbf{e}_k,\mathbf{e}_{k'})\sim \CN(0,\sum_{i=1}^{M(n)}|e_{k_i}-e_{k'_i}|).
    \end{align}
    As a result, for any constant $\beta>0$,
    \begin{align}
        \pr\big(|E(\mathbf{e}_k,\mathbf{e}_{k'})|\leq  2^{-\beta n}\big)=
        2\operatorname{erf}\bigg(\frac{ 2^{-\beta n}}{\sum_{i=1}^{M(n)}|e_{k_i}-e_{k'_i}|}\bigg)\leq
        \frac{4}{\sqrt{\pi}}\frac{ 2^{-\beta n}}{\sum_{i=1}^{M(n)}|e_{k_i}-e_{k'_i}|}\leq
        \frac{4\cdot  2^{-\beta n}}{\sqrt{\pi}e_{\min}}.
    \end{align}
    In the last inequality we introduce $e_{\min}=\min_{e,e'\in\CE,e\neq e'}\{|e-e'|\}$ and use the fact that $\mathbf{e}_k\neq\mathbf{e}_{k'}$. Due to union bound,
    \begin{align}
        \pr\Big(\min_{\{k,k'\}\subset[2^{M(n)}]}\{|\lambda_k-\lambda_{k'}|\}
        \leq  2^{-\beta n}\Big)
        &\leq \sum_{\{k,k'\}\in[2^{M(n)}]}\pr\big(|
        \lambda_k-\lambda_{k'}|\leq  2^{-\beta n}
        \big)\notag\\
        &\leq \frac{2^{2M(n)+1-\beta n}}{\sqrt{\pi}e_{\min}}.
    \end{align}
    Since $M(n)=\Theta(n)$, we can always choose a $\beta$ such that RHS is smaller.
\end{proof}

Built upon Lemma~\ref{lem:smallest-gap-random-commuting}, we prove the same property holds for digitalized coefficients.
\begin{lemma}
    [Smallest gap under digitalization]\label{lem:small-gap-digitalization}
    Let $H_n = \sum_{i=1}^{M(n)} \CJ_i h_i$ be a random commuting Hamiltonian with $\CJ_i$ being $(R,\delta)$-digitalizatio of $\CN(0,1)$. 
    $\{\lambda_i\}_{1\leq i\leq 2^n}$ labels the eigenvalues of $H_n$. Then for $R=\poly{n}$, $\delta=O(1/(e^{\beta n}M(n)))$, and  sufficiently large $\beta>0$, 
    \begin{align}
        \pr\Big(\min_{(k,k')\in[2^{M(n)}]^2,k\neq k'}\{|\lambda_k-\lambda_{k'}|\}
        \leq  2^{-\beta n}
        \Big)=\negl{n}.
    \end{align}
    The digitalization using $\poly{n}$ number of bits in total.
\end{lemma}
\begin{proof}
Let $\CE=\bigcup_{h\in\CT}\spec(h)$. For $\mathbf{e}\neq \mathbf{e}'\in\CE^{M(n)}$, let
\begin{align}
    f_n^{(\mathbf{e},\mathbf{e}')}(x_1,\cdots,x_{M(n)})=\Big|
        \sum_{i=1}^{M(n)}x_i (e_i-e'_i)
        \Big|.
\end{align}
$f_n^{(\mathbf{e},\mathbf{e}')}$ satisfies
\begin{align}
    \Big|f_n^{(\mathbf{e},\mathbf{e}')}(x_1,\cdots,x_{M(n)})
    -f_n^{(\mathbf{e},\mathbf{e}')}(x'_1,\cdots,x'_{M(n)})
    \Big|&\leq \Bigg|
        \sum_{i=1}^{M(n)}(x_i-x'_i) (e_i-e'_i)
        \Bigg|\notag\\
        &\leq e_{\max}\sqrt{M(n)}\norm{\mathbf{x}-\mathbf{x}}_2,
\end{align}
where $e_{\max}=\max_{e,e'\in\CE}|e-e'|$, and $\mathbf{x}=(x_1,\cdots,x_{M(n)})$. As a result, $f_n^{(\mathbf{e},\mathbf{e}')}$ is $e_{\max}\sqrt{M(n)}$-Lipschitz for all $\mathbf{e}\neq \mathbf{e}'$.
\begin{align}
f_n(x_1,x_2,\cdots,x_{M(n)})=\min_{\mathbf{e},\mathbf{e}'\in\CE^{M(n)},\mathbf{e}\neq\mathbf{e}'}
f_n^{(\mathbf{e},\mathbf{e}')}(x_1,\cdots,x_{M(n)})
\end{align}
is also a $e_{\max}\sqrt{M(n)}$-Lipschitz function. 

Let $\CJ_i'$ be the random Gaussian seed to define $\CJ_i$.
Using the fact that if $|\CJ'_i|<R-\delta/2$ for all $i$, then $\norm{(\CJ_1,\cdots,\CJ_{M(n)})-(\CJ'_1,\cdots,\CJ_{M(n)}')}_2\leq \delta\sqrt{M(n)}/2$, we have
\begin{align}\label{eq:digitalization:step2}
    \pr\Big(\min_{(k,k')\in[2^{M(n)}]^2,k\neq k'}\{|\lambda_k-\lambda_{k'}|\}
        \leq  2^{-\beta n}
        \Big)&\leq
    \pr\Big(f_n(\CJ_1,\CJ_2,\cdots,\CJ_{M(n)})\leq  2^{-\beta n}\Big)\notag\\
    &= \pr\Big(f_n(\CJ_1,\CJ_2,\cdots,\CJ_{M(n)})\leq  2^{-\beta n}~,~\exists |\CJ'_i|\geq R-\delta/2\Big)\notag\\
    &\quad +
    \pr\Big(f_n(\CJ_1,\CJ_2,\cdots,\CJ_{M(n)})\leq  2^{-\beta n}~,~
    \forall i,|\CJ'_i|<R-\delta/2\Big)\notag\\
    &\leq \pr\big(\exists |\CJ'_i|\geq R-\delta/2\big)\notag\\
    & \quad+\pr\Big(f_n(\CJ'_1,\CJ'_2,\cdots,\CJ'_{M(n)})\leq  2^{-\beta n}+
    \delta e_{\max}M(n)/2
    \Big).
\end{align}
In the last inequality, we use $\pr(|f(x)|\leq a)\leq \pr(|f(y)|\leq a +L\norm{x-y}_2)$ if $f$ is $L$-Lipschitz and $x$ depends on $y$. For the first term, using union bound,
\begin{align}
    \pr\big(\exists |\CJ_i|\geq R-\delta/2\big)\leq M(n)\sqrt{\frac{2}{\pi}}
    e^{-(R-\delta/2)^2/2}
\end{align}
For the second term, taking $\delta\leq 1/(2^{\beta n}M(n)e_{\max})$ and using Lemma~\ref{lem:smallest-gap-random-commuting},
\begin{align}
    \pr\Big(f_n(\CJ_1,\CJ_2,\cdots,\CJ_{M(n)})\leq  2^{-\beta n-1}+
    \delta e_{\max}M(n)/2
    \Big)&\leq  \pr\Big(f_n(\CJ_1,\CJ_2,\cdots,\CJ_{M(n)})\leq  2^{-\beta n}\Big)\notag\\
    &=O(2^{-(\beta n-2M(n))}).
\end{align}
    As a result, taking $R=\poly{n}$, $\delta = \mathcal{O}(1/(2^{\beta n}M(n))$ and sufficiently large $\beta$, RHS of Eq.~\eqref{eq:digitalization:step2} is $\negl{n}$. The number of bits used for digitalization is $M(n)\cdot O(\log_2(R/\delta))=\poly{n}$.
\end{proof}

\subsection{Random phase unitary with quantum phase estimation}
\begin{figure}
    \centering
    \includegraphics[width=1.0\linewidth]{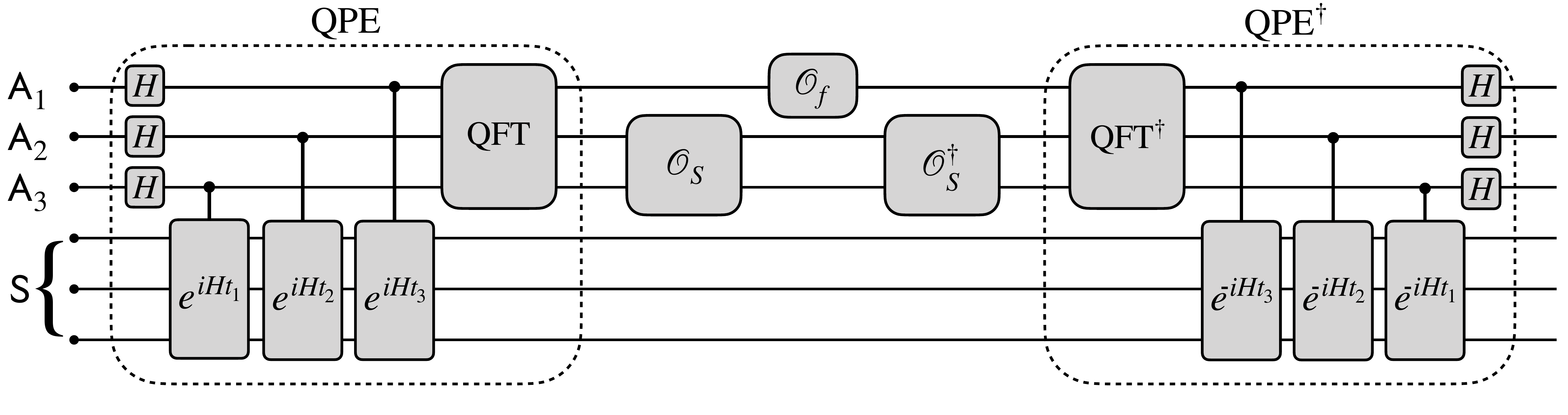}
    \caption{Entire construction of energy-conserving PRU for random commuting Hamiltonians.}
    \label{fig:pru-commuting}
\end{figure}

In this subsection, we describe how to construct an arbitrary phase unitary $U_\theta=\sum_{k=1}^{2^n}e^{i\theta_k}\ketbra{\lambda_k}{\lambda_k}$,
given a unitary $U=\sum_{k=1}^{2^n}e^{i2\pi\lambda_k}\ketbra{\lambda_k}{\lambda_k}$ with significantly large energy gap, using oracles.  We will promote this to the pseudo-random ensemble and discuss the implementation of oracles in the next subsection.

We denote $\sfS$ as the $n$-qubit system register that $U$ residues, and introducing an ancillary register $\sfA$ to facilitate the construction.
For convenience, we split the ancillary register into three disjoint subset, $\sfA=\sfA_1\cup\sfA_2\cup\sfA_3$, with sizes $m_1,m_2,m_3$, respectively. Our construction uses the following three oracles:
\begin{itemize}
    \item A QPE oracle $\CO_{\qpe}$ that performs quantum phase estimation for $U$, on system $\sfS$ and ancilla $\sfA$.
    \item A random offset oracle $\CO_S$. We choose $\CO_S$ to act only on $\sfA_2\cup\sfA_3$, i.e., $S\in\{0,1\}^{m_2+m_3}$, and $\CO_S:\ket{x}\rightarrow \ket{x+0_{\sfS\cup\sfA_1}\cdot S_{\sfA_2\cup\sfA_3}\mod 2^{m_1+m_2+m_3}}$, where we use $0_{\sfS\cup\sfA_1}\cdot S_{\sfA_2\cup\sfA_3}$ to denote the concatenation to two bitstrings $0_{\sfS\cup\sfA_1}$ and $S_{\sfA_2\cup\sfA_3}$.
    \item A phase oracle acting on $\sfA_1$. Let $f:\{0,1\}^{m_1}\rightarrow [0,2\pi)$, then $\CO_f:\ket{x}\rightarrow e^{if(x_{\sfA_1})}\ket{x}$.
\end{itemize}
We will demonstrate that $U_{f,S}=\CO_{\qpe}^\dagger \CO_{S}^\dagger \CO_f \CO_S \CO_{\qpe}$ approximates $V_{f,S}=\sum_{i=1}^{2^n}e^{if(\widetilde{\lambda}^S_k)}\ketbra{\lambda_k}{\lambda_k}$ well, with $\widetilde{\lambda}_k^S$ the following coarse-grained value of $\lambda_k$:
\begin{align}
        \widetilde{\lambda}_k^S: \hat{\widetilde{\lambda}}^S_k=\argmax_{N/2^{m_1},N\in\mathbb{Z}}\Big\{\lambda_k+\frac{\hat{S}'}{2^{m_1+m_2+m_3}}
        \geq \frac{N}{2^{m_1}}
        \Big\},\quad S'=0_{\sfA_1}S_{\sfA_2\cup\sfA_3}.
\end{align}
When the gap of $U$ is larger than $2^{-m_1}$, different $\lambda_k$ corresponds to different $\widetilde{\lambda}_k^S$. Thus every $\ket{\lambda_k}$ could have different phases. 
We will later take $U=e^{iH}$ to implement the energy-conserving PRU. The entire circuit is sketched in Fig~\ref{fig:pru-commuting}.

We emphasize that the coarse-grained resolution of oracle $\CO_f$ and the random offset oracle $\CO_S$ are necessary. At its core, the QPE algorithm performs the following transformation: $\ket{\lambda_k}\ket{0_{\sfA}}\rightarrow \ket{\lambda_k}\ket{\phi_k}$, where $\ket{\phi_k}$ encodes a binary approximation of $\lambda_k$. However, since QPE is of finite precision, $\ket{\phi_k}$ as non-negligible weight on several neighboring bitstrings on $\sfA$.
We choose $\CO_f$ to  depend only on $m_1$ coarse bits of the phase register $\sfA$, ignoring the remaining ``fine" bits on $\sfA_2\cup\sfA_3$. Any $\ket{\phi_k}$ has approximately the same coarse-bit over its support (with high probability). This ensures $\CO_{\qpe}^\dagger$ can coherently uncompute the phase register to return it to $\ket{0_\sfA}$.

The random offset $\CO_S$ eliminates the principle worst case: if $\lambda_k$ lies near the coarse-bin boundary, i.e., $\lambda_k\approx N/2^{m_1}$ for integer $N$, then $\ket{\phi_k}$ straddles two different coarse values, spoiling the uncomputation. By adding a uniformly random shift to the fine bits, $\CO_S$ moves $\ket{\phi_k}$ away from boundary with high probability, thus rescuing the uncomputation.

We begin by showing that $\CO_S$ can prevent  ``hitting the boundary of coarse-bin''.
\begin{lemma}
    [Random offset prevents hitting the boundary]\label{lem:qpe-offset}
    Let $U$ be a $n$-qubit untary acting on the system register $\sfS$ with the set of eigenvalues and eigenstates $\{e^{i2\pi \lambda_i},\ket{\lambda_i}\}_{1\leq i\leq 2^n},\lambda_i\in[0,1)$, and they satisfy $|\lambda_i-\lambda_j|>2^{-\beta n}$ for all $i\neq j$. Introducing ancilla register $\sfA=\sfA_1\cup\sfA_2\cup\sfA_3$ such that $m_1>\beta n$.
    Denote $\CO_{\qpe}$ to be the QPE oracle performed on system $\sfS$ and ancilla $\sfA$, and $O_S$ to be the random offset oracle acting on $\sfA_2\cup\sfA_3$. Then for any state $\ket{\psi}$ on $\sfS$,
    \begin{align}
        \pr_S\Bigg(
        \norm{\exists i\in[2^n],\prod_{m}\CO_{S}\CO_{\qpe}\ket{\lambda_i}
        \ket{0_{\sfA}}}_2> 2^{-m_3/2}
        \Bigg)  \leq 2^{-(m_1-n-2)},
    \end{align}
    where
    \begin{align}
        \prod_{m}=\operatorname{I}_{\sfS}\otimes\sum_{x\in\{0,1\}^{m_1}}
        \sum_{b\in\{0,1\}^{m_1+m_2+m_3},|\hat{b}/2^{m_1+m_2+m_3}-\hat{x}/2^{m_1}|<2^{-(m_1+m_2)}}
        \ketbra{b_\sfA}{b_\sfA}
    \end{align}
    is the projector to bitstrings near the boundary (See Fig~\ref{fig:illustration-projector}).
\end{lemma}
\begin{proof}
    From Fact~\ref{fact:accuracy-qpe},
        \begin{gather}
            \norm{\prod_{\lambda_i}\CO_{\qpe}\ket{\lambda_i}
        \ket{0_{\sfA}}}_2\leq 2^{-m_3/2},\notag\\
        \prod_{\lambda_i}=\operatorname{I}_\sfS\otimes\sum_{b\in\{0,1\}^{m_1+m_2+m_3},|\hat{b}/2^{m_1+m_2+m_3}-\lambda_i|>2^{-(m_1+m_2)}}
        \ketbra{b_\sfA}{b_\sfA}.
        \end{gather}
    This leads to
    \begin{gather}
        \norm{\prod_{\lambda_i,S}\CO_S\CO_{\qpe}\ket{\lambda_i}
        \ket{0_{\sfA}}}_2\leq 2^{-m_3/2},\notag\\
        \prod_{\lambda_i,S}=\operatorname{I}_\sfS\otimes\sum_{b\in\{0,1\}^{m_1+m_2+m_3},|\hat{b}/2^{m_1+m_2+m_3}-\lambda_i-\hat{S}'/2^{m_1+m_2+m_3}|>2^{-(m_1+m_2)}}
        \ketbra{b_\sfA}{b_\sfA},
    \end{gather}
    where $S'=0_{\sfA_1}S_{\sfA_2\cup\sfA_3}$.
    When uniformly sampling $S$, $\hat{S}'/2^{m_1+m_2+m_3}$ is uniformly distributed over $[0,2^{-m_1})$.
    As a result,
    \begin{align}
        \pr_S\Bigg(
        \norm{\prod_{m}\CO_{S}\CO_{\qpe}\ket{\lambda_i}
        \ket{0_{\sfA}}}_2> 2^{-m_3/2}
        \Bigg) &\leq
        \pr_S\Big(
        \operatorname{im}\Big(\prod_{m}\Big)\not\subseteq
        \operatorname{im}\Big(\prod_{\lambda_i,S}\Big)
        \Big)\notag\\
        &\leq \frac{2(2^{-(m_1+m_2)}+2^{-(m_1+m_2)})}{2^{-m_1}}\notag\\
        &=2^{-(m_2-2)}
    \end{align}
    The second to last line comes from sweeping over the region of length $2^{-m_1}$ that contains $\lambda_i$. See Fig~\ref{fig:illustration-projector} for an illustration. Then a union bound yields the desired inequality.
\end{proof}

\begin{figure}[t!]
    \centering
    \includegraphics[width=0.5\linewidth]{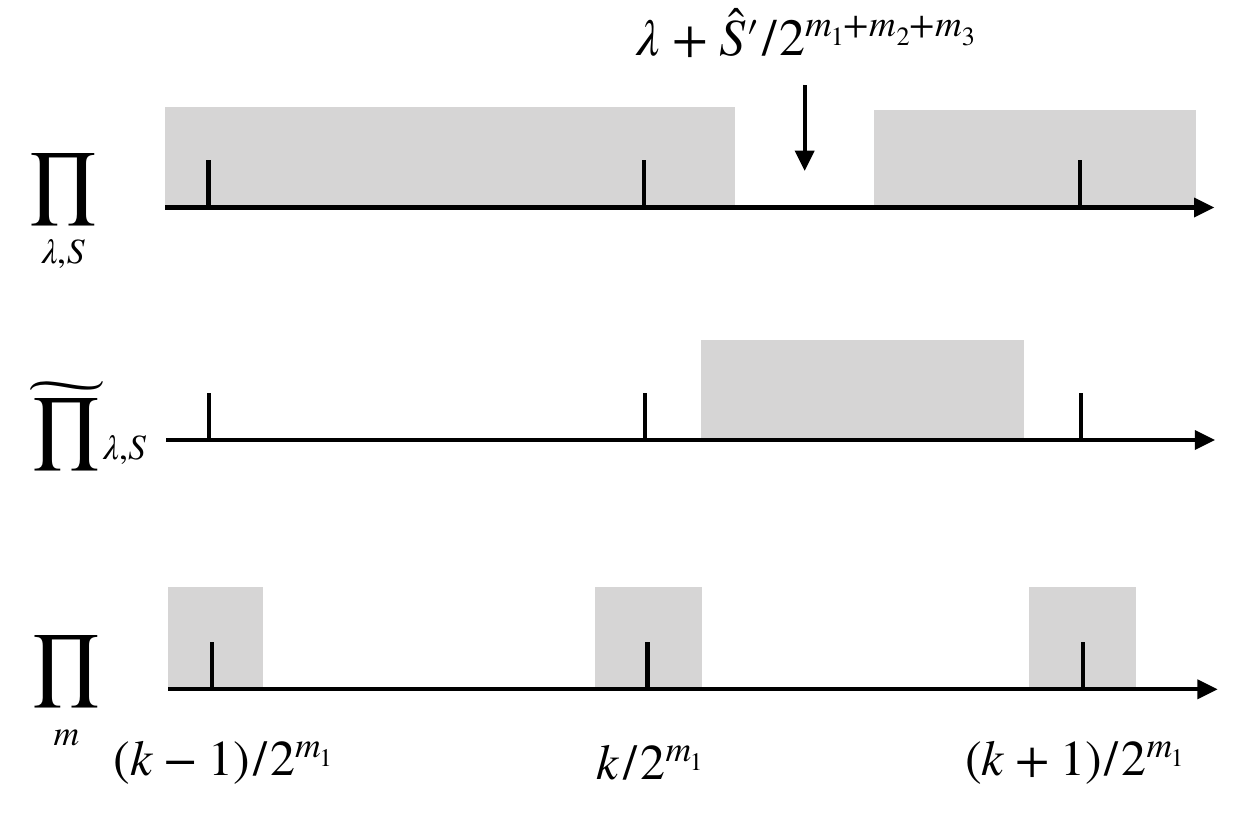}
    \caption{Definitions of projectors used in Lemma~\ref{lem:qpe-offset} and Lemma~\ref{lem:phase-unitary}.}
    \label{fig:illustration-projector}
\end{figure}

Based on Lemma~\ref{lem:qpe-offset}, we can prove that $U_{f,S}$ approximates the phase unitary well, using the number of ancillas $m_1,m_2,m_3=O(n)$.
\begin{lemma}
    [Constructing arbitrary phase unitary]\label{lem:phase-unitary}
    Let $U$ be a $n$-qubit untary acting on the system register $\sfS$ with the set of eigenvalues and eigenstates $\{e^{i2\pi \lambda_i},\ket{\lambda_i}\}_{1\leq i\leq 2^n},\lambda_i\in[0,1)$, and they satisfy $|\lambda_i-\lambda_j|>2^{-\beta n}$ for all $i\neq j$. Introducing ancilla register $\sfA=\sfA_1\cup\sfA_2\cup\sfA_3$ such that
    \begin{align}
        m_1>\beta n,\quad m_2>n+2,\quad m_3>2n,\quad m_1>m_3.
    \end{align}
     Denote $\CO_{\qpe}$ be the oracle to implement the QPE for $U$, acting on system $\sfS$ and ancilla $\sfA$. 
    Denote $\CO_S$ to be the random offset oracle acting on $\sfA_2\cup\sfA_3$, and $\CO_f$ be the oracle that applies the phase function $\ket{x}\rightarrow e^{if(x)}\ket{x} $ for $f:\{0,1\}^{m_1}\rightarrow [0,2\pi)$ on  $\sfA_1$. 
    Then $U_{f,S}=\CO_{\qpe}^\dagger \CO_S^\dagger \CO_f\CO_S \CO_{\qpe} $ satisfies that for any quantum state $\ket{\psi}$ on $\sfS$,
    \begin{align}
        \pr_S\Bigg(\norm{
        U_{f,S}\ket{\psi}\otimes \ket{0_{\sfA}}-\sum_{i=1}^{2^n}e^{if(\widetilde{\lambda}_k^S)}
        \bra{\lambda_ik}\psi\rangle \ket{\lambda_k}\otimes\ket{0_{\sfA}}
        }_2\leq2^{-(m_3/2+2+n)} \Bigg)> 1-2^{-(m_2-n-2)}.
    \end{align}
\end{lemma}
\begin{proof}
    Denote,
    \begin{align}
        \CO_{S}\CO_{\qpe}\ket{\lambda_k}
        \ket{0_{\sfA}}=\ket{\lambda_k}\ket{\phi_k}.
    \end{align}
    From Fact~\ref{fact:accuracy-qpe},
    \begin{gather}
        \norm{\prod_{k,S}\ket{\lambda_k}\ket{\phi_k}}_2
        \leq 2^{-m_3/2},\notag\\
        \prod_{k,S}=\operatorname{I}_\sfS\otimes\sum_{b\in\{0,1\}^{m_1+m_2+m_3},|\hat{b}/2^{m_1+m_2+m_3}-\lambda_k-\hat{S}'/2^{m_1+m_2+m_3}|>2^{-(m_1+m_2)}}
        \ketbra{b_\sfA}{b_\sfA}.
    \end{gather}
    
    Define
    \begin{align}
        \widetilde{\prod}_{k,S}=\operatorname{I}_\sfS\otimes\sum_{b\in\{0,1\}^{m_1+m_2+m_3},2^{-(m_2+m_3)}<|\hat{b}/2^{m_1+m_2+m_3}-\widetilde{\lambda}^S_k|<2^{-m_1}-2^{-(m_2+m_3)}}
        \ketbra{b_\sfA}{b_\sfA}.
    \end{align}
    We have
    \begin{align}
        \CO_f\ket{\lambda_k}\ket{\phi_k}&=
        \CO_f \Big(\widetilde{\prod}_{k,S}\ket{\lambda_k}\ket{\phi_k}+
        \Big(1-\widetilde{\prod}_{k,S}\Big)\ket{\lambda_k}\ket{\phi_k}
        \Big) \notag\\
        &=e^{if(\widetilde{\lambda}^S_k)}\widetilde{\prod}_{k,S}\ket{\lambda_k}\ket{\phi_k}+\CO_f\Big(1-\widetilde{\prod}_{k,S}\Big)\ket{\lambda_k}\ket{\phi_k}.
    \end{align}\label{eq:phase-unitary:1}
    
    Define
    \begin{align}
        \prod_{m} = \operatorname{I}_\sfS\otimes\sum_{x\in\{0,1\}^{m_1}}
        \sum_{b\in\{0,1\}^{m_1+m_2+m_3},|\hat{b}/2^{m_1+m_2+m_3}-\hat{x}/2^{m_1}|<2^{-(m_2+m_3)}}
        \ketbra{b_\sfA}{b_\sfA}.
    \end{align}
    We assume to sample $S$ such that for all $i\in[2^n]$
    \begin{align}
        \norm{\prod_{m}\ket{\lambda_k}\ket{\phi_k}}_2\leq 2^{-m_3/2}.
    \end{align}
    The probability is at least $1-2^{-(m_2-n-2)}$ according to Lemma~\ref{lem:qpe-offset}. Note that when $m_1>m_3$,
    \begin{align}
         \operatorname{im}\Big(1-\widetilde{\prod}_{k,S}\Big)\subset
         \operatorname{im}\Big(\prod_{k,S}\Big)\bigcup
         \operatorname{im}\Big(\prod_{m}\Big)
    \end{align}
    As a result,
    \begin{align}
        \norm{\Big(1-\widetilde{\prod}_{k,S}\Big)\ket{\lambda_k}\ket{\phi_k}}_2\leq
        \norm{\prod_{k,S}\ket{\lambda_k}\ket{\phi_k}}_2
        +\norm{\prod_{m}\ket{\lambda_k}\ket{\phi_k}}_2
        \leq 2^{-m_3/2+1}.
    \end{align}
    Combining this with Eq.~\eqref{eq:phase-unitary:1}, we have
    \begin{align}
        \norm{\CO_f\ket{\lambda_k}\ket{\phi_k}-e^{if(\widetilde{\lambda}^S_k)}\widetilde{\prod}_{k,S}\ket{\lambda_k}\ket{\phi_k}}_2\leq 2^{-m_3/2+1}.
    \end{align}
    Using triangle inequality,
    \begin{align}
        \norm{O_f\ket{\lambda_k}\ket{\phi_k}-e^{if(\widetilde{\lambda}^S_k)}\ket{\lambda_k}\ket{\phi_k}}_2\leq 2^{-m_3/2+2}.
    \end{align}

    Now we analyze $\CO_{S}\CO_{\qpe}\ket{\psi}
        \ket{0_{\sfA}}=\sum_{k}
        \bra{\lambda_k}\psi\rangle\ket{\lambda_k}\ket{\phi_k}$.
    \begin{align}
        \norm{\CO_f \CO_{S}\CO_{\qpe}\ket{\psi}
        \ket{0_{\sfA}}-\sum_k\bra{\lambda_k}\psi\rangle
        e^{if(\widetilde{\lambda}^S_k)}\ket{\lambda_k}\ket{\phi_k}}_2
        &\leq \sum_k
        \big|\bra{\lambda_k}\psi\rangle\big|
        \norm{\CO_f\ket{\lambda_k}\ket{\phi_k}-e^{if(\widetilde{\lambda}^S_k)}\ket{\lambda_k}\ket{\phi_k}}_2\notag\\
        &\leq 2^{-m_3/2+2+n}.
    \end{align}
    Applying $\CO_{\qpe}^\dagger \CO_S^\dagger$ to the two terms in the LHS, we get the desired inequality.
\end{proof}

\subsection{Constructing energy-conserving PRU}
In this subsection, we turn the oracle-based construction of arbitary phase unitary given in the previous subsection to an explicit construction energy-conserving PRU for random commuting Hamiltonians.

We introduce the following two technical lemmas that are helpful.
\begin{lemma}
    [From state distance to channel distance]\label{lem:state-to-channel}
    Let $\sfS$, $\sfA$ be two registers and $U$, $V$ be unitaries acting on $\sfS\cup\sfA$ and $\sfS$, respectively. If for all state $\ket{\psi}$ on $\sfS$,
    \begin{align}
        \norm{\big(U-V\otimes \operatorname{I}_{\sfA}\big)\ket{\psi}\ket{0_{\sfA}}}_2\leq \epsilon,
    \end{align}
    then the following two channels
    \begin{align}
        \mathsf{U}[\rho]:=\Tr_{\sfA}\big[{U\big(\rho\otimes \ketbra{0_\sfA}{0_{\sfA}}\big)U^\dagger}\big],\quad
        \mathsf{V}[\rho]:=V\rho V^\dagger
    \end{align}
    are close in diamond distance, i.e.,
    \begin{align}
        \norm{\mathsf{U}-\mathsf{V}}_{\diamond}\leq 2\epsilon.
    \end{align}
\end{lemma}
\begin{proof}
    Let $\sfB$ be any ancillary system.
    It is sufficient to prove that, 
    \begin{align}
        \norm{
        (U\otimes\operatorname{I}_{\sfB})(\rho\otimes \ketbra{0_\sfA}{0_{\sfA}})
        (U^\dagger\otimes\operatorname{I}_{\sfB})-
        (V\otimes\operatorname{I}_{\sfA\cup\sfB})(\rho\otimes \ketbra{0_\sfA}{0_{\sfA}})
        (V^\dagger\otimes\operatorname{I}_{\sfA\cup\sfB})
        }_1\leq 2\epsilon 
    \end{align}
    holds for any pure state $\rho=\ketbra{\Psi}{\Psi}$ on $\sfS\cup\sfB$. 

    Denote the Schmidt decomposition of $\ket{\Psi}$ as $\ket{\Psi}=\sum_{i}\mu_i\ket{\psi^i_{\sfS}}\ket{\psi^i_{\sfB}}$, where $\mu_i$ is real positive and $\sum_i\mu_i^2=1$. We have
    \begin{gather}
        \norm{
        (U\otimes\operatorname{I}_{\sfB})(\rho\otimes \ketbra{0_\sfA}{0_{\sfA}})
        (U^\dagger\otimes\operatorname{I}_{\sfB})-
        (V\otimes\operatorname{I}_{\sfA\cup\sfB})(\rho\otimes \ketbra{0_\sfA}{0_{\sfA}})
        (V^\dagger\otimes\operatorname{I}_{\sfA\cup\sfB})
        }_1\notag\\
        =\norm{
        \sum_{ij}\mu_i\mu_j\Bigg(\Big( U \ketbra{\psi^i_{\sfS}}{\psi^j_{\sfS}}\otimes \ketbra{0_\sfA}{0_{\sfA}} U^\dagger
        \Big)-
        \Big( V\otimes \operatorname{I}_{\sfA} \ketbra{\psi^i_{\sfS}}{\psi^j_{\sfS}}\otimes \ketbra{0_\sfA}{0_{\sfA}} V^\dagger
        \otimes \operatorname{I}_{\sfA}\Big)
        \Bigg)\otimes \Big(\ketbra{\psi^i_{\sfB}}{\psi^j_{\sfB}} \Big)
        }_1\notag\\
        \leq \sum_{ij}\mu_i\mu_j
        \norm{
        \Big( U \ket{\psi^i_{\sfS}}\bra{\psi^j_{\sfS}}\otimes \ketbra{0_\sfA}{0_{\sfA}} U^\dagger
        \Big)-
        \Big( V\otimes \operatorname{I}_{\sfA} \ket{\psi^i_{\sfS}}\bra{\psi^j_{\sfS}}\otimes \ketbra{0_\sfA}{0_{\sfA}} V^\dagger
        \otimes \operatorname{I}_{\sfA}\Big)
        }_1\notag\\
        \leq \sum_{ij}\mu_i\mu_j\Big(
        \norm{
        U \ket{\psi^i_{\sfS}}\ket{0_{\sfA}}-V\otimes \operatorname{I}_{\sfA}\ket{\psi^i_{\sfS}}\ket{0_{\sfA}}
        }_2
        +\norm{
        U \ket{\psi^j_{\sfS}}\ket{0_{\sfA}}-V\otimes \operatorname{I}_{\sfA}\ket{\psi^j_{\sfS}}\ket{0_{\sfA}}
        }_2
        \Big)\notag\\
        \leq 2\epsilon\sum_{ij}\mu_i\mu_j\notag\\
        \leq 2\epsilon.
    \end{gather}
    In the third line we use $\norm{\ket{\psi}\bra{\phi}}_1=1$ and triangle inequality. In the fourth line we use 
    \begin{align}
        \norm{\ket{\psi_1}\bra{\phi_1}-\ket{\psi_2}\bra{\phi_2}}_1\leq \norm{\ket{\psi_1}-\ket{\psi_2}}_2+\norm{\ket{\phi_1}-\ket{\phi_2}}_2.
    \end{align}
    In the last line we use Cauchy-Schwartz inequality.
\end{proof}

\begin{lemma}
    [Simulate the random phase unitary]\label{lem:mimic-random-phase}
    Let $U$ be a $n$-qubit untary with the set of eigenvalues and eigenstates $\{e^{i2\pi \lambda_k},\ket{\lambda_k}\}_{1\leq k\leq 2^n},\lambda_k\in[0,1)$. Assume the eigenvalues are well-separated. Define the following two ensemble of unitaries:
    \begin{itemize}
        \item The random phase ensemble, $\mathcal{U}_{\operatorname{RU}}=\{\sum_{k=1}^{2^n}e^{i\theta_k}\ketbra{\lambda_k}{\lambda_k}\}_{\theta}$, $(\theta_1,\cdots\theta_{2^n})\sim\operatorname{Unif[0,2\pi)^{ 2^n}}$.
        \item The random function ensemble, $\mathcal{V}_{\operatorname{RF}}=\{V_{f,S}\}_{f,S}=\{\sum_{k=1}^{2^n}e^{if(\widetilde{\lambda}^S_k)}\ketbra{\lambda_k}{\lambda_k}\}_{f,S}$. $S$ is uniformly sampled from $\{0,1\}^{m_2+m_3}$, with $m_1,m_2,m_3$ satisfying the requirement in Lemma~\ref{lem:phase-unitary}. $f\sim \operatorname{Unif[0,2\pi)^{ 2^n}}$ is a random function.
    \end{itemize}
    Then $\mathcal{U}_{\operatorname{RU}}=\mathcal{V}_{\operatorname{RF}}$.
\end{lemma}
\begin{proof}
    Denote $\mathbb{N}_{N}=[N/2^{m_1},(N+1)/2^{m_1})$, $N=[2^{m_1}-1]$. By the requirement of Lemma~\ref{lem:phase-unitary}, each $\mathbb{N}_N$ contains at most one $\lambda_k$. As a result, different $\widetilde{\lambda}^S_k$ belongs to different $\mathbb{N}_N$. Thus for each $S$, $f(\widetilde{\lambda}^S_k)$ are independent random numbers drawn from $\operatorname{Unif[0,2\pi)}$. This concludes $\mathcal{U}_{\operatorname{RU}}=\mathcal{V}_{\operatorname{RF}}$.
\end{proof}

Based on these preparations, we now demonstrate that any exponentially-accurate approximation of $\mathcal{U}=\{U_{f,S}\}_{f,S}=\{\CO^\dagger_{\qpe}\CO^\dagger_S \CO_f \CO_S \CO_{\qpe}\}_{f,S}$ is indistinguishable from the random phase ensemble $\mathcal{U}_{\operatorname{RU}}$ against any quantum algorithms that makes adaptive queries.
\begin{lemma}
    [Secure against adaptive queires]\label{lem:secure-query}
    Let $U$ be a $n$-qubit untary on the system register $\sfS$ whose the eigenvalues are well-separated.
    Let $\sfA$ be an ancillary register whose size satisfying the requirement of Lemma~\ref{lem:phase-unitary}. 
    
    Define $U_{f,S}=\CO^\dagger_{\qpe}\CO^\dagger_S \CO_f \CO_S \CO_{\qpe}$ as the ideal oracle on $\sfS\cup\sfA$ and the associated channel as $\mathsf{U}_{f,S}[\rho]:=U_{f,S}\rho U_{f,S}^\dagger$. Define $\mathcal{U}_{\operatorname{RU}}$ being the random phase ensemble w.r.t. $U$ and the associated channel as $\mathsf{U}_{\theta}[\rho]:=U_{\theta}\rho U_{\theta}^\dagger, U_{\theta}\in\mathcal{U}_{\operatorname{RU}}$. 
    
    Let $\{\mathsf{W}_{f,S}\}_{f,S}$ be a set of quantum channels implemented to approximate $\mathsf{U}_{f,S}$, such that $\norm{\mathsf{W}_{f,S}-\mathsf{U}_{f,S}}_\diamond\leq 2^{-\alpha n}$ for a positive constant $\alpha$ and all $f,S$.

    Introduce an ancillary register $\sfB$ with $|\sfB|=\poly{n}$. Fix a sequence of unitaries $\{A_1,\cdots A_t,A_{t+1}\}$ acting on $\sfS\cup\sfB$, and write $\mathsf{A}_i[\rho]:=A_{i}\rho A_{i}^\dagger$. Define the following two adaptive adversaries
    \begin{gather}
        \mathsf{AW}_{f,S}= \mathsf{A}_{t+1}\circ \prod_{i=1}^t\Big(
        \mathsf{W}_{f,S}\circ \mathsf{A}_{i}
        \big)
        ,\quad
        \mathsf{AU}_{\theta}[\rho_0]=\mathsf{A}_{t+1}\circ \prod_{i=1}^t\Big(
        \mathsf{U}_{\theta}\circ \mathsf{A}_{i}
        \big).
    \end{gather}
    Then for $S$ uniformly drawn from $\{0,1\}^{m_2+m_3}$, and $f$ is uniformly drawn from $[0,2\pi)^{2^n}$,
    \begin{gather}
        \norm{
        \mathbb{E}_{S,f}\Big[\mathsf{AW}_{f,S}\big[\ketbra{0_{\sfS\cup\sfA\cup\sfB}}{0_{\sfS\cup\sfA\cup\sfB}}\big] \Big]-
        \mathbb{E}_{U_{\theta}\sim\mathcal{U}_{\operatorname{RU}}}\Big[\mathsf{AU}_{\theta}\big[\ketbra{0_{\sfS\cup\sfA\cup\sfB}}{0_{\sfS\cup\sfA\cup\sfB}}\big] \Big]
        }_1\notag\\
        <t\Big(2^{-(m_3/2+1+n)}+2^{-\alpha n}\Big)+2^{-(m_2-n-3)}
    \end{gather}
\end{lemma}

\begin{proof}
    Introduce the following two channels
    \begin{gather}
        \mathsf{AV}_{f,S}=\mathsf{A}_{t+1}\circ \prod_{i=1}^t\Big(
        \mathsf{V}_{f,S}\circ \mathsf{A}_{i}
        \big),\quad \mathsf{AU}_{f,S}=\mathsf{A}_{t+1}\circ \prod_{i=1}^t\Big(
        \mathsf{U}_{f,S}\circ \mathsf{A}_{i}
        \big),
    \end{gather}
    where $\mathsf{V}_{f,S}[\rho]:=V_{f,S}\rho V_{f,S}^\dagger$, and $V_{f,S}=\sum_{k=1}^{2^n}e^{if(\widetilde{\lambda}^S_k)}\ketbra{\lambda_k}{\lambda_k}$. Lemma~\ref{lem:mimic-random-phase} asserts that $\mathbb{E}_{\theta}[\mathsf{AU}_{\theta}]=\mathbb{E}_{f,S}[\mathsf{AV}_{f,S}]$. Using triangle inequality,
    \begin{gather}
        \norm{
        \mathbb{E}_{S,f}\Big[\mathsf{AW}_{f,S}\big[\ketbra{0_{\sfS\cup\sfA\cup\sfB}}{0_{\sfS\cup\sfA\cup\sfB}}\big] \Big]-
        \mathbb{E}_{\theta}\Big[\mathsf{AU}_{\theta}\big[\ketbra{0_{\sfS\cup\sfA\cup\sfB}}{0_{\sfS\cup\sfA\cup\sfB}}\big] \Big]
        }_1 
        \notag\\
        \leq \mathbb{E}_{S,f}\Big[\norm{
        \mathsf{AW}_{f,S}\big[\ketbra{0_{\sfS\cup\sfA\cup\sfB}}{0_{\sfS\cup\sfA\cup\sfB}}\big] -
    \mathsf{AV}_{f,S}\big[\ketbra{0_{\sfS\cup\sfA\cup\sfB}}{0_{\sfS\cup\sfA\cup\sfB}}\big] 
        }_1 \Big]\notag\\
        \leq \mathbb{E}_{S,f}\Big[\norm{
        \mathsf{AW}_{f,S}\big[\ketbra{0_{\sfS\cup\sfA\cup\sfB}}{0_{\sfS\cup\sfA\cup\sfB}}\big] -
    \mathsf{AU}_{f,S}\big[\ketbra{0_{\sfS\cup\sfA\cup\sfB}}{0_{\sfS\cup\sfA\cup\sfB}}\big] 
        }_1\notag\\
        +\norm{
        \mathsf{AU}_{f,S}\big[\ketbra{0_{\sfS\cup\sfA\cup\sfB}}{0_{\sfS\cup\sfA\cup\sfB}}\big] -
    \mathsf{AV}_{f,S}\big[\ketbra{0_{\sfS\cup\sfA\cup\sfB}}{0_{\sfS\cup\sfA\cup\sfB}}\big] 
        }_1
        \Big].
    \end{gather}
    Using $\norm{\mathsf{W}_{f,S}-\mathsf{U}_{f,S}}_{\diamond}\leq 2^{-\alpha n}$,
     the submultiplicative of diamond norm, and $\norm{\Phi}_{\diamond}=1$ when $\Phi$ is a channel,
    \begin{align}\label{eq:adaptive-distance}
    \norm{\mathsf{AW}_{f,S}-\mathsf{AU}_{f,S}}_{\diamond}&=
        \norm{
        \mathsf{A}_{t+1}\circ\prod_{i=1}^t\Big(\mathsf{W}_{f,S}\circ\mathsf{A}_i\Big)-
        \mathsf{A}_{t+1}\circ\prod_{i=1}^t\Big(\mathsf{U}_{f,S}\circ\mathsf{A}_i\Big)
        }_{\diamond} \notag\\
        &=
        \norm{
        \prod_{i=1}^t\Big(\mathsf{W}_{f,S}\circ\mathsf{A}_i\Big)-
        \prod_{i=1}^t\Big(\mathsf{U}_{f,S}\circ\mathsf{A}_i\Big)
        }_{\diamond} \notag\\
        &=\norm{
        \sum_{j=1}^t \Big(\prod_{i=j+1}^t \mathsf{W}_{f,S}\circ\mathsf{A}_i\Big)\circ
        \Big(\mathsf{W}_{f,S}\circ\mathsf{A}_j-\mathsf{U}_{f,S}\circ\mathsf{A}_j\Big)\circ
        \Big(\prod_{i=1}^{j-1} \mathsf{U}_{f,S}\circ\mathsf{A}_i\Big)
        }_{\diamond}\notag\\
        &\leq \sum_{j=1}^t\norm{
         \Big(\prod_{i=j+1}^t \mathsf{W}_{f,S}\circ\mathsf{A}_i\Big)\circ
        \Big(\mathsf{W}_{f,S}\circ\mathsf{A}_j-\mathsf{U}_{f,S}\circ\mathsf{A}_j\Big)\circ
        \Big(\prod_{i=1}^{j-1} \mathsf{U}_{f,S}\circ\mathsf{A}_i\Big)
        }_{\diamond}
        \notag\\
        &\leq \sum_{j=1}^t \norm{
        \mathsf{W}_{f,S}-\mathsf{U}_{f,S}
        }_{\diamond}\notag\\
        &\leq t2^{-\alpha n}. 
    \end{align}
    As a result,
    \begin{align}
        \mathbb{E}_{S,f}\Big[\norm{
        \mathsf{AW}_{f,S}\big[\ketbra{0_{\sfS\cup\sfA\cup\sfB}}{0_{\sfS\cup\sfA\cup\sfB}}\big] -
    \mathsf{AU}_{f,S}\big[\ketbra{0_{\sfS\cup\sfA\cup\sfB}}{0_{\sfS\cup\sfA\cup\sfB}}\big] 
        }_1\Big]\leq t2^{-\alpha n}
    \end{align}
    
    Define
    \begin{gather}
        \CS_1:=\bigg\{
        S~\bigg|~\norm{
        U_{f,S}\ket{\psi}\otimes \ket{0_{\sfA}}-\big(V_{f,S} \otimes \operatorname{I}_\sfA\big)\ket{\psi}\otimes\ket{0_{\sfA}}
        }_2\leq2^{-(m_3/2+2+n)},\quad \forall \ket{\psi}\in\sfS
        \bigg\},\notag\\
        \CS_2:=\{0,1\}^{m_2+m_3}\setminus \CS_1.
    \end{gather}
    When $S\in\CS_1$, using Lemma~\ref{lem:state-to-channel},
    \begin{align}
        \norm{
        \mathsf{U}_{f,S}-
    \mathsf{V}_{f,S}
        }_\diamond \leq 2^{-(m_3/2+1+n)}.
    \end{align}
    By repeating the calculation of Eq.~\eqref{eq:adaptive-distance}, we have
    \begin{align}
        \norm{
        \mathsf{AU}_{f,S}\big[\ketbra{0_{\sfS\cup\sfA\cup\sfB}}{0_{\sfS\cup\sfA\cup\sfB}}\big] -
    \mathsf{AV}_{f,S}\big[\ketbra{0_{\sfS\cup\sfA\cup\sfB}}{0_{\sfS\cup\sfA\cup\sfB}}\big] 
        }_1\leq t2^{-(m_3/2+1+n)}.
    \end{align}

    Therefore,
    \begin{gather}
        \mathbb{E}_{f,S}\Big[\norm{
        \mathsf{U}_{f,S}[\ketbra{0_{\sfS\cup\sfA\cup\sfB}}{0_{\sfS\cup\sfA\cup\sfB}}] -
    \mathsf{V}_{f,S}[\ketbra{0_{\sfS\cup\sfA\cup\sfB}}{0_{\sfS\cup\sfA\cup\sfB}}] 
        }_1\Big]\notag\\
        =\pr(\CS_1)\cdot\mathbb{E}_{f}\Big[\norm{
        \mathsf{U}_{f,S\in\CS_1}[\ketbra{0_{\sfS\cup\sfA\cup\sfB}}{0_{\sfS\cup\sfA\cup\sfB}}] -
    \mathsf{V}_{f,S\in\CS_1}[\ketbra{0_{\sfS\cup\sfA\cup\sfB}}{0_{\sfS\cup\sfA\cup\sfB}}] 
        }_1\Big]\notag\\
        +\pr(\CS_2)\cdot\mathbb{E}_{f}\Big[\norm{
        \mathsf{U}_{f,S\in\CS_2}[\ketbra{0_{\sfS\cup\sfA\cup\sfB}}{0_{\sfS\cup\sfA\cup\sfB}}] -
    \mathsf{V}_{f,S\in\CS_2}[\ketbra{0_{\sfS\cup\sfA\cup\sfB}}{0_{\sfS\cup\sfA\cup\sfB}}] 
        }_1\Big]\notag\\
        \leq t2^{-(m_3/2+1+n)}+\pr(\CS_2)\norm{
        \mathsf{U}_{f,S\in\CS_2}[\ketbra{0_{\sfS\cup\sfA\cup\sfB}}{0_{\sfS\cup\sfA\cup\sfB}}] -
    \mathsf{V}_{f,S\in\CS_2}[\ketbra{0_{\sfS\cup\sfA\cup\sfB}}{0_{\sfS\cup\sfA\cup\sfB}}] 
        }_1\notag\\
        < t2^{-(m_3/2+1+n)}+2^{-(m_2-n-3)}.
    \end{gather}
    In the last line we use $\pr(\CS_2)<2^{-(m_2-n-2)}$ from Lemma~\ref{lem:phase-unitary}, and the fact $\norm{\rho_1-\rho_2}_1\leq 2$. 

    These together give
    \begin{gather}
        \norm{
        \mathbb{E}_{S,f}\Big[\mathsf{AW}_{f,S}\big[\ketbra{0_{\sfS\cup\sfA\cup\sfB}}{0_{\sfS\cup\sfA\cup\sfB}}\big] \Big]-
        \mathbb{E}_{\theta}\Big[\mathsf{AU}_{\theta}\big[\ketbra{0_{\sfS\cup\sfA\cup\sfB}}{0_{\sfS\cup\sfA\cup\sfB}}\big] \Big]
        }_1 \notag\\<t\Big(2^{-(m_3/2+1+n)}+2^{-\alpha n}\Big)+2^{-(m_2-n-3)}.
    \end{gather}
\end{proof}\noindent
By choosing large enough $m_1,m_2,m_3=O(n)$, these two adversaries give an order-$\negl{n}$ difference even when $t$ is super-polynomially large.
Thus, any inverse-exponentially accurate approximations of ensemble
$\{\CO_{\qpe}^\dagger \CO_{S}^\dagger \CO_{f}\CO_S \CO_{\qpe}\}_{f,S}$ is computationally indistinguishable from the random phase ensemble of $U$.

Finally, we turn to explicit constructions of $U_{f,S}$. We note that the QPE oracle, which uese $e^{iHt}$ for exponentially-long $t$, can be constructed efficiently for commuting Hamiltonians. There the time evolution operator $e^{iHt}=e^{i\sum_jh_j t}$ can be factorized into $\prod_je^{ih_j t}$. One only needs to run the quantum simulation algorithm for each local term. In other words, the commuting Hamiltonians can be fast-forwarded.
\begin{fact}
    [Fast-forwarding of commuting Hamiltonians, cf.\cite{atia2017fast}]
    \label{fact:fast-commuting}
    For any exponentially small precision $\epsilon$, there exists an efficient quantum algorithm $\{U_n(T)\}$ acting on $n+c=n+\poly{n}$ qubits such that for any $T=2^{O(n)}$ and any $n$-qubit
    state $\ket{\psi}$,
    \begin{align}
        \norm{(e^{-iH_nT}\otimes \mathbb{I}_c-U_n(T))\ket{\psi}\otimes\ket{0_c}}_2\leq \epsilon,
    \end{align}
    where $H_n$ is any $n$-qubit commuting Hamiltonian with $d=\CO(\log n)$.
\end{fact}

Putting everything together, we have
\begin{theorem}
    [Constructing energy-conserving PRU for random commuting Hamiltonians]
    \label{thm:construct-pru-commuting}
    Let $H_n=\sum_i\CJ_ih_i$ be a commuting Hamiltonian with digitalized Gaussian random coefficients $\CJ_i$. There is an efficient ensemble of unitary $\{U_{g,S}\}$ that forms the energy-conerving PRU of $H_n$ with probability at least $1-\negl{n}$. Each unitary $U_{g,S}$ is an exponentially accurate approximation of $\CO_{\qpe}^\dagger \CO_{S}^\dagger \CO_{g}\CO_S \CO_{\qpe}$, where $\CO_{\qpe}$ is the QPE oracle of $U=e^{i2\pi\widetilde{H}_n}$ with $\widetilde{H}_n$ being the normalized version of $H_n$ such that all the eigenvalues belong to $[0,1)$, $\CO_S$ is the random offset oracle, and $O_{g}:\ket{x}\rightarrow (-1)^{g(x)}\ket{x}$ implements a pseudo-random function $g:\{0,1\}^*\rightarrow \{0,1\}$. The oracles acting on system and ancillary registers with sizes properly specified by Lemma~\ref{lem:phase-unitary}.
\end{theorem}
\begin{proof}

The construction is sketched in Fig~\ref{fig:pru-commuting}.
Lemma~\ref{lem:small-gap-digitalization} asserts that with probability at least $1-\negl{n}$, $H_n$ has at a least inverse-exponentially small energy gap. We prove that the construction forms energy-conserving PRU for this case.
    \paragraph{Security.}  
    Replace the pseudo-random function $g$ by a truly random  function $f$. By Lemma~\ref{lem:secure-query}, any inverse-exponentially accurate approximation of the ensemble $\{\CO_{\qpe}^\dagger \CO_{S}^\dagger \CO_{f}\CO_S \CO_{\qpe}\}_{f,S}$ is computationally indistinguishable from the random phase ensemble of $U=e^{i2\pi\widetilde{H}_n}$, which is exactly the energy-conserving random unitary ensemble of $H_n$ due to Fact~\ref{non-degenerated-pru}. Since a pseudo-random function $g$ is indistinguishable from the truly random function $f$, ensemble $\{\CO_{\qpe}^\dagger \CO_{S}^\dagger \CO_{g}\CO_S \CO_{\qpe}\}_{g,S}$ is also indistinguishable from the energy-conserving random unitary ensemble.
    %By Lemma~\ref{lem:secure-query}, any exponentially accurate approximation of $\{\CO_{\qpe}^\dagger \CO_{S}^\dagger \CO_{g}\CO_S \CO_{\qpe}\}_{g,S}$ also has this property.
    Thus, $\{U_{g,S}\}$ is the energy-conserving PRU.
    
    \paragraph{Construction time.}
    Oracle $\CO_g$ is the standard phase oracle. Since $g$ can be efficiently specified and computed, this oracle is efficiently constructable. $\CO_S$ is the addition oracle, which is also efficiently constructable when $S$ is drawn from a set of size $2^{O(n)}$. 

    $\CO_{\qpe}$ uses $O(n)$ number of unitary $U(t)=e^{i2\pi\widetilde{H}_n t}$, with $t=2^{O(n)}$. By Fact~\ref{fact:fast-commuting}, each $U(t)$ can be efficiently implemented up to any inverse-exponentially small precision. By choosing a large enough precision, an  approximation of $\CO_{\qpe}$ with inverse exponential error can be constructed. As a result, $U_{g,S}$ is an inverse exponentially accurate approximation for $\CO_{\qpe}^\dagger \CO_{S}^\dagger \CO_{g}\CO_S \CO_{\qpe}$.
\end{proof}

\section{Solving $\pspace$ problems with energy-conserving random unitary}
\label{ap:hardness-ham-construction}

From this section, we delve into the construction and demonstration of 1D translational-invariant local Hamiltonian whose energy-conserving PRU does not exist. The distinctive feature for this Hamiltonian is that its energy-conserving random unitary can be used to solve $\pspace$-complete problems, therefore essentially differs from any efficient quantum circuit.

In this section, we construct the Hamiltonian and $\pspace$ solver. Using the Feynman-Kitaev type construction\cite{feynman1986quantum,kitaev2002classical}, one can construct 1D translational-invariant local Hamiltonians that encode the dynamics of any classical Turing machine up to polynomial memory size (limited by the system length). Our method generalizes this construction to enable its energy-conserving random unitary to solve $\pspace$ problems.

\subsection{Quantum simulation of reversible Turing machines}\label{subsec:simulation-rtm}

We begin by describe our construction for RTM. The procedure is illustrated in Fig~\ref{fig:illustration-tm}.

Fix the length of the tape, the configuration space of any RTM can be embedded into a Hilbert space made up of local pieces. The core idea is illustrated in Fig.~\ref{fig:illustration-tm}. First, we can put the head inside the tape, so that the entire RTM is an one-dimensional system. This one-dimensional system can be represented by a product state $\ket{\psi}=\ket{x_{k_1}}_1\ket{x_{k_2}}_2\cdots\ket{x_{k_{R-1}}}_{R-1}\ket{q_\mu}_R\ket{x_{k_R+1}}_{R+1}\cdots \ket{x_{k_{l+1}}}_{l+1}\in \otimes_{i=1}^{l+1}\hilbert_i$, where $l$ is the tape length, $R$ denotes the location of the head, and $x_{k_m}\in\Gamma$, $q_\mu\in Q$.

Having the quantum state representation, the transition functions $\delta$ can be represented as a set of local quantum isometries. For example, for $\delta_k:(q_\mu\rightarrow q_\nu,x\rightarrow y,+)$, it can be represented by 
\begin{align*}
V_{\delta_k}=\sum_{i=1}^{l}\ket{y}_i\ket{q_\nu}_{i+1}\bra{q_\mu}\bra{x}_{i+1}.
\end{align*}
$V_{\delta_k}$ encodes transition function $\delta_k$ in the following sense: For any quantum product state $\ket{\psi_\C}$ that corresponds to the configuration $\C$ of the RTM, 
$V_k\ket{\psi_\C}=\ket{\psi_{\C'}}$, where $\C':\C\xrightarrow{\delta_k}\C'$.
While $V_k$ is not unitary, it is local and translational-invariant. So we can encode the transition functions $\delta$ of the TM into a local, translational-invariant Hamiltonian defined as $H=V+V^\dagger$. 

In below we summarize this construction. For later convenience, we include the cases where the RTM contains reverse form of transition rules (Definition~\ref{def:inverse-transition-rule}).
\begin{definition}
    [Hamiltonian and state representations of reversible of Turing machine]
    \label{hamiltonian-representation}
    Given a reversible Turing machine $\tm$, whose  transition rules may contain both standard and reverse forms, $\Delta=\Delta_s\cup\Delta_n$.
    The Hamiltonian representation $\{H_n\}_{n\in\mathbb{N}^+}$ of $\tm$ is a uniform family of qudits Hamiltonians where $H_n$ is the Hamiltonian representation of $\tm(n)$. $H_n$ is defined on Hilbert space $\hilbert_{\tm(n)}=\otimes_{i=1}^{n+1}\hilbert_i$, such that $\hilbert_i$ is spanned by orthonormal states $\{\ket{x}_i,\ket{q}_i\}$ for all $x\in \Gamma$ and $q\in Q$. 
    %We call the basis spanned by $\{\ket{x}_i,\ket{q}_i\}$ the computational basis.

    Let $\mathcal{T}$ be the translation operator that maps any $\ket{\psi}_i\in\hilbert_i$ to $\ket{\psi}_{i+1}\in\hilbert_{i+1}$ (with periodic boundary condition). $H_n$ is defined as 
    \begin{align}
    H_n=\sum_{i=0}^{n}\mathcal{T}^{-1,i}\big(\sum_{\ket{v_k}\bra{v_l}\in\mathcal{V}}
    (\ket{v_k}+\ket{v_{l}})(\bra{v_k}+\bra{v_l})
    \big)
    \mathcal{T}^i,
    \end{align}
    where
    \begin{itemize}
        \item $\mathcal{V}=(\mathcal{V}^s_+\cup\mathcal{V}^s_-\cup\mathcal{V}^s_{\mathsf{0}})
        \cup (\mathcal{V}^n_+\cup\mathcal{V}^n_-)
        $.
        \item $\mathcal{V}^s_+=\{\ket{x'}_1\ket{q'}_2\bra{q}_1\bra{x}_2~\big|~ (q,x,q',x',+)\in\Delta^s\}$.  $\mathcal{V}^n_+=\{\ket{y}_1\ket{q'}_2\ket{x'}_3\bra{q}_1\bra{y}_2\bra{x}_3~\big|~  (q,+,x,q',x')\in\Delta^n, y\in\Gamma \}$.
        \item $\mathcal{V}^s_-=\{\ket{q'}_1\ket{y}_2\ket{x'}_3\bra{y}_1\bra{q}_2\bra{x}_3 ~\big| ~ (q,x,q',x',-)\in\Delta^s, y\in\Gamma\}$.
        $\mathcal{V}^n_-=\{\ket{q'}_1\ket{x'}_2\bra{x}_1\bra{q}_2~\big|~ (q,-,x,q',x')\in\Delta^n\}$.
        \item $\mathcal{V}^s_\mathsf{0}=\{\ket{q'}_1\ket{x'}_2\bra{q}_1\bra{x}_2~\big|~ (q,x,q',x',\mathsf{0})\in\Delta^s\}$.
    \end{itemize}
    $H_n$ is a local, translational-invariant Hamiltonian that encodes the dynamics of $\tm(n)$.

    Any configuration $\C=\braket{q,x,R}$ of $\tm(n)$ can be represented by a product state in $\hilbert_{\tm(n)}$: 
    \begin{align}
        \ket{\psi_\C}=\ket{x_{1}}_1\ket{x_{2}}_2\cdots\ket{x_{{j-1}}}_{R-1}\ket{q}_R\ket{x_{j}}_{R+1}\cdots \ket{x_{l}}_{l+1},
    \end{align}
    which we named the configuration state. 
    We call the subspace of $\hilbert_{\tm(n)}$ spanned by all configuration states the computational subspace $\hilbert_{\tm(n)}^\mathscr{C}$.
\end{definition}

To study the structure of this Hamiltonian, we find that the configuration space of a RTM (with fixed-size memory) can be enumerated.
\begin{fact}\label{computation-structure}
    Let $\tm(n)$ be a RTM with fixed-size memory, and $\mathscr{C}$ be the (finite) set of configurations of $\tm(n)$. Let $\mathcal{G}=(\mathscr{C},E)$ be a directed graph whose vertices are the configurations. The edges are defined as follows: if $C_2$ is the successor of $C_1$, draw a directed edge from $C_1$ to $C_2$.

    $\mathcal{G}$ can be divided into disjoint subgraphs of loops and paths:
    \begin{enumerate}
        \item Directed loops where $q_0,q_a,q_r$ are never contained.
        \item Directed paths where $q_0,q_a,q_r$ are never contained.
        \item Directed paths where a configuration that contains $q_0$ as the source. $q_0,q_a,q_r$ are never contained in other vertices.
        \item Directed paths where a configuration that contains $q_a$ or $q_r$ as the sink. $q_0,q_a,q_r$ are never contained in other vertices.
        \item Directed paths where a configuration that contains $q_0$ as the source and a configuration that contains $q_a$ or $q_r$ as the sink. $q_0,q_a,q_r$ are never contained in other vertices.
    \end{enumerate}
 \end{fact}
  We emphasize that any configurations that contain $q_0$ ($q_r$ or $q_a$) must be the source (sink) of a path, because they are designed to not have predecessor (successor).

This fact asserts that $H_n$ is highly fragmented in the subspace $\hilbert_{\tm(n)}^\mathscr{C}$. In each subspace, the effective Hamiltonian is as follows.
\begin{lemma}\label{structure-ham}
    Let $H_n$ be the Hamiltonian representation of a RTM $\tm(n)$. For any configuration $\C\in\mathscr{C}$ that belongs to a ``path'', there exist $T_1,T_2\in \mathbb{N}^+$ and a $T_1+T_2+1$-dimensional invariant subspace $\hilbert_{\tm(n)}^\C\subseteq \hilbert_{\tm(n)}^\mathscr{C}$ that contains $\ket{\psi_\C}$, where the effective Hamiltonian reads
    \begin{align*}
        H_n\Big|_{\hilbert_{\tm(n)}^\C}=\sum_{t=-T_1}^{T_2-1}
    \big(\ket{\psi_{\C_{t+1}}}\bra{\psi_{\C_{t}}}+h.c.\big)+2\sum_{t=-T_1+1}^{T_2-1}\big(\ket{\psi_{\C_t}}
        \bra{\psi_{\C_t}}\big)+\ket{\psi_{\C_{T_1}}}\bra{\psi_{\C_{T_1}}}
        +\ket{\psi_{\C_{T_2}}}\bra{\psi_{\C_{T_2}}},
    \end{align*}
    where $\C_0=\C$, and $\C_{t+1}$ denotes the successor of $\C_t$.
    
    If $\C$ belongs to a ``loop'', the effective Hamiltonian reads 
    \begin{align*}
        H_n\Big|_{\hilbert_{\tm(n)}^\C}=\sum_{t=-T_1}^{T_2}
        \big(\ket{\psi_{\C_{t+1}}}\bra{\psi_{\C_{t}}}+h.c.\big)+2\sum_{t=-T_1}^{T_2+1}\big(\ket{\psi_{\C_t}}
        \bra{\psi_{\C_t}}\big),
    \end{align*}
    where $\C_{-T_1}=\C_{T_2+1}$.  
\end{lemma}
\begin{proof}
We prove the result for the case of $\C$ belonging to a path. The case that $\C$ belongs to a loop follows directly.

Let $V_{\text{forward}}=\sum_{i=0}^{n}T^{-1,i}\big(\sum_{\ket{v_k}\bra{v_l}\in\mathcal{V}}\ket{v_k}\bra{v_l}\big)T^i$. Since $\ket{v_k}\bra{v_l}$ corresponds to a transition rule of $\tm$, we have
$V_{\text{forward}}\ket{\psi_{\C_t}}=\ket{\psi_{\C_{t+1}}}$. $V_{\text{forward}}+V_{\text{forward}}^\dagger$ becomes the first term of effective Hamiltonian.

The remaining part is $V_{p}=\sum_{i=0}^{n}T^{-1,i}\big(\sum_{\ket{v_k}\bra{v_l}\in\mathcal{V}}\ket{v_k}\bra{v_k}+\ket{v_l}\bra{v_l}\big)T^i$. $V_p$ gives energy penalty 1 to every configuration with one predecessor and also to every configuration with one successor. So the energy penalty of source and sink terms is 1 while those for the remaining terms are 2.

    Since a RTM is reversible, each configuration can appear only once in a subspace. Furthermore, different configuration states are orthogonal to each other. So the  dimension of $\hilbert_{\tm(n)}^\C=\operatorname{span}\{\ket{\psi_{\C_t}} \}$ is  $T_1+T_2+1$.
\end{proof}

In each invariant subspace, the effective Hamiltonian can be interpreted as a 1D hopping Hamiltonian along the chain noded by $\ket{\psi_{\C_t}}$. The hardness of the $\pspace$ problem lies in the fact that this path could be exponentially long, preventing any local updating algorithm on $\ket{\psi_{\C_t}}$ to efficiently get the solution. The energy-conserving Haar-random unitary of $H_n$ can potentially overcome this difficulty, with  the following ideas. 

To begin with, one prepares the product state $\ket{\psi_{\C_1}}$ that corresponds to the input configuration $\C_1$ of a $\pspace$ RTM for a given instance $x$. This state serves as the source of one path, whose sink encodes the solution.  Then one can sample a unitary $U$ from the energy-conserving Haar-random ensemble, and acting it on $\ket{\psi_{\C_1}}$. The scrambling nature of $U$ will produce a wavefunction $U\ket{\psi_{\C_1}}$ dispersing along the entire path. A follow-up measurement on the basis $\operatorname{span}\{\ket{x}_i,\ket{q}_i\}$ will collapse $U\ket{\psi_{\C_1}}$ to any state $\ket{\psi_{\C_t}}$ along the path with almost equal probabilities. In this way, with high probability, one can make an exponentially large step forward using one query of $U$. This procedure is illustrated by Fig~\ref{fig:illustration-tm}.

However, there are two obstacles in the way turning this intuition into a practical protocol. First of all, even if one can make a big step forward, the probability to precisely collapse to the sink is still exponentially small. More severely, different paths and loops could have energy degeneracies. Therefore, a energy-conserving Haar-random unitary could mix different subspaces with $\acc$ and $\rej$ as outputs, producing a wrong solution. 

In the next two subsections, we resolve these two issues.

 %As dipicted in Fig.~\ref{DRTM}, a DRTM constitutes two parallel tapes. The internal states set of RTM, $Q$, is also duplicated as $Q_D=Q_1\cup Q_2$, $Q_{1,2}=\{q_\mu^{1,2}|q_\mu\in Q\}$. A DRTM operates as follows. Initially, the input is copied twice on each tape. The full computation is divided into two stages. On Stage 1, the machine start with $q_0^1\in Q_1$, and follows exactly the transition function $\Delta$ of the  original RTM, whereas the head can rewrite simultaneously on the parallel cells of two tapes. When the internal state changes to $q_h^1$, it immediately changes to $q_h^2$ afterwords and the machine enters the second stage. On Stage 2, the machine starts with $q_h^2$ and follows exactly the reverse transition functions $\Delta^{-1}$. But in this stage, the head is only allowed to read and rewrite cells on tape 1. So during stage two,  tape 2  is not affected. This makes the computation result can in principle be  obtained by reading tape 2 at any step of Stage 2. We will use this feature to solve $\pspace$ problems later.

\subsection{Duplicated Turing machine}
\begin{figure}
    \centering
    \includegraphics[width=0.6\linewidth]{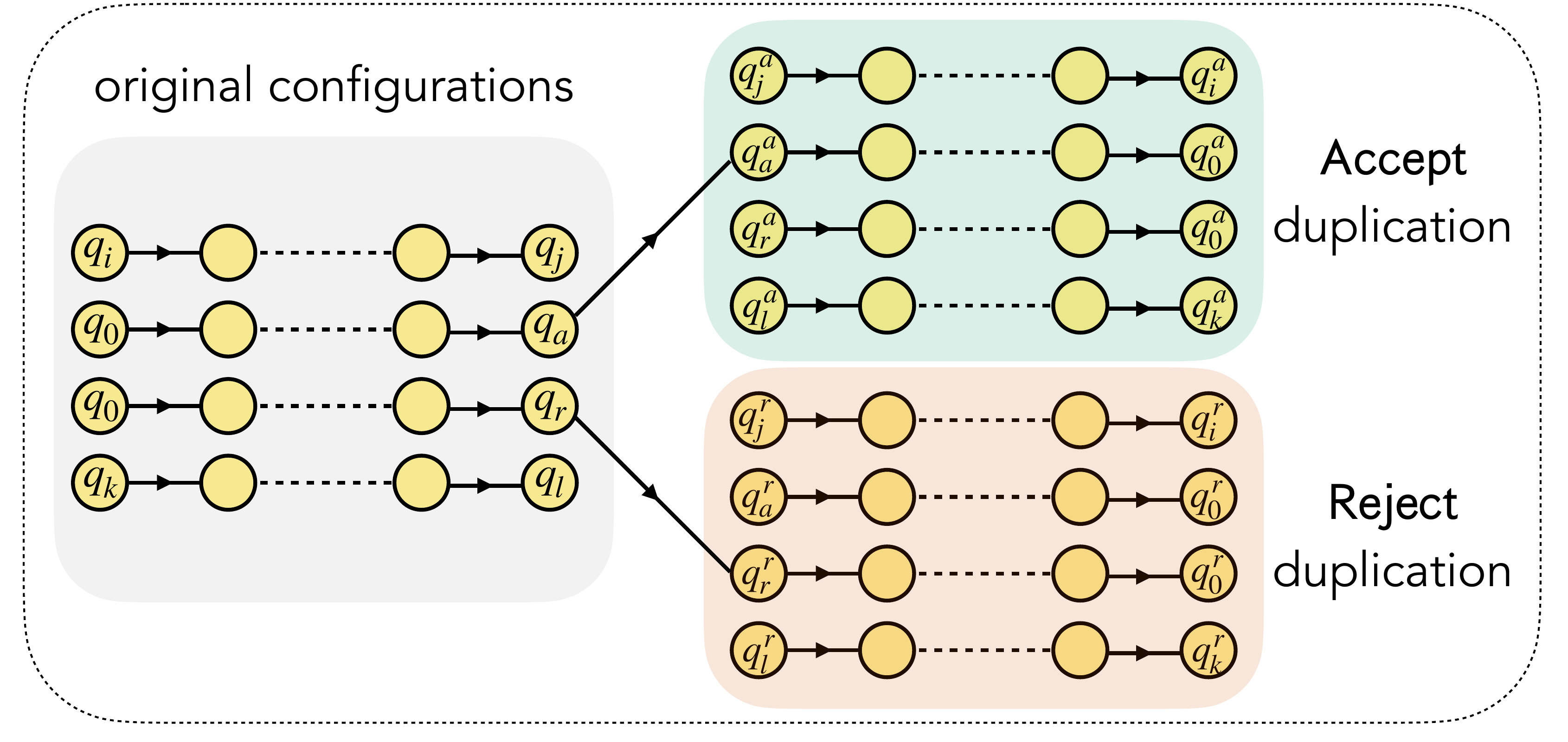}
    \caption{Configuration space of the duplicated RTM.}
    \label{fig:drtm}
\end{figure}

In this subsection, we introduce a construction we named the duplicated TM to amplify the probability to collapse to a state containing the solution information. 

A duplicated TM $\D[\tm]$ is constructed from a reference RTM $\tm=\langle Q,\Gamma,\Delta\rangle$. It contains two additional sets of internal states $Q^a$ and $Q^r$, together with two extra sets of transition rules $\Delta^a$ and $\Delta^r$. $Q^{a,r}$ are duplications of $Q$, i.e., $Q^{a,r}=\{q^{a,r}~\big|~q\in Q\}$. Similarly, $\Delta^{a,r}$ are duplications of $\Delta^{-1}$ (the inverse transition rules, see Definition ~\ref{def:inverse-transition-rule}) whose internal states are chosen from $Q^{a,r}$ instead. In addition, when the machine has configuration $\C=(q_{a},B,i)$ or $(q_{r},B,i)$, we forward it to $\C'=(q_{a}^a,B,i)$ or $(q_{r}^r,B,i)$, respectively. The global halting states for $\D[\tm]$ are $q_0^a$ and $q_0^r$ instead.

Following this construction, whenever the machine enters the old halting states $q_a$ or $q_r$, it enters $q_a^a$ or $q_r^r$ and does exactly the inverse computation. As a result, every computation process from legitimate input states now requires a doubled computation time, whereas the solution (i.e., $a$ or $r$) is encoded in all the configurations belonging to the latter half process. Using the ideas at the end of \ref{subsec:simulation-rtm}, we can read out the solution as long as the measured state collapses to the latter half path, whose probability is roughly $1/2=\CO(1)$.

\begin{definition}[Duplicated RTM]
    Given a RTM $\tm=\langle Q,\Gamma,\Delta\rangle$ whose transition rules are of the standard form, the Duplicated reversible Turing machine is another Turing machine, $\D[\tm]=\braket{Q^\D,\Gamma,\Delta^\D}$, where 
    \begin{itemize}
        \item $Q^\D=Q\cup Q^a\cup Q^r$, where $Q^{a,r}=\{q^{a,r}~\big|~q\in Q\}$. The halting state are now $q_0^a$ and $q^r_0$.
        \item $\Delta^\D=\Delta\cup \Delta^a\cup\Delta^r\cup\Delta^{\operatorname{trans}}$, where
        $\Delta^{a,r}=\{(q^{a,r},s,x,(q^{a,r})',x'~\big|~ (q,s,x,q',x')\in\Delta^{-1}\}$, and 
        $\Delta^{\operatorname{trans}}=\{(q_a,x,q_a^a,x,\mathsf{0}),(q_r,x,q_r^r,x,\mathsf{0})~\big|~ x\in \Gamma\}$.
    \end{itemize}
    In other words, when the RTM enters its own halting state $q_{a,r}$, the Duplicated RTM then enters $Q^{a,r}$ and does the exact reverse operation until reaching the global halting state $q_0^{a,r}$.
    The duplicated machine of a RTM with fixed-size memory can be defined analogously and labeled as $\D[\tm(n)]$,
\end{definition}

Following Fact~\ref{computation-structure}, we can also enumerate the configuration space of $\D[\tm]$. The structure is illustrated in Fig~\ref{fig:drtm}. 
\begin{lemma}\label{computation-structure-drtm}
    Let $\tm=\langle Q,\Gamma,\Delta\rangle$ be a RTM and $\D[\tm(n)]$ be the duplicated RTM with fixed-size memory, and $\mathscr{C}$ be the (finite) set of configurations of $\D[\tm(n)]$. Let $\mathcal{G}=(\mathscr{C},E)$ be a directed graph whose vertices are the configurations. The edges are defined as follows: if $C_2$ is the successor of $C_1$, draw a directed edge from $C_1$ to $C_2$.

    $\mathcal{G}$ can be divided into disjoint subgraphs of loops and paths of the following types:
    \begin{enumerate}
        \item Directed loops where $q_0,q_0^a,q_0^r$ are never contained.
        \item Directed paths where $q_0,q_0^a,q_0^r$ are never contained.
        \item Directed paths where a configuration that contains $q_0$ as the source. $q_0,q_0^a,q_0^r$ are never contained in other vertices.
        \item Directed paths where a configuration that contains $q_0^a$ or $q_0^r$ as the sink. $q_0,q_0^a,q_0^r$ are never contained in other vertices.
        \item Directed paths where a configuration that contains $q_0$ as the source and a configuration that contains $q_0^a$ or $q_0^r$ as the sink. $q_0,q_0^a,q_0^r$ are never contained in other vertices.
    \end{enumerate}
    Moreover they have the following properties, 
    \begin{itemize}
        \item Any configuration belongs to directed path of  type-3 does not contain $q\in Q^{a,r}$ as internal states.
        \item Any configuration belongs to directed path of  type-4 kind does not contain $q\in Q$ as internal states.
        \item Any direct path of  type-5 has length $2T$, where the first half of vertices contain $q\in Q$ as internal states, and the second half of vertices contain $q\in Q^r$ or $Q^a$ as internal states.
        \item Any initial configuration of $\D[\tm(n)]$ corresponding to some legal inputs belongs to directed path of the 5th kind.
    \end{itemize}
 \end{lemma}
 \begin{proof}
     By construction, $\D[\tm(n)]$ is a RTM. So the enumeration of configuration space follows from Fact. \ref{computation-structure}. 

     Let $P=(\C_1,\C_2\cdots,\C_T)$ be a directed path where $\C_1$ contains $q_0$ as internal state. If $\C_t$ contains $q\in Q^{a}$ as internal state for some $2\leq t\leq T$, by construction there must be $\C_{t'}$ with $t'\leq t$ that contains $q^a_a$ as internal state. As a result, $\C_{t'-1}$ contains $q_a$ as internal state.

     Since $\Delta^a$ is the reverse of $\Delta$, the $\D[\tm(n)]$ does exactly the reverse computation after reaching $\C_{t'}$. That is, for any $0\leq s\leq t'-1$, $\C_{t'-1+s}$ and $\C_{t'-s}$ are the same configuration except the former contains internal state $q_\mu\in Q$ and the latter contians $q_\mu^a$. As a result, whenever $P$ has a vertex that contains internal state in $Q^a$, $P$ must be symmetric with sink contains $q_0^a$. Otherwise, $P$ does not have any vertex that contains internal state in $Q^a$. This argument extends to the cases where $\C_T$ contains $q_0^a$ as internal state, and to the cases relative to $Q^r$. This proves the first three properties.

     Legal input configuration of $\D[\tm(n)]$ is also a legal input configuration of $\tm(n)$, thus it must contains $q_0$ and be the source of a direct path. Moreover, this path must contain $q_r$ or $q_a$ and then contain $q^r_0$ or $q^a_0$ by above arguments. So it belongs to the 5th kind.
 \end{proof}
From now, we denote $H_n$ to be the Hamiltonian representation of $\D[\tm(n)]$. It is clear the Hilbert subspace spanned by configuration states of any path (loop) is an invariant subspace of $H_n$. We call these subspaces the type-$i$ ($i=1,2,3,4,5$) path (loop) subspace with length $T$, and denote $\hilbert_{\D[\tm(n)]}^{(i)}$ as the union of all the type-$i$ subspaces.
Every path (loop) subspace is contained in the computation subspace $\hilbert^C_{\D[\tm(n)]}$ (Definition~\ref{hamiltonian-representation}) of $H_{n}$. $\cup_{1\leq i\leq 5}\hilbert_{\D[\tm(n)]}^{(i)}=\hilbert^\mathscr{C}_{\D[\tm(n)]}$

 In below we will choose $\tm$ to be a RTM that solves a $\pspace$-complete problem specified by language $L$. Assume $\tm$ decides if $x\in L$ using space $p(|x|)\in\mathbb{N}^+$.  For all $n\in\mathbb{N}^+$, the Hamiltonian representation of $\D[\tm(p(n))]$ is $H_{p(n)}$, which forms a uniform family of Hamiltonians.

\subsection{Lifting the degeneracy}

Loop and path subspaces of $H_{p(n)}$ can have energy degeneracy. Therefore energy-conserving Haar-random unitary of $H_{p(n)}$ can mix a subspace containing $\C_x$ with $x\in L$ and the one with $x\notin L$, and even with a subspace not corresponding to any legitimate computation process. This will cause unexpected computational errors.

In the following we energetically separate out the path subspaces corresponding to legitimate computation processes. Furthermore, we separate out the $\acc$ and $\rej$ path subspaces, so that the energy-conserving Haar-random unitary can be safely used to solve $L$.

In below we will call a directed path of type-4 that contains $q_0^a$ ($q_0^r$) as the sink the type-4$a$ (type-4$r$) path, and use the same notations for type-5$a$ and 5$r$.

First,  $\hilbert_{\D[\tm_L(p(n))]}$ contains not only the computational subspace $\hilbert^\mathscr{C}_{\D[\tm(n)]}$, but also the subspaces spanned by illegal configurations, i.e, the configurations containing multiple internal states $\ket{q}$s. To energetically separate the computational subspace, consider the following Hamiltonian 
\begin{align}
    H = H_{p(n)} + \gamma \sum_{q\in Q}\sum_{i=1}^{p(n)+1}\ket{q}_i\bra{q}_i.
\end{align}
Note that the second term commutes with $H_{p(n)}$. After adding this term, the total Hilbert space $\hilbert_{\D[\tm_L(p(n))]}$ is divided into:
\begin{itemize}
    \item The subspace spanned by the states that do not contain $\ket{q}$. This subspace has energy exactly 0.
    \item The subspace spanned by the states that contain one $\ket{q}$ (the computational subspace). Due to Lemma~\ref{structure-ham}, this subspace has energy $\gamma\leq E\leq \gamma +4$
    \item Any subspace containing state with more than one $\ket{q}$. Since $H_{p(n)}$ is a positive semidefinite Hamiltonian, this kind of subspaces has non-negative energy before adding $\gamma \sum_{q\in Q}\sum_{i=1}^{p(n)+1}\ket{q}_i\bra{q}_i$. So after adding the penalty term, their energies are at least $2\gamma$.
\end{itemize}
As a result, one can always energetically separate out the computational subspace by choosing large enough $\gamma>4$.

Now we separate the type-5a\&4a and type-5r\&4r paths subspaces from rest of the states in the computation subspace. The key technical tool is the following lemma.

\begin{lemma}
    \label{lem:solution-eigenvalue}
    For the following equation
    \begin{align}
        -\frac{1}{2^{m+1}-1}\tan\Big(\frac{k}{2}\Big)=\tan(Nk),\quad k\in(0,\pi),\quad
        m,N\in\mathbb{N}^+,
    \end{align}
    denote $\CS(m,N)$ as the solution set for $k$ with fixed $m$ and $N$. Define $\CS(m)=\bigcup_{N\in\mathbb{N}^+}\CS(m,N)$, then it satisfies the following:
    \begin{itemize}
        \item For $m_1\neq m_2$, $\CS(m_1)\bigcap\CS(m_2)=\CS(m_1)\bigcap\pi\mathbb{Q}=\CS(m_2)\bigcap\pi\mathbb{Q}=\emptyset$.
        \item For any $m$ and $N_1\neq N_2$, $\CS(m,N_1)\bigcap \CS(m,N_2)=\emptyset$ .
    \end{itemize}
\end{lemma}\noindent
We provide the proof in Appendix~\ref{ap:proof}.

Combining the two results above, we have the following lemma.
\begin{lemma}\label{computation-Hamiltonian}
Let $H_{p(n)}$ be the Hamiltonian representation of a duplicated RTM with memory-size $p(n)$. 
Define the computation Hamiltonian as
    \begin{align}
        H_{\operatorname{comp}} = H_{p(n)}+\sum_{i=1}^{p(n)+1}\Big(10 \sum_{q\in Q}(\ket{q}_i\bra{q}_i)+\frac{1}{2}\ket{q_0^a}_i\bra{q_0^a}_i+\frac{1}{4} \ket{q_0^r}_i\bra{q_0^r}_i 
        \Big).
    \end{align}
The additional terms preserve all the loop and path subspaces. Furthermore, $H_{\operatorname{comp}}$ satisfies:
\begin{itemize}
    \item Any energy eigenstate contained in $\hilbert^{(4a)}_{\D[\tm(p(n))]}\cup\hilbert^{(5a)}_{\D[\tm(p(n))]}$ has no energy degeneracy with rest of the eigenstates, and so as the ones contained in $\hilbert^{(4r)}_{\D[\tm(p(n))]}\cup\hilbert^{(5r)}_{\D[\tm(p(n))]}$.
    \item Inside $\hilbert^{(4a)}_{\D[\tm(p(n))]}\cup\hilbert^{(5a)}_{\D[\tm(p(n))]}$, the path subspaces with different lengths have no energy degeneracy with each other. The same for paths in $\hilbert^{(4r)}_{\D[\tm(p(n))]}\cup\hilbert^{(5r)}_{\D[\tm(p(n))]}$.
\end{itemize}
\end{lemma}
\begin{proof}
    First choose a large enough $\gamma$ to energetically separate out the computational subspace. As described before, $\gamma=10$ suffices. 
    Now we analyze the energy spectra of loop and path subspaces from type-1 to type-5. In this proof we  use $\ket{\psi_{\C_1}}$ and $\ket{\psi_{\C_T}}$ to denote the source and sink of a length-$T$ path, respectively.

    \textbf{Type-1}: The eigenvalues of type-1 loops are not changed by the additional terms. Due to Lemma~\ref{structure-ham}, for a loop of length $T$, the energy eigenvalues are $12+2\cos(2\pi t/T)$ for $t=0,1,\cdots, T-1$. 
    Any eigenvalue of loops is contained in $12+2\cos (\pi\mathbb{Q})$.

    \textbf{Type-2\&3}: The eigenvalues of type-2\&3 paths are not changed by the additional terms. Due to Lemma~\ref{structure-ham}, the effective Hamiltonian is a hopping Hamiltonian along an one-dimensional chain with boundary potentials $-\ket{\psi_{\C_{1}}}\bra{\psi_{\C_{1}}}-\ket{\psi_{\C_{T}}}\bra{\psi_{\C_{T}}}$. This problem is analyzed in details in Appendix~\ref{hopping-problem}. The energy eigenvalue is $E=12+2\cos k$, where $k$ is determined by $\tan(Tk)=0$. So $k=n\pi/T$, $n=1,2,\cdots, T$. As a result, any eigenvalue of type-2 paths is contained in $12+2\cos (\pi\mathbb{Q})$.

    \textbf{Type-4\&5}: The effective Hamiltonians of  type-4a\&5a paths are hopping Hamiltonians with boundary potential $-1$ and $-1/2$. Leveraging the result from Appendix~\ref{hopping-problem}, the eigenvalues are $E=12+2\cos k$, where $k$ is determined by
    \begin{align}\label{hopping-solution}
        -\frac{1}{3}\tan \frac{k}{2}=\tan Tk,
    \end{align}
    where $T$ is the length of the path. Similarly, the eigenvalue of type-4r\&5r paths is also $E=12+2\cos k$, where
    \begin{align}\label{hopping-solution-2}
        -\frac{1}{7}\tan \frac{k}{2}=\tan Tk.
    \end{align}
    Due to the first statement of Lemma~\ref{lem:solution-eigenvalue}, eigenstates contained in  $\hilbert^{(4a)}_{\D[\tm(p(n))]}\cup\hilbert^{(5a)}_{\D[\tm(p(n))]}$ and $\hilbert^{(4r)}_{\D[\tm(p(n))]}\cup\hilbert^{(5r)}_{\D[\tm(p(n))]}$ have no energy degeneracy with rest of the eigenstates. Due to the second statement of Lemma~\ref{lem:solution-eigenvalue}, type-4a\&5a (4r\&5r) paths with different lengths have no energy degeneracy with each other.
\end{proof}
Note that it is not necessary to separate a type-4a path from a type-5a one with the same length, because they both contain $q^a$s as internal states, serving as the signal for $\acc$. Similarly, we can keep type-4r and 5r paths degenerated.

\subsection{Construct the $\pspace$ solver}
In this subsection, we put everything together and
 demonstrate how to use the energy-conserving random unitary of $H_{\operatorname{comp}}$ to solve $\pspace$ problems. 
 
 Note that while Lemma~\ref{computation-Hamiltonian} successfully separate out the desired subspaces, the effective Hamiltonian of one subspace becomes the hopping Hamiltonian with unequal boundary potentials. Thus, the collapse probability to the second half of the chain is no longer guaranteed to be $1/2$. Luckily, it can still be bounded.
 To this end, we need to first deal with the normalization factor of wavefunctions of $H_{\operatorname{comp}}$ from Appendix~\ref{hopping-problem}. 

\begin{lemma}\label{lem:bound_C}
    Let $C_k(T) = T-\sin(Tk)/\sin k$. Let $k_1,\cdots k_T\in(0,\pi)$ be the solutions of Eq.~\eqref{hopping-solution} or Eq.~\eqref{hopping-solution-2}. There exists constants $T_0\in\mathbb{N}^+$,  such that for all $T>T_0$, the following holds: 
\begin{itemize}
    \item 
    $C_{k_t}(T)/C_{k_t}(2T)\leq 5/6$ for $t=2,3,\cdots,T-1$.
    \item $C_{k_1}(T)/C^2_{k_1}(2T),C_{k_T}(T)/C^2_{k_T}(2T)\leq 1/96$.
\end{itemize}
\end{lemma}
\begin{proof}
    The solutions of Eq~\eqref{hopping-solution} satisfies $(t-1/2)\pi/T\leq k_t\leq t\pi/T$ for all $t\in[T]$. First, we prove the first statement. Notice
    \begin{align}
        \frac{C_{k_t}(T)}{C_{k_t}(2T)}=\frac{T-\frac{\sin Tk}{\sin k}}{2T-\frac{\sin 2Tk}{\sin k}}\leq \frac{T+\frac{1}{\sin k}}{2T-\frac{1}{\sin k}}=
        \frac{1+\frac{1}{T\sin k}}{2-\frac{1}{T\sin k}}
    \end{align}
For $2\leq t\leq T-1$, $k_t\in (3\pi/(2T),\pi-\pi/T)$. Using the fact that there always exists $T_1>0$ such that for all $T>T_1$, 
\begin{align}
    \min\Big\{\sin\Big(\frac{3\pi}{2T}\Big),\sin\Big(\pi-\frac{\pi}{T}\Big)\Big\}=
    \sin\Big(\pi-\frac{\pi}{T}\Big)\geq \frac{\pi}{0.99 T}.
\end{align}
As a result, for $T>T_1$, $1/(T\sin k_t)\leq 0.99/\pi \leq 1/3 $. Then
\begin{align}
    \frac{C_{k_t}(T)}{C_{k_t}(2T)}=\frac{1+\frac{1}{T\sin k}}{2-\frac{1}{T\sin k}}\leq
    \frac{1}{2}+\frac{1}{T\sin k}\leq \frac{5}{6}.
\end{align}

Now we prove the second statement. To prove this, we first analyze the behavior of $k_T$. Note that for $k\in(\pi-\pi/(2T)+0^+,\pi)$, the LHS of Eq.~\eqref{hopping-solution} monotonically decreases, whereas the RHS monotonically increases. At $k=\pi-\pi/(2T)+0^+$, 
\begin{align}
    -\frac{\alpha}{\alpha+2}\tan\Big(\frac{\pi}{2}-\frac{\pi}{4T} \Big)=\operatorname{LHS}>\operatorname{RHS}=-\infty.
\end{align}
At $k = \pi-\pi/(3T)$, there exists $T_2>0$ such that for all $T>T_2$,
\begin{align}
    -\frac{\alpha}{\alpha+2}\tan\Big(\frac{\pi}{2}-\frac{\pi}{6T}\Big)=\operatorname{LHS}<\operatorname{RHS}=-\sqrt{3}
\end{align}
As a result, for $T>T_2$, $k_T\in(\pi-\pi/(2T),\pi-\pi/(3T))$. Following this result, we have
\begin{align}
    \max\Big\{\frac{1}{C_{k_1}(2T)},\frac{1}{C_{k_T}(2T)}\Big\}\leq
    \frac{1}{2T-\frac{1}{\min_{k_1,k_T}\{\sin k\}}}=\frac{1}{2T-\frac{1}{\sin(\pi/(3T))}}.
\end{align}
Now use the fact that there exists $T_3>0$, such that for all $T>T_3$, $1/\sin(\pi/(3T))\leq 3 \cdot0.99 T/\pi$. Then
\begin{align}
    \max\Big\{\frac{1}{C_{k_1}(2T)},\frac{1}{C_{k_T}(2T)}\Big\}\leq\frac{1}{T}
    \cdot \frac{\pi}{2\pi-3\cdot 0.99}.
\end{align}

Next, following the same proof the the first statement, there exists $T_4>0$ such that for $T>T_4$, 
\begin{align}
    \max\Big\{\frac{C_{k_1}(T)}{C_{k_1}(2T)},\frac{C_{k_T}(T)}{C_{k_T}(2T)}\Big\}\leq
    \frac{1}{2}+\frac{1}{T\sin (\pi/(3T))}\leq \frac{1}{2}+\frac{3}{0.99\cdot \pi}.
\end{align}
As a result, 
\begin{align}
    \max\Big\{\frac{C_{k_1}(T)}{C^2_{k_1}(2T)},\frac{C_{k_T}(T)}{C^2_{k_T}(2T)}\Big\}\leq
    \frac{1}{T}\Big(\frac{1}{2}+\frac{3}{0.99\cdot \pi}\Big)\Big(\frac{\pi}{2\pi-3\cdot 0.99}\Big)
\end{align}
Let $T_5>0$ be defined such that for all $T>T_5$, RHS of above equation is smaller than $1/96$.

Combining these two, we have that for $T_0=\max\{T_1,T_2,T_3,T_4,T_5\}$, the two statements hold.

\end{proof}

Using the above lemma, we can prove that the random phase unitary of $H_{\operatorname{comp}}$ can be used to solve the $\pspace$ problem. See the following two lemmas.
\begin{lemma}\label{solving-with-random-phase}
    Let $H_{\operatorname{comp}}$ be the computation Hamiltonian of $\D[\tm(p(n)]$, and $\mathcal{U}$ is the random phase unitary ensemble  w.r.t. eigenstates of $H_{\operatorname{comp}}$. Let $\ket{\psi_{\C_1}}$ be the configuration state corresponds to the initial state of $\D[\tm(p(n))]$ with $x$, where $|x|=n$. And $T$ (being even) denotes the running time of $\D[\tm(p(n)]$ on input $x$. Then there exists $T_0>0$, such that when $T>T_0$ the following holds:

    Prepare $\ket{\psi_{\C_1}}$ and randomly draw $U\in\mathcal{U}$, then doing measurement on $U\ket{\psi_{\C_1}}$ on the local basis $\{\ket{x}_i,\ket{q}_i\}$ of $\D[\tm(p(n))]$. The probability to collapse to the second half path, i.e., states $\ket{\psi_{\C_t}}$, $t=T/2+1,\cdots,T$, is greater than $1/12$.
\end{lemma}
 \begin{proof}
Let $E_k$ and $\ket{k}$ denote the eigenvalue and eigenstate of the effective Hamiltonian of the Type-5 path that contains $\ket{\psi_{\C_1}}$. Then $\ket{\psi_{\C_1}} = \sum_k\ket{k}\braket{k|\psi_{\C_1}}$. So after applying the random phase unitary, $\ket{\psi_{\C_1}}\rightarrow \sum_k e^{i\theta_k}\ket{k}\braket{k|\psi_{\C_1}}$, where $\theta_k$ is i.i.d. from uniform distribution over $[0,2\pi)$.

The probability to get any of the output state $\ket{\psi_{\C_t}}$ is 
\begin{align}
    \pr(\ket{\psi_{\C_t}})=\mathbb{E}_{\theta_1\cdots \theta_T}\big[\big|
    \sum_k\braket{\psi_{\C_t}|k}\braket{k|\psi_{\C_1}}e^{i\theta_k}\big|^2
    \big]=\sum_k|\braket{k|\psi_{\C_t}}|^2|\braket{k|\psi_{\C_1}}|^2.
\end{align}
In the last equality we use $\mathbb{E}_{\theta_1\cdots\theta_T}[e^{i(\theta_i-\theta_j)}]=\delta_{ij}$. Using the solution from Appendix~\ref{hopping-problem}, above expression is further simplified 
\begin{align}
    \pr(\ket{\psi_{\C_t}})=\sum_k\frac{1}{N_k(2T)^2}|e^{-ikt}-e^{ikt-ik}|^2|e^{-ik}-1|^2.
\end{align}
Here we introduce $N_k(2T)=\sum_{t=1}^{T}|e^{-ikt}-e^{ikt-ik}|^2=2T-\sin 2Tk/\sin k$.
According to Lemma. \ref{lem:bound_C}, there exists $T_0>0$ such that for all $T>T_0$,
  the probability to get measurement outcome $t=T/2+1,\cdots T$ is 
\begin{align}
    \pr\Big(\frac{T}{2}+1\leq t\leq T\Big) &=
    1-\sum_{t=1}^{T/2}\pr(\ket{\psi_{\C_t}})\notag\\
    &=1-\sum_{t=1}^{T}\frac{N_{k_t}(T)}{N_{k_t}^2(2T)}|e^{-i{k_t}}-1|^2\notag\\
    & = 1-\sum_{t=2}^{T-1}\frac{N_{k_t}(T)}{N^2_{k_t}(2T)}|e^{-ik_t}-1|^2-
    \frac{N_{k_1}(T)}{N^2_{k_1}(2T)}|e^{-ik_1}-1|^2-\frac{N_{k_T}(T)}{N^2_{k_T}(2T)}|e^{-ik_T}-1|^2\notag\\
    &\geq 1-\frac{5}{6}\sum_{t=1}^T\frac{|e^{-ik_t}-1|^2}{N_{k_t}(2T)}
    -\frac{1}{96}\big(|e^{-ik_1}-1|^2+|e^{-ik_T}-1|^2\big)\notag\\
    &\geq \frac{1}{12}
\end{align}
In the last inequality, we use $\sum_{t=1}^T|e^{-ik_t}-1|^2/N_{k_t}(2T)=\sum_{t=1}^T|\braket{\psi_{\C_1}|k_t}|^2=1$ and $ |e^{-ik_1}-1|^2\leq 4$.
 \end{proof}

\begin{lemma}\label{lemma:pru-solve-pspace}
    Let $L\in\pspace$. Then there exists one Turing machine $\tm=\braket{Q,\Gamma,\Delta}$ that can determine if $x\in L$ using polynomial space $p(|x|)$, such that the followings are satisfied for large-enough $|x|$: 

    Denote the computation Hamiltonian of $\D[\tm(p(n))]$ as $H_{\operatorname{comp}}$, and its energy-conserving random unitary ensemble as $\{\mathcal{U}_{p(n)}\}$. Let 
    \begin{align}
        P_{a,r}=\sum_{i=1}^{p(n)}\sum_{q\in Q}\ket{q^{a,r}}_i\bra{q^{a,r}}_i
    \end{align}
    be operators in $\hilbert_{\D[\tm(p(n))]}$, and $\ket{\psi_{\C_1}}$ be the configuration state corresponds to initial configuration of $\D[\tm(p(n))]$ with input $x$ where $|x|=n$. Then there exists $n_0$ that for all $n>n_0$, if $x\in L$,
    \begin{align}
        \mathbb{E}_{U\sim \mathcal{U}_{p(n)}}[\bra{\psi_{\C_1}}U^\dagger P_aU \ket{\psi_{\C_1}}]\geq \frac{1}{12},\quad  \mathbb{E}_{U\sim \mathcal{U}_{p(n)}}[\bra{\psi_{\C_1}}U^\dagger P_rU \ket{\psi_{\C_1}}]=0,
    \end{align}
    and if $x\notin L$,
    \begin{align}
        \mathbb{E}_{U\sim \mathcal{U}_{p(n)}}[\bra{\psi_{\C_1}}U^\dagger P_rU \ket{\psi_{\C_1}}]\geq \frac{1}{12},\quad  \mathbb{E}_{U\sim \mathcal{U}_{p(n)}}[\bra{\psi_{\C_1}}U^\dagger P_aU \ket{\psi_{\C_1}}]=0.
    \end{align}
\end{lemma}
\begin{proof}
    First, choose a $\tm$ that determines if $x\in L$ in polynomial space, such that for energy legitimate input $x$ with $|x|=n$, the running time is at least $n$. This can be easily achieved by demanding $\tm$ to first sweep the whole input string before doing computation. In the following, we prove the $x\in L$ statement holds for $H_{\operatorname{comp}}$ defined from $\tm$. Generalization to $x\notin L$ cases is straightforward.

    We begin by dealing with an idealized situation.
    Assume the subspace of path that contains $\ket{\psi_{\C_1}}$ (type-5a) has no energy degeneracy with rest of the eigenstates, then the energy-conserving PRU in this subspace is equivalent to the random phase unitary ensemble. Due to Lemma~\ref{solving-with-random-phase}, when the length of the path is greater than $T_0$, measuring the local basis projects the state $U\ket{\psi_{\C_1}}$ to the second half path with probability $\geq1/12$. By our construction of $\tm$, this always holds for large enough $|x$. Finally, notice that since the second half path contains configurations having $q^a$ as internal state (Lemma~\ref{computation-structure-drtm}), the probability to collapse to the second half path is exactly equal to $\mathbb{E}_{U\sim \mathcal{U}_{p(n)}}[\bra{\psi_{\C_1}}U^\dagger P_aU \ket{\psi_{\C_1}}]$. So we prove the theorem with the aforementioned assumption.

    In reality, the path that contains $\ket{\psi_{\C_1}}$ can have energy degeneracy with other type-4a\&5a paths with the same length. $U\ket{\psi_{\C_1}}$ has equal probabilities to collapse to each of the path.
    However, recall that type-4a paths have $q^a$ as internal states along the whole path. Therefore, whether the degenerated path is of type-4a or 5a, the second half path must contain configurations with $q^a$ as internal states. As a result, $\mathbb{E}_{U\sim \mathcal{U}_{p(n)}}[\bra{\psi_{\C_1}}U^\dagger P_aU \ket{\psi_{\C_1}}]$ cannot decrease. That completes the proof.
\end{proof}

Putting everything together, we prove that with query access to the energy-conserving random unitary ensemble of $H_\text{comp}$, one can construct a $\pspace$ solver.

\begin{algorithm}
\caption{$\mathsf{PSPACE}$ solver}
\label{alg:pspace-solver}
\KwIn{Input $x\in\{0,1\}^*$, query access to the unitary ensemble $\{\mathcal{U}_n\}$, and a polynomial $p(n)$}
\KwOut{$\mathsf{Accept}$ or $\mathsf{Reject}$}
Construct state $\ket{\psi_{\mathcal{C}_1}}$ that corresponds to the input $x$ in the Hilbert space of length-$p(|x|)$ chain.\;
\For{$t=1$ \KwTo $15$}{
     Query the unitary ensemble and get $U\ket{\psi_{\mathcal{C}_1}}$ for $U\in\mathcal{U}_{p(|x|)}$.\;
    Measure $U\ket{\psi_{\mathcal{C}_1}}$ on the local basis. \;
    \uIf{$q^a$ appears in the measurement outcome}
        {\Return $\mathsf{Accept}$}
    \uElseIf{$q^r$ appears in the measurement outcome}
        {\Return $\mathsf{Reject}$}
}
\Return $\mathsf{Accept}$
\end{algorithm}

\begin{theorem}
[Energy-conserving random unitary can be used to solve $\pspace$ problems]
\label{thm:ru-solve-pspace}
    For any language $L\in\pspace$, there exists a (uniform family of) translational-invariant one-dimensional local Hamiltonian $H$, whose matrix elements belong to $\{0,1,10,1/2,1/4\}$, such that a polynomial quantum algorithm with query access to the energy-conserving random unitary ensemble of $H$ can solve all the instances of $L$ with sufficiently large input size, with success probability at least $2/3$.
\end{theorem}
\begin{proof}
    Let $\tm$ be the Turing machine specified in Lemma~\ref{lemma:pru-solve-pspace} which solves $L$ in polynomial space $p(\cdot)$, and choose $H_\text{comp}$ to be the (uniform family of) computational Hamiltonian of $\D[\tm_L]$ specified by Lemma~\ref{computation-Hamiltonian}. With query access to the energy-conserving PRU of $H_\text{comp}$, Algo~\ref{alg:pspace-solver} can efficiently solve $L$ for sufficiently large input size. Due to Lemma~\ref{lemma:pru-solve-pspace}, the failure probability is at most $(1-1/12)^5\approx0.27<1/3$.
\end{proof}

\section{True Quantified Boolean Formula (TQBF)}
\label{ap:tqbf}

In the following three sections, we construct an algorithm to distinguish a $\pspace$-complete solver from any polynomial size quantum circuit. With query access to the solver, one can efficiently invert quantum-secure one-way functions, a property not held by any efficient quantum circuit.

We use \emph{True Quantified Boolean Formula} problem (TQBF) to demonstrate this property.

\subsection{Problem Statement}

TQBF is the canonical PSPACE-complete problem. It consists of all fully quantified Boolean formulas that evaluate to true over the Boolean domain. Formally, the language is defined as:
\begin{equation}
\mathrm{TQBF} := \left\{ \varphi = Q_1 x_1 \cdots Q_n x_n \; \psi(x_1, \ldots, x_n) \;\middle|\;
\begin{array}{l}
Q_i \in \{ \forall, \exists \},\ \psi \text{ is a Boolean formula}\\
\text{and } \varphi \text{ evaluates to true.}
\end{array}
\right\}.
\end{equation}
We assume that formulas are written in \emph{prenex normal form}, with all quantifiers preceding the propositional formula, and that $\psi$ is encoded either in conjunctive normal form or as a Boolean circuit. The input size of a formula is the total number of bits required to encode the quantifiers and $\psi$. As an example, the formula
$\phi = \forall x_1 \exists x_2 \forall x_3 \;\left[(x_1 \lor \neg x_2) \land (x_2 \lor x_3)\right]$
is in $\mathrm{TQBF}$.

TQBF is PSPACE-complete as TQBF $\in$ PSPACE and every language in PSPACE reduces to TQBF under polynomial-time reductions. As such, TQBF serves as the standard complete problem for reasoning about the computational power of PSPACE.

\begin{fact}
    TQBF is a $\pspace$-complete problem, cf.~\cite{stockmeyer1973word,chandra1981alternation}.
\end{fact}

\subsection{Power of TQBF}

\begin{lemma}[TQBF oracle can invert one-way function]
\label{lem:tqbf-owf-inversion}
Let $f: \{0,1\}^n \to \{0,1\}^m$ be any polynomial-time computable function. There exists a polynomial-time algorithm that, given oracle access to a correct $\mathsf{TQBF}$ solver, can invert $f$ on any input $y \in \{0,1\}^m$. That is, the algorithm either finds $x \in \{0,1\}^n$ such that $f(x) = y$, or correctly determines that no such $x$ exists.
\end{lemma}
\begin{proof}
Given a circuit $C$ computing $f$ and target $y$, we construct the TQBF formula:
$\Phi_{C,y} := \exists x_1, \ldots, x_n \; \psi(x_1, \ldots, x_n, y)$
where $\psi(x_1, \ldots, x_n, y)$ is a Boolean formula that evaluates to true if and only if $C(x_1, \ldots, x_n) = y$. This formula can be constructed in polynomial time by simulating the circuit $C$. First, we query the TQBF oracle on $\Phi_{C,y}$. If the oracle returns $\mathsf{False}$, then no preimage exists, and we output $\mathsf{Fail}$. If the oracle returns $\mathsf{True}$, we extract a witness using binary search: for each bit position $i = 1, \ldots, n$, we construct the formula: $\Phi_i^0 = \exists x_{i+1}, \ldots, x_n \; \psi(a_1, \ldots, a_{i-1}, 0, x_{i+1}, \ldots, x_n)$ where $a_1, \ldots, a_{i-1}$ are the bits determined in previous iterations. We query the oracle on $\Phi_i^0$. If it returns $\mathsf{True}$, we set $a_i = 0$; otherwise, we set $a_i = 1$. After $n$ such queries, we obtain $(a_1, \ldots, a_n)$, which is guaranteed to be a valid preimage since the original formula was satisfiable and the oracle is correct. The total number of oracle queries is $n + 1$, and each formula has polynomial size, so the algorithm runs in polynomial time.
\end{proof}

The algorithm for establishing Lemma~\ref{lem:tqbf-owf-inversion} is given in the following.

\begin{algorithm}
\caption{Function Inversion with TQBF Oracle}
\label{alg:OWF-inversion}
\KwIn{ Circuit $C$ computing $f$, target $y$}
\KwOut{ $x$ such that $f(x) = y$, or $\mathsf{Fail}$}
Construct TQBF formula $\Phi_{C,y} = \exists x_1, \ldots, x_n \; \psi(x_1, \ldots, x_n, y)$,
where $\psi(x_1, \ldots, x_n, y) \equiv [C(x_1, \ldots, x_n) = y]$\;
Query oracle: $b \leftarrow \mathcal{O}(\Phi_{C,y})$\;
\uIf{$b = \mathsf{False}$}
    {\Return $\mathsf{Fail}$}
\For{$i = 1$ \KwTo $n$}{
     Construct $\Phi_i^0 = \exists x_{i+1}, \ldots, x_n \; \psi(a_1, \ldots, a_{i-1}, 0, x_{i+1}, \ldots, x_n)$\;
    Query oracle: $b_i \leftarrow \mathcal{O}(\Phi_i^0)$\;
    \uIf{$b_i = \mathsf{True}$}{
         $a_i \leftarrow 0$}
    \uElse{
         $a_i \leftarrow 1$}
}
\Return $(a_1, \ldots, a_n)$
\end{algorithm}

\section{Verifying TQBF Oracles via One-Way Function}
\label{ap:verify-tqbf-via-oneway}

We consider the problem of distinguishing between a true $\mathsf{TQBF}$ oracle and a pseudo $\mathsf{TQBF}$ oracle constructed from a uniform family of polynomial-size quantum circuit claiming to solve $\mathsf{TQBF}$.

\begin{theorem}[TQBF Oracle Verification]
\label{thm:tqbf-verification}
Assume quantum-secure one-way functions exist. There exists a polynomial-time classical algorithm $\mathcal{V}$ (the verifier) that, given black-box access to a purported $\mathsf{TQBF}$ oracle $\mathcal{O}$, outputs either $\mathsf{True}$ or $\mathsf{Pseudo}$ such that:
\begin{enumerate}
\item If $\mathcal{O}$ correctly solves all $\mathsf{TQBF}$ instances with probability at least $2/3$, then the verifier $\mathcal{V}^{\mathcal{O}}$ outputs $\mathsf{True}$ with probability at least $2/3$;
\item If $\mathcal{O}$ is implemented by polynomial-size quantum circuits, then for sufficiently large system sizes, the verifier $\mathcal{V}^{\mathcal{O}}$ outputs $\mathsf{Pseudo}$ with probability at least $2/3$.
\end{enumerate}
\end{theorem}

\subsection{Verifier Algorithm}

The verifier uses Algorithm~\ref{alg:OWF-inversion} for function inversion with TQBF oracle as a subroutine to test whether the oracle can be used to consistently invert quantum-secure one-way functions.

\begin{algorithm}
\caption{TQBF Oracle Verifier}
\label{alg:verifier}
\KwIn{ System size parameter $n$, oracle $\mathcal{O}$}
\KwOut{ $\mathsf{True}$ or $\mathsf{Pseudo}$}
Let $f$ be a quantum-secure one-way function\;
Set $T = 100 \log n$ \quad \textit{// Number of test instances}\;
Initialize success counter $S = 0$\;
\For{$i = 1$ to $T$}{
    Choose input size $n' \leftarrow \{1, 2, \ldots, n\}$ uniformly at random\;
    Sample $x \leftarrow \{0,1\}^{n'}$ uniformly at random\;
    Compute $y \leftarrow f(x)$\;
    Run $x' \leftarrow \text{Algorithm}~\ref{alg:OWF-inversion}(C_f, y)$ where $C_f$ computes $f$\;
    \uIf{$x' \neq \mathsf{Fail}$ and $f(x') = y$}
        {$S \leftarrow S + 1$}
}
\uIf{$S \geq 2T/3$}
    {\Return $\mathsf{True}$}
\uElse
    {\Return$\mathsf{Pseudo}$}
\end{algorithm}

\subsection{Supporting Lemmas}

\begin{lemma}[Amplified $\mathsf{TQBF}$ oracle]
\label{lem:amplified-oracle}
Given a probabilistic oracle $\mathcal{O}$ that correctly solves $\mathsf{TQBF}$ instances with probability at least $2/3$, there exists a polynomial-time procedure that creates an amplified oracle $\mathcal{O}'$ which correctly solves any $\mathsf{TQBF}$ instance with probability at least $1 - 2^{-k}$ for any desired parameter $k$.
\end{lemma}
\begin{proof}
For any $\mathsf{TQBF}$ query $\phi$, the amplified oracle $\mathcal{O}'(\phi)$ queries $\mathcal{O}$ on $\phi$ exactly $6k$ times, obtaining answers $b_1, b_2, \ldots, b_{6k}$, then outputs the majority value among these answers. Since each individual query to $\mathcal{O}$ succeeds with probability at least $2/3$, the expected number of correct answers is at least $\mathbb{E}[\text{correct}] = 6k \cdot (2/3) = 4k$. Let $X$ be the number of correct answers. By Chernoff bound, $\Pr[X < 3k] = \Pr[X < \mathbb{E}[X] - k] \leq \exp(-2k^2/(6k)) = \exp(-k/3) \leq 2^{-k/5}$. For $k \geq 5$, we have $2^{-k/5} \leq 2^{-k}$. Since the majority of $6k$ answers is correct when at least $3k$ answers are correct, we have $\Pr[\mathcal{O}'(\phi) \text{ is correct}] \geq 1 - 2^{-k}$.
\end{proof}

\begin{lemma}[Function inversion with true oracle]
\label{lem:inversion-true-oracle}
Let $f$ be any polynomial-time computable function and $\mathcal{O}$ be a probabilistic oracle that correctly solves $\mathsf{TQBF}$ instances with probability at least $2/3$. Then there exists a polynomial-time algorithm that uses $\mathcal{O}$ to invert $f$ with probability at least $1 - \mathrm{negl}(n)$.
\end{lemma}
\begin{proof}
We modify Algorithm~\ref{alg:OWF-inversion} to use the amplified oracle $\mathcal{O}'$ from Lemma~\ref{lem:amplified-oracle} with parameter $k = 2n$. This ensures each individual $\mathsf{TQBF}$ query succeeds with probability at least $1 - 2^{-2n}$. Algorithm~\ref{alg:OWF-inversion} makes at most $n+1$ queries to the oracle: one initial query to check if a preimage exists, and at most $n$ queries for witness extraction. By union bound, the probability that any of these queries fails is at most $(n+1) \cdot 2^{-2n} \leq 2^{n} \cdot 2^{-2n} = 2^{-n}$, which is negligible in $n$. Therefore, the algorithm successfully inverts $f$ for any given image $y = f(x)$ with probability at least $1 - 2^{-n} \geq 1 - \mathrm{negl}(n)$.
\end{proof}

\begin{lemma}[Function inversion with quantum circuit oracle]
\label{lem:quantum-circuit-limitation}
If $\mathcal{O}$ is implemented by polynomial-size quantum circuits and $f$ is a quantum-secure one-way function, then for sufficiently large $n$, Algorithm~\ref{alg:OWF-inversion} using oracle $\mathcal{O}$ succeeds in inverting $f$ on a random input with probability at most $\mathrm{negl}(n)$.
\end{lemma}
\begin{proof}
When $\mathcal{O}$ is implemented by polynomial-size quantum circuits, the composition of Algorithm~\ref{alg:OWF-inversion} with $\mathcal{O}$ yields a uniform family of polynomial-size quantum circuits attempting to invert~$f$. The verifier samples $x \leftarrow \{0,1\}^n$ uniformly at random, computes $y = f(x)$, then runs Algorithm~\ref{alg:OWF-inversion} to find a preimage of $y$. By Definition~\ref{def:quantum-owf}, for any uniform family of polynomial-size quantum circuits $\{Q_n\}$, there exists $n_0$ such that for all $n \geq n_0$, $\Pr_{x \leftarrow \{0,1\}^n}[f(Q_n(f(x))) = f(x)] \leq \mathrm{negl}(n)$. Since our algorithm only succeeds if it finds a correct preimage of $f(x)$, and any correct preimage $x'$ satisfies $f(x') = f(x)$, the success probability is at most $\mathrm{negl}(n)$ for sufficiently large $n$.
\end{proof}

\subsection{Main Proof}
\begin{proof}[Proof of Theorem~\ref{thm:tqbf-verification}]
We analyze both cases of the theorem and the runtime.

\paragraph{Case 1: $\mathcal{O}$ correctly solves $\mathsf{TQBF}$ with probability $\geq 2/3$.} By Lemma~\ref{lem:inversion-true-oracle}, Algorithm~\ref{alg:OWF-inversion} succeeds in inverting $f$ on any given input with probability at least $1 - 2^{-n}$. In Algorithm~\ref{alg:verifier}, we perform $T = 100 \log n$ independent tests, where each test succeeds with probability at least $1 - 2^{-n}$. Let $X$ be the number of successful inversions. Then $\mathbb{E}[X] \geq T(1 - 2^{-n}) \geq T - T \cdot 2^{-n} \geq T - 1$ for sufficiently large $n$. By Chernoff bound, $\Pr[X < 2T/3] = \Pr[X < \mathbb{E}[X] - (T-1-2T/3)] = \Pr[X < \mathbb{E}[X] - T/3 + 1] \leq \Pr[X < \mathbb{E}[X] - T/4] \leq \exp(-2(T/4)^2/T) = \exp(-T/8) = \exp(-12.5 \log n) = n^{-12.5}$ for sufficiently large $n$. Therefore, $\mathcal{V}^{\mathcal{O}}$ outputs $\mathsf{True}$ with probability at least $1 - n^{-12.5} \geq 2/3$ for sufficiently large $n$.

\paragraph{Case 2: $\mathcal{O}$ is implemented by polynomial-size quantum circuits.} By Lemma~\ref{lem:quantum-circuit-limitation}, for sufficiently large $n$, Algorithm~\ref{alg:OWF-inversion} succeeds in inverting $f$ on a random input with probability at most $\mathrm{negl}(n)$. Let $X$ be the number of successful inversions in Algorithm~\ref{alg:verifier}. Then $\mathbb{E}[X] \leq T \cdot \mathrm{negl}(n)$. Since $\mathrm{negl}(n) = o(n^{-c})$ for any constant $c > 0$, and $T = 100 \log n$, for sufficiently large $n$ we have $\mathbb{E}[X] \leq T \cdot n^{-2} = 100 n^{-2} \log n \leq T/6$. By Chernoff bound, $\Pr[X \geq 2T/3] = \Pr[X \geq \mathbb{E}[X] + (2T/3 - \mathbb{E}[X])] \geq \Pr[X \geq \mathbb{E}[X] + T/2] \leq \exp(-2(T/2)^2/T) = \exp(-T/2) = \exp(-50 \log n) = n^{-50}$. Therefore, $\mathcal{V}^{\mathcal{O}}$ outputs $\mathsf{Pseudo}$ with probability at least $1 - n^{-50} \geq 2/3$ for sufficiently large $n$.

\paragraph{Runtime:} Algorithm~\ref{alg:verifier} makes $T = 100 \log n$ calls to Algorithm~\ref{alg:OWF-inversion}. Each call to the amplified oracle takes $O(n)$ repetitions, and each $\mathsf{TQBF}$ query takes polynomial time. Therefore, the total runtime is polynomial in $n$.
\end{proof}

\section{Verifier for random unitaries that conserve energy}
\label{ap:verifier-random-unitary}

The framework developed for verifying $\mathsf{TQBF}$ oracles can be extended to efficiently distinguish between truly random unitaries (Haar-random) that conserve energy and pseudorandom unitaries generated by polynomial-size quantum circuits. This extension relies on the computational power that access to Haar-random unitaries can provide.

\begin{theorem}[Universal distinguisher]
\label{thm:universal-distinguisher}
Assume quantum-secure one-way functions exist. There exists a family of one-dimensional, local, and translational-invariant Hamiltonian whose matrices elements belong to $\{0,1,10,1/2,1/4\}$, and a polynomial-time quantum algorithm $\mathcal{V}_{\text{unitary}}$ (the universal distinguisher) that, given black-box access to a unitary oracle $U$, outputs either $\mathsf{HaarRandom}$ or $\mathsf{Pseudorandom}$ such that:
\begin{enumerate}
\item If $U$ is a Haar-random unitary conserving the Hamiltonian, then the distinguisher $\mathcal{V}_{\text{unitary}}^{U}$ outputs $\mathsf{HaarRandom}$ with probability at least $2/3$;
\item If $U$ is generated by a uniform family of polynomial-size quantum circuits, then for sufficiently large system sizes, the distinguisher $\mathcal{V}_{\text{unitary}}^{U}$ outputs $\mathsf{Pseudorandom}$ with probability at least $2/3$.
\end{enumerate}
\end{theorem}
\begin{proof}
The proof follows by combining Theorem~\ref{thm:ru-solve-pspace} with Theorem~\ref{thm:tqbf-verification}. We construct the universal distinguisher $\mathcal{V}_{\text{unitary}}$ as follows:
\begin{enumerate}
\item Given access to unitary oracle $U$, use the construction from Theorem~\ref{thm:ru-solve-pspace} to build a purported $\mathsf{TQBF}$ oracle $\mathcal{O}_U$.
\item Run the $\mathsf{TQBF}$ oracle verifier $\mathcal{V}$ from Algorithm~\ref{alg:verifier} on $\mathcal{O}_U$.
\item If $\mathcal{V}^{\mathcal{O}_U}$ outputs $\mathsf{True}$, return $\mathsf{HaarRandom}$.
\item If $\mathcal{V}^{\mathcal{O}_U}$ outputs $\mathsf{Pseudo}$, return $\mathsf{Pseudorandom}$.
\end{enumerate}

\paragraph{Case 1: $U$ is Haar-random.} By Theorem~\ref{thm:ru-solve-pspace}, the constructed oracle $\mathcal{O}_U$ correctly solves $\mathsf{TQBF}$ instances with probability at least $2/3$. By Theorem~\ref{thm:tqbf-verification}, the verifier $\mathcal{V}^{\mathcal{O}_U}$ outputs $\mathsf{True}$ with probability at least $2/3$. Therefore, $\mathcal{V}_{\text{unitary}}^{U}$ outputs $\mathsf{HaarRandom}$ with probability at least $2/3$.

\paragraph{Case 2: $U$ is generated by polynomial-size quantum circuits.} When $U$ is pseudorandom (i.e., generated by a polynomial-size quantum circuits), the construction in Theorem~\ref{thm:ru-solve-pspace} yields an oracle $\mathcal{O}_U$ that is also implementable by polynomial-size quantum circuits. By Theorem~\ref{thm:tqbf-verification}, for sufficiently large system sizes, the verifier $\mathcal{V}^{\mathcal{O}_U}$ outputs $\mathsf{Pseudo}$ with probability at least $2/3$. Therefore, $\mathcal{V}_{\text{unitary}}^{U}$ outputs $\mathsf{Pseudorandom}$ with probability at least $2/3$.

\paragraph{Runtime:} The construction of $\mathcal{O}_U$ from $U$ takes polynomial time by Lemma~\ref{thm:ru-solve-pspace}, and running the $\mathsf{TQBF}$ verifier takes polynomial time by Theorem~\ref{thm:tqbf-verification}. Therefore, the total runtime of $\mathcal{V}_{\text{unitary}}$ is polynomial in the system size.
\end{proof}

\section{Undecidability of the existence of energy-conserving PRUs}
\label{ap:undecidable}

In this section, we prove that determining if a given local Hamiltonian has energy-conserving PRU is an undecidable problem. The reason is to define uniform family of Hamiltonians, one need to specify a TM that generates the description of that Hamiltonian family. That is, the problem in fact takes TMs as inputs. As a result, we can construct a Hamiltonian family, who has or does not have energy-conserving PRU depending on the solution of halting problem for a given input.

\begin{lemma}
[Simulating two machines simutaneously]
    \label{lem:one-simulate-two}
    Let $\tm_1=\langle Q_1,\Gamma_1,\Delta_1\rangle$ and $\tm_2=\langle Q_2,\Gamma_2,\Delta_2\rangle$ be two Turing machines. When taking $x_1,x_2\in\{0,1\}^*$ as inputs, respectively, the maximum memory cost in the first $T$ steps is $S(T,x_1,x_2)$. Then there is a simulating Turing machine $\tm$ that can simulate the first $T$ steps of operations of both $\tm_1$ and $\tm_2$ simultaneously, with simulation time $O(T\cdot S(T,x_1,x_2))$ for any $x_1,x_2$ and $T$.
\end{lemma}
\begin{proof}
    $\tm$ can be constructed in the following sense. Define the new alphabet as $\Gamma=Q_1\cup Q_2\cup (\Gamma_1\times \Gamma_2)$. That is, we introduce two tapes to write the symbols of $\tm_1$ and $\tm_2$ simultaneously, and put the head internal state into the tape. A symbol $q\in Q_1$ at cell $i$ in the tape represents the head of $\tm_1$ is located at cell $i+1$ and has internal state $q$, and so as any $q\in Q_2$. Then to simulate one step of $\tm_1$ and $\tm_2$, the head of $\tm$ sweeps along the length $S(T,x_1,x_2)$ region of tape and update the symbol according to the transition rule of $\tm_1$ ($\tm_2$) when encountering any $q\in Q_1$ ($Q_2$). The time overhead of each step is at most $O(S(T,x_1,x_2))$. So the total running time to simulate until the $T$-th step is at most $O(T\cdot S(T,x_1,x_2))$
\end{proof}
\begin{lemma}
[PRU and halting problem]
    \label{lem:halting-tm}
    Let $\utm$ denotes the universal Turing machine. $\{H_n\}$ is the uniform family of Hamiltonians whose energy-conserving PRU does not exist.
    Then for any $x\in\{0,1\}^*$, there exists a Turing machine $\tm_x$ and a constant $\alpha>0$, such that for any $n\in\mathbb{N}^+$ as input,  if $\utm$ does not halt in time $\leq \alpha n$ upon input $x$, $\tm_x$ generates $H_n$, otherwise generates $0$. The overall running time of $\tm_x$ is $\poly{n}$. The whole process takes time at most $\poly{n}$.
\end{lemma}
\begin{proof}
    Let $\tm_d$ be the Turing machine that generates $\{H_n\}$ upon any input size $n$. Since $H_n$ is a local Hamiltonian, the running time of $\tm_d$ is $\alpha n$ for some $\alpha>0$. We design $\tm_d$ to be reversible. Let $\tm_d^{\text{(rev)}}$ be the reverse Turing machine of $\tm_d$, and $\utm$ be the universal Turing machine. Then $\tm_x$ can be constructed as follows. For each $n$, simulate the operations of $\tm_d$ upon input $n$ and $\utm$ upon input $x$ simultaneously using Lemma.\ref{lem:one-simulate-two}. If $\utm$ halts before $\tm_d$ halts (which takes $\alpha n$ steps for $\tm_d$), switch to $\tm_d^{\text{(rev)}}$ to reverse all the operations of $\tm_d$.
\end{proof}

Lemma.\ref{lem:halting-tm} asserts that for any $x$, there is a uniform family of Hamiltonians (generated by $\tm_x$), such that whether it has energy-conserving PRU depends on if the universal Turing machine halts upon input $x$. As a direct consequence, we have 
\begin{theorem}[Undecidability of the existence of energy-conserving PRU]
    \label{thm:undecidable-deciding-pru}
    If we have an algorithm that takes any uniform family of local Hamiltonians as input (more precisely, we should take the Turing machine that generates this Hamiltonian family as input), we can solve the halting problem.
\end{theorem}
\begin{proof}
    Lemma~\ref{lem:halting-tm} says that for any $x$, we can design a uniform family of Hamiltonians $\{H_n\}$ efficiently generated by $\tm_x$, such that if $\utm$ halts upon input $x$, $H_n$ is $0$ for sufficiently large $n$, whose energy-conserving PRU is the conventional PRU. If $\utm$ does not halt upon $x$, $\{H_n\}$ does not have energy-conserving PRU. So deciding the existence of PRU requires solving the halting problem for all instances.
\end{proof}

\section{One-dimensional hopping problem}\label{hopping-problem}

In this section we exactly solve the one-dimensional hopping problem with boundary potentials.
Let the Hamiltonian be
\begin{align}
    H = \sum_{i=1}^{N-1}\ket{i}\bra{i+1}+\ket{i+1}\bra{i}+V_1\ket{1}\bra{1}+V_2\ket{N}\bra{N}.
\end{align}
Choose ansatz wavefunction
\begin{align}
    \ket{E_k} = \sum_{n=1}^N\psi_n\ket{n},\quad \psi_n=A_ke^{ikn}+B_ke^{-ikn}.
\end{align}
The corresponding energy in bulk is $E_k=2\cos k$. Now consider the boundary terms. The boundary conditions are
\begin{align}
    \psi_2+V_1\psi_1=E\psi_1,\quad \psi_{N-1}+V_2\psi_N=E\psi_N.
\end{align}
This gives us
\begin{align}
    A_k(V_1e^{ik}-1)+B_k(V_1e^{-ik}-1)=0,\quad
    A_k(V_2e^{iNk}-e^{i(N+1)k})+B_k(V_2e^{-iNk}-e^{-i(N+1)k})=0.
\end{align}
To bring solutions to exist, 
\begin{align}
    \frac{V_1e^{ik}-1}{V_1e^{-ik}-1}=\frac{V_2e^{iNk}-e^{i(N+1)k}}{V_2e^{-iNk}-e^{-i(N+1)k}}
\end{align}
Equivalently,
\begin{align}\label{ap:eq:solution}
    \tan (Nk)=-\frac{(1-V_1V_2)\sin k}{(V_1V_2+1)\cos k-(V_1+V_2)}.
\end{align}
One obvious solution for the above equation is $k=0$. But this solution is not consistent with boundary conditions, except for $V_1=V_2=1$. Similarly, $k=\pi$ is a solution only when $V_1=V_2=-1$. In the main text we deal with the case of $V_1=-1$. In this case Eq.~\eqref{ap:eq:solution} becomes
\begin{align}
    \frac{1+V_2}{V_2-1}\tan\Big(\frac{k}{2}\Big)=\tan (Nk).
\end{align}
For $V_2\in(-1,1)$, the above equation has $N$ real solutions of $k\in(0,\pi)$. Thus, the system has no bound state. For wavefunction, using
\begin{align}
    \frac{A_k}{B_k} = -\frac{e^{-ik}+1}{e^{ik}+1}=-e^{-i k}.
\end{align}
This leads to
\begin{align}
    \ket{\psi_k} = \frac{1}{\sqrt{C_k}}\sum_{n=1}^N (e^{-ikn}-e^{ikn-ik})\ket{n},
\end{align}
where the normalization factor is determined by
\begin{align}
    C_k = \sum_{n=1}^N|e^{-ikn}-e^{ikn-ik}|^2 =
     \sum_{n=1}^N 2-e^{i(2n-1)k}-e^{-i(2n-1)k} = 2N-\frac{\sin 2Nk}{\sin k}
\end{align}

\section{Proof of Lemma~\ref{lem:solution-eigenvalue}}\label{ap:proof}
\newcounter{savedlemma} 
\setcounter{savedlemma}{\value{lemma}} 
\setcounter{lemma}{\numexpr\getrefnumber{lem:solution-eigenvalue}-1\relax}
\begin{lemma}
    For the following equation
    \begin{align}\label{ap:eq:solution2}
        -\frac{1}{2^{m+1}-1}\tan\Big(\frac{k}{2}\Big)=\tan(Nk),\quad k\in(0,\pi),\quad
        m,N\in\mathbb{N}^+,
    \end{align}
    denote $\CS(m,N)$ as the solution set for $k$ with fixed $m$ and $N$. Define $\CS(m)=\cup_{N\in\mathbb{N}^+}\CS(m,N)$, then it satisfies the following:
    \begin{itemize}
        \item For $m_1\neq m_2$, $\CS(m_1)\bigcap\CS(m_2)=\CS(m_1)\bigcap\pi\mathbb{Q}=\CS(m_2)\bigcap\pi\mathbb{Q}=\emptyset$.
        \item For any $m$ and $N_1\neq N_2$, $\CS(m,N_1)\bigcap \CS(m,N_2)=\emptyset$ .
    \end{itemize}
\end{lemma}
\setcounter{lemma}{\value{savedlemma}}

The proof idea is as follows.
Defining $x=\tan(k/2)$, $\tan(Nk)$ can be written as the ration of two polynomials with integer coefficients, by iteratively using the formulas for double angles. Therefore,
$x$ is the root of some polynomial with degree depending on $N$ and integer coefficients depending on $n$. 
It is sufficient to prove all of these polynomials are irreducible for different $n$ and $N$. In this case, they have no common roots.
\begin{definition}
    [Irreducible polynomials]
    \label{ap:def:irreducible}
    A polynomial $f(x)$ over $\mathbb{Z}$ (which means all the coefficients of $f$ belong to $\mathbb{Z}$) is an irreducible polynomial if it is non-constant and cannot be written as $f(x)=g(x)h(x)$ for $g$ and $h$ are both non-constant polynomails over $\mathbb{Z}$.
\end{definition}
\begin{fact}
    [Irreducible polynomials have no common roots]
    Let $f(x)$ and $g(x)$ be two non-constant polynomials over $\mathbb{Z}$, and $f(x)\neq \alpha g(x)$ for all $\alpha \in \mathbb{Z}$. Then $f$ and $g$ have no common roots.
\end{fact}\noindent
The proof idea for this fact is that both $f$ and $g$ must divide the minimal polynomial of the common root, contradicting with irreducibility. For a comprehensive overview for polynomials, one can refer to standard textbooks for number theory or Galois theory, e.g.,~\cite{dummit2004abstract,lang2012algebra,stewart2022galois}.

To prove the irreducibility of these polynomials, we use (1) the fundamental theorem for Galois theory; (2)
the properties of Cyclotomic polynomials. We state these standard results in a superficial and non-rigorous manner for accessibility.

\begin{fact}
    [Fundamental theorem of Galois theory]\label{ap:fact:galois}
    Let $f(x)$ be a polynomial over $\mathbb{Q}$.
    Let $K$ be the splitting field, i.e., smallest field that contains $\mathbb{Q}$ and all the roots of $f$. Define the Galois group $\operatorname{Gal(K/\mathbb{Q})}$ as the group of all symmetries over $K$ that leaves $\mathbb{Q}$ invariant (e.g., permutations of irrational numbers in $K$). Then
    \begin{itemize}
        \item Every subgroup $H$ of $\operatorname{Gal(K/\mathbb{Q})}$ corresponds to an intermediate field $E:\mathbb{Q}\subseteq E\subseteq K$ left invariant by the action of $H$, and vice versa.
        \item The correspondence is inclusion-reversing, i.e., the larger the subgroup $H$, the smaller the field $E$. In particular, $\operatorname{Gal(K/\mathbb{Q})}$ corresponds to $\mathbb{Q}$.
    \end{itemize}
\end{fact}
    
\begin{definition}
    [Cyclotomic polynomials]\label{ap:def:cyclotomic}
    The $n$-th order cyclotomic polynomial refers to a unique irreducible monic polynomial (the coefficient for highest order monomial is $1$) $\Phi_n(x)$ over $\mathbb{Z}$, such that its roots are all $n$-th primitive roots of unit $e^{i2\pi\frac{k}{n}}$. In other words,
    \begin{align}
        \Phi_n(x)=\prod_{1\leq k\leq n,\gcd(n,k)=1}\big(x-e^{i2\pi\frac{k}{n}}\big).
    \end{align}
\end{definition}
Examples of cyclotomic polynomials for $n=1,2,3,4,5$ are
\begin{align}
    \Phi_1(x)&=x-1,\notag\\
    \Phi_2(x)&=x+1,\notag\\
    \Phi_3(x)&=x^2+x+1,\notag\\
    \Phi_4(x)&=x^2+1,\notag\\
    \Phi_5(x)&=x^4+x^3+x^2+x+1.
\end{align}

We these preparations, we prove the following two technical lemmas.
\begin{lemma}
    [Kronecker's lemma]\label{ap:lem:kronecker}
    Let $f(x)$ be a monic polynomial over $\mathbb{Z}$. If all roots of $f$ have absolute values at most 1, then $f$ is the product of cyclotomic polynomials and powers of $x$.
\end{lemma}
\begin{proof}
    Let $r$ be the multiplicity of root $0$. Denote $f(x)=x^rg(x)$. $g(x)$ is a monic polynomial that has only nonzero roots. Denote $\Lambda=\{\lambda_1,\cdots \lambda_n\}$ as the multiset of roots of $g(x)$. All $\lambda_i$ are algebraic integers (roots of some monic polynomial over $\mathbb{Z}$).

    Then define 
    \begin{align}
        g_k(x)=\prod_{i=1}^n\big(x-\lambda_i^k\big),\quad k\in\mathbb{N}^+.
    \end{align}
    Every coefficients of $g_k(x)$ are algebraic integers. i.e., sum and multiplicity of numbers that are the roots of monic polynomials over $\mathbb{Z}$. Furthermore, every coefficients are symmetric sum of $\lambda$s, which means that the are left invariant under any permutations of $\lambda$s.
    Thus, they are rational numbers due to Fact~\ref{ap:fact:galois}. These together assert that coefficients of $g_k(x)$ must be integers.
    
    Since $|\lambda_i|\leq 1$, the norm of coefficient of $x^m$ is at most $\binom{n}{m}$. As a result, for all $k\in\mathbb{N}^+$, possible forms of $g_k(x)$ are finite. Hence, there must exist integers $k_1\neq k_2$, such that $g_{k_1}(x)=g_{k_2}(x)$. Therefore, multisets $\{\lambda_1^{k_1},\cdots \lambda_n^{k_1}\}$ and $\{\lambda_1^{k_2},\cdots \lambda_n^{k_2}\}$ coincide. That is, there exist a permutation $\sigma$ of $\{1,2,\cdots,n\}$, such that
    \begin{align}
        \lambda^{k_1}_i=\lambda^{k_2}_{\sigma(i)},\quad 1\leq i\leq n.
    \end{align}
    Iterating along a cycle of length $l$ gives 
    \begin{align}
        \lambda_i^{k_1^l}=\lambda_i^{k_2^l}.
    \end{align}
    Since $\lambda_i\neq 0$, this gives $\lambda_i=e^{i2\pi \frac{p}{k_1^l-k_2^l}}$ for $p\in\mathbb{N}^+$. Thus, every $\lambda_i$ is the root of a cyclotomic polynomial. This completes the proof
\end{proof}

\begin{lemma}
    \label{ap:lem:no-rational-solution}
    Let $n\in\mathbb{N}^+$. For any $N\in\mathbb{N}^+$, the solutions of following equation
    \begin{align}\label{ap:eq:solution-2}
        -\frac{1}{2^{n+1}-1}\tan\Big(\frac{k}{2}\Big)=\tan(Nk),\quad
        k\in(0,\pi),
    \end{align}
    have no intersection with $\pi\mathbb{Q}$.
\end{lemma}
\begin{proof}
    Denote $x = e^{ik/2}$. We have
    \begin{align}
        \tan\Big(\frac{k}{2}\Big)=\frac{x-1}{i(x+1)},\quad
        \tan (Nk)=\frac{x^{2N}-1}{i(x^{2N}+1)}.
    \end{align}
    Using this expression,
    solutions of Eq.~\eqref{ap:eq:solution-2} correspond to roots of
    \begin{align}
        g_N(x)=2^n x^{2N}+(2^{n+1}-1)(x^{2N-1}+x^{2N-2}+\cdots+x^2+x)+2^n.
    \end{align}
    Note that roots of $g_N(x)$ are in $x,x^*$ pairs. So we can only consider the roots where $\operatorname{Im}x>0$, which correspond to $k\in(0,\pi)$.
     Assume for $N=N_0$ there is solution $y=e^{i2\pi p/q}$ for contradiction, where $p,q\in\mathbb{N}^+$, $p<q/2$ and $\gcd(p,q)=1$. Then $y$ is the root of cyclotomic polynomial $\Phi_q(x)$. As a result, $\Phi_q(x)\mid g_{N_0}(x) $. Since $e^{i2\pi/q}$ must also be a root of $\Phi_q(x)$, is is also a root of $g_{N_0}(x)$, i.e.,
    \begin{align}\label{ap:eq:solution-rational}
        -\frac{1}{2^{n+1}-1}\tan\Big(\frac{2\pi}{q}\Big)=\tan
        \Big(\frac{4\pi N}{q}\Big).
    \end{align}
    Let 
    \begin{align}
        \delta:=\min_{s\in\mathbb{Z}}\Big|\frac{2N}{q}-s\Big|\in[0,\frac{1}{2}].
    \end{align}
    Since $\tan$ is periodic with $\pi$, $|\tan(4\pi N/q)|=|\tan(2\pi\delta)|$. Clearly, $\delta\neq 0,1/2$. In this case, $\delta\geq 1/q$. As a result,
    \begin{align}
        \Big|\tan
        \Big(\frac{4\pi N}{q}\Big)\Big|=|\tan(2\pi\delta)|\geq 
        \Big|\tan\Big(\frac{2\pi}{q}\Big)\Big|>\frac{1}{2^{n+1}-1}
        \Big|\tan\Big(\frac{2\pi}{q}\Big)\Big|,
    \end{align}
    contradicting Eq.~\eqref{ap:eq:solution-rational}.
\end{proof}

%Now we prove Lemma~\ref{lem:solution-eigenvalue}.

Now we prove the desired lemma.
\begin{proof}[Proof of Lemma~\ref{lem:solution-eigenvalue}]

    $\CS(m_1)\cap\pi\mathbb{Q}=\CS(m_2)\cap\pi\mathbb{Q}=\emptyset$ is a direct consequence of Lemma~\ref{ap:lem:no-rational-solution}. Denote $x=e^{ik/2}$, then solutions of Eq.~\eqref{ap:eq:solution2} are roots of 
    \begin{align}
        g_{m,N}(x)=2^m x^{2N}+(2^{m+1}-1)(x^{2N-1}+x^{2N-2}+\cdots+x^2+x)+2^m.
    \end{align}

    Assume for $m_0,N_0\in\mathbb{N}^+$, $g_{m_0,N_0}(x)$ is reducible, i.e., there exists a polynomial $\phi(x)=a_0x^k+a_1x^{k-1}+\cdots$ over $\mathbb{Z}$ such that  $\phi(x)\mid g_{m_0,N_0}(x)$. If $a_0=2^{s}$ for $1\leq s\leq m-1$, the coefficient for the highest order monomial in $g_{m_0,N_0}(x)/\phi(x)$ is $2^{m-s}$. As a result,
     the coefficient of $x^{2N-1}$ in $g_{m_0,N_0}$ is even, which is a contradiction. If $a_0=1$,  since the roots of $\phi(x)$ are also roots of $g_{m_0,N_0}(x)$, $\phi(x)$ is a cyclotomic polynomial by Lemma~\ref{ap:lem:kronecker}, contraditing with Lemma~\ref{ap:lem:no-rational-solution}. If $a_0=2^s$, same argument applies to $g_{m_0,N_0}(x)/\phi(x)$.
     As a result, for all $m,N\in\mathbb{N}^+$, $g_{m,N}(x)$ is a irreducible polynomial. Moreover, every two such polynomials are not proportional to each other, thus
     they  have no common roots. This gives $\CS(m_1)\cap\CS(m_2)=\emptyset$ for $m_1\neq m_2$, and $\CS(m,N_1)\cap \CS(m,N_2)=\emptyset$ for  $N_1\neq N_2$.
    %Now we prove the second statement. Assume that there exists $k_0\in(0,\pi)$ and $N_1\neq N_2\in\mathbb{N}^+$ such that $k_0\in\CS(m,N_1)\cap\CS(m,N_2)$. Then according to Eq~\ref{ap:eq:solution2}, $\tan(N_1k_0)=\tan(N_2k_0)$. So $k_0=N_3\pi/(N_1-N_2)$ for $N_3\in\mathbb{N}^+$, contradicting with Lemma~\ref{ap:lem:no-rational-solution}.
\end{proof}
\setcounter{lemma}{\value{savedlemma}}

\clearpage
\bibliography{refs}

\end{document}